\documentclass[sigconf,nonacm]{acmart}
\AtBeginDocument{%
  }

\copyrightyear{2025}
\acmYear{2025}
\setcopyright{none}
\acmConference[CCS '25]{Proceedings of the 2025 ACM SIGSAC Conference on
Computer and Communications Security}{October 13--17, 2025}{Taipei, Taiwan}
\acmBooktitle{Proceedings of the 2025 ACM SIGSAC Conference on Computer
and Communications Security (CCS '25), October 13--17, 2025, Taipei,
Taiwan}\acmDOI{10.1145/3719027.3765077}
\acmISBN{979-8-4007-1525-9/2025/10}

\setcitestyle{numbers,sort&compress}

\usepackage{graphicx}
\usepackage{subcaption}
\usepackage{amsmath}
\usepackage{amsfonts} 
\usepackage{listings}
\usepackage{mathtools}
\usepackage{multirow}

\usepackage{tikz}

\usepackage{pgfplots}
\pgfplotsset{compat=1.18}
\usepgfplotslibrary[colorbrewer]
\usetikzlibrary{pgfplots.colorbrewer}

\definecolor{color1}{HTML}{377eb8}
\definecolor{color2}{HTML}{4daf4a}
\definecolor{color3}{HTML}{984ea3}
\definecolor{color4}{HTML}{ff7f00}
\definecolor{color5}{HTML}{a65628}
\definecolor{color6}{HTML}{f781bf}

\pgfplotsset{
filter discard warning=false,
cycle multiindex* list={
mark=square,mark=x,mark=o,mark=+,mark=diamond\nextlist
color1,color2,color3,color4,color5,color6\nextlist
solid,dashed,dotted,dashdotted,loosely dashed,dashdotdotted,densely dotted\nextlist
},
/pgfplots/xmajorgrids=true,
/pgfplots/ymajorgrids=true,
}
\pgfplotsset{
  every axis/.append style={legend cell align=left, line width=0.5pt},
  every axis plot/.append style={line width=1pt}
}
\pgfplotsset{
    node near coord/.style={
        nodes near coords*={
            \ifnum\coordindex=#1\pgfmathprintnumber{\pgfplotspointmeta}\fi
        }
    },
    nodes near some coords/.style={ 
        scatter/@pre marker code/.code={},
        scatter/@post marker code/.code={}, 
        node near coord/.list={#1}
    }
}

\graphicspath{{ ../eps/}}

\hypersetup{
  breaklinks   = true,
  colorlinks   = true,    
  urlcolor     = blue,    
  linkcolor    = blue,    
  citecolor    = red      
}

\newcommand{\eqreft}[1]{Eq.~\ref{#1}}

\newtheorem{definition}{Definition}
\newtheorem{theorem}{Theorem}
\newtheorem{lemma}{Lemma}
\newtheorem{corollary}{Corollary}
\newtheorem{proposition}{Proposition}

\newcommand{\paragraphEmph}[1]{\paragraph{\textbf{#1}}}

\newcommand{\integers}{\ensuremath{\mathbb{Z}}}
\newcommand{\realNumbers}{\ensuremath{\mathbb{R}}}
\newcommand{\positiveRealNumbers}{\ensuremath{\mathbb{R}_{> 0}}}
\newcommand{\nonNegativeRealNumbers}{\ensuremath{\mathbb{R}_{\geq 0}}}

\newcommand{\validator}{\ensuremath{v}}
\newcommand{\allValidators}{\ensuremath{V}}
\newcommand{\allValidatorsAt}[1]{\ensuremath{\allValidators_{#1}}}
\newcommand{\attackingValidators}{\ensuremath{{\allValidators_{\allAttackStakes}}}}

\newcommand{\service}{\ensuremath{s}}
\newcommand{\allServices}{\ensuremath{S}}
\newcommand{\allServicesAt}[1]{\ensuremath{\allServices_{#1}}}
\newcommand{\attackedServices}{\ensuremath{{\allServices_\allAttackStakes}}}
\newcommand{\attackedServicesAt}[1]{\ensuremath{\allServices_{\allAttackStakes_{#1}}}}

\newcommand{\baseServices}{\ensuremath{S_{\textit{base}}}}

\newcommand{\allStakes}{\ensuremath{\sigma}}
\newcommand{\allStakesAt}[1]{\ensuremath{\allStakes_{#1}}}
\newcommand{\stake}[1]{\ensuremath{\allStakes\!\left({#1}\right)}}
\newcommand{\symmetricStake}{\ensuremath{\allStakes}}
\newcommand{\symmetricStakeAt}[1]{\ensuremath{\allStakesAt{#1}}}

\newcommand{\stakeAt}[2]{\ensuremath{\allStakesAt{#1}\!\left({#2}\right)}}

\newcommand{\allAttackThresholds}{\ensuremath{\theta}}
\newcommand{\allAttackThresholdsAt}[1]{\ensuremath{\allAttackThresholds_{#1}}}
\newcommand{\attackThreshold}[1]{\ensuremath{\allAttackThresholds\!\left({#1}\right)}}

\newcommand{\symmetricAttackThreshold}{\ensuremath{\allAttackThresholds}}

\newcommand{\allAttackPrizes}{\ensuremath{\pi}}
\newcommand{\allAttackPrizesAt}[1]{\ensuremath{\allAttackPrizes_{#1}}}
\newcommand{\attackPrize}[1]{\ensuremath{\allAttackPrizes\!\left({#1}\right)}}
\newcommand{\attackPrizeAt}[2]{\ensuremath{\allAttackPrizesAt{#1}\!\left({#2}\right)}}

\newcommand{\totalAttackPrize}[1]{\ensuremath{\Pi_{\networkState}\left({#1}\right)}}
\newcommand{\totalAttackPrizeAt}[2]{\ensuremath{\Pi_{\networkStateAt{#1}}\left({#2}\right)}}

\newcommand{\allAllocations}{\ensuremath{w}}
\newcommand{\allAllocationsAt}[1]{\ensuremath{\allAllocations_{#1}}}
\newcommand{\allocation}[2]{\ensuremath{\allAllocations\!\left({#1},{#2}\right)}}
\newcommand{\allocationAt}[3]{\ensuremath{\allAllocationsAt{#1}\!\left({#2},{#3}\right)}}
\newcommand{\symmetricAllocation}[1]{\ensuremath{\allAllocations\!\left({#1}\right)}}
\newcommand{\symmetricAllocationAt}[2]{\ensuremath{\allAllocationsAt{#1}\!\left({#2}\right)}}

\newcommand{\allAttackStakes}{\ensuremath{\alpha}}
\newcommand{\allAttackStakesAt}[1]{\ensuremath{\allAttackStakes_{#1}}}
\newcommand{\attackStake}[2]{\ensuremath{\allAttackStakes\!\left({#1},{#2}\right)}}
\newcommand{\attackStakeAt}[3]{\ensuremath{\allAttackStakesAt{#1}\!\left({#2},{#3}\right)}}

\newcommand{\validatorAttackCost}[2]{\ensuremath{c_{\networkState}\!\left({#1}, {#2}\right)}}
\newcommand{\validatorAttackCostAt}[3]{\ensuremath{c_{\networkStateAt{#1}}\!\left({#2}, {#3}\right)}}

\newcommand{\attackCost}[1]{\ensuremath{C_{\networkState}\!\left({#1}\right)}}
\newcommand{\attackCostAt}[2]{\ensuremath{C_{\networkStateAt{#1}}\!\left({#2}\right)}}

\newcommand{\maxByzantineServices}{\ensuremath{{f}}}

\newcommand{\networkState}{\ensuremath{G}}

\newcommand{\networkStateAt}[1]{\ensuremath{\networkState_{#1}}}
\newcommand{\networkAdvance}{\ensuremath{\searrow}}

\newcommand{\byzantineServices}{\ensuremath{\allServices^B}}
\newcommand{\byzantineSubsets}[1]{\ensuremath{\mathbb{B}_{\networkState}\!\left({#1}\right)}}

\newcommand{\validatorUtilitySecurityGame}[2]{\ensuremath{u_{#1}\!\left({#2}\right)}}
\newcommand{\validatorPrizeShareSecurityGame}[2]{\ensuremath{\gamma_{\networkState}\!\left({#1}, {#2}\right)}}

\newcommand{\strategyProfileSecurityGame}{\ensuremath{\sigma}}

\newcommand{\allStrategySpacesSecurityGame}{\ensuremath{\Sigma_{\networkState}}}

\newcommand{\adversaryBudget}{\ensuremath{\beta}}

\newcommand{\restakingDegree}[1]{\ensuremath{\text{deg}_{\networkState}\!\left({#1}\right)}}

\newcommand{\restakingDegreeSymmetric}{\ensuremath{\deg_{\networkState}}}
\newcommand{\targetRestakingDegree}{\ensuremath{\text{d}^\ast}}

\newcommand{\serviceReward}[1]{\ensuremath{\allServiceRewards(#1)}}
\newcommand{\allServiceRewards}{\ensuremath{R}}
\newcommand{\validatorReward}[2]{\ensuremath{r\!\left({#1}, {#2}\right)}}
\newcommand{\validatorUtility}[2]{\ensuremath{u_{#1}\!\left({#2}\right)}}

\newcommand{\allNashAllocations}{\ensuremath{\allAllocations^\ast}}
\newcommand{\nashAllocation}[2]{\ensuremath{\allNashAllocations\!\left({#1}, {#2}\right)}}

\newcommand{\allocationStrategy}[1]{\ensuremath{\omega_{#1}}}
\newcommand{\allocationStrategyPrime}[1]{\ensuremath{\omega'_{#1}}}
\newcommand{\utilityFunction}{\ensuremath{U}}

\newcommand{\nashConstant}[1]{\ensuremath{c_{#1}}}

\newcommand\restr[2]{{
  \left.\kern-\nulldelimiterspace 
  #1 
  \littletaller 
  \right|_{#2} 
  }}

\newcommand{\littletaller}{\mathchoice{\vphantom{\bigr|}}{}{}{}}

\newcommand{\floor}[1]{\ensuremath{\left\lfloor{#1}\right\rfloor}}

\newcommand{\sspElement}[1]{\ensuremath{b_{#1}}}
\newcommand{\sspElementCount}{\ensuremath{n}}
\newcommand{\sspElementSum}{\ensuremath{B}}
\newcommand{\sspTarget}{\ensuremath{T}}

\newcommand{\milpVariableLetter}{\ensuremath{x}}

\newcommand{\milpLargeNumberFeasibility}{\ensuremath{M_1}}
\newcommand{\milpLargeNumberCost}{\ensuremath{M_2}}

\newcommand{\milpAttackStakeVariable}[2]{\ensuremath{\milpVariableLetter_{{#1}, {#2}}^{\allAttackStakes}}}
\newcommand{\milpValidatorCostVariable}[1]{\ensuremath{\milpVariableLetter_{#1}^{c}}}
\newcommand{\milpValidatorCostAuxiliaryVariable}[1]{\ensuremath{\milpVariableLetter_{#1}^{c, \text{aux}}}}
\newcommand{\milpAttackedServiceVariable}[1]{\ensuremath{\milpVariableLetter_{#1}^{\allServices}}}
\newcommand{\milpVariables}{\ensuremath{\vec{\milpVariableLetter}}}

\newcommand{\milpLargeNumberRemainingStake}{\ensuremath{M_3}}
\newcommand{\milpLargeNumberRemainingAllocation}{\ensuremath{M_4}}
\newcommand{\milpLargeNumberByzantineServices}{\ensuremath{M_5}}

\newcommand{\milpByzantineServiceVariable}[1]{\ensuremath{\milpVariableLetter_{#1}^{\allServices, \text{byz}}}}
\newcommand{\milpByzantineServiceAuxiliaryVariable}[0]{\ensuremath{\milpVariableLetter^{\allServices, \text{aux}}}}
\newcommand{\milpRemainingStakeVariable}[1]{\ensuremath{\milpVariableLetter_{#1}^{\allStakes}}}
\newcommand{\milpRemainingStakeAuxiliaryVariable}[1]{\ensuremath{\milpVariableLetter_{#1}^{\allStakes, \text{aux}}}}
\newcommand{\milpRemainingAllocationVariable}[2]{\ensuremath{\milpVariableLetter_{{#1}, {#2}}^{\allAllocations}}}
\newcommand{\milpRemainingAllocationAuxiliaryVariable}[2]{\ensuremath{\milpVariableLetter_{{#1}, {#2}}^{\allAllocations, \text{aux}}}}

\begin{document}

\title{Elastic Restaking Networks}
\subtitle{United we fall, (partially) divided we stand}

\author{Roi Bar-Zur}
\affiliation{%
  \institution{Technion}
  \city{}
  \state{}
  \country{Israel}
}

\author{Ittay Eyal}
\affiliation{%
  \institution{Technion}
  \city{}
  \state{}
  \country{Israel}
}

\renewcommand{\shortauthors}{Roi Bar-Zur and Ittay Eyal}

\begin{abstract}
Many blockchain-based decentralized \emph{services} require their \emph{validators} (operators) to deposit \emph{stake} (collateral), which is forfeited (slashed) if they misbehave.
\emph{Restaking networks} let validators secure multiple services by reusing stake.
These networks have quickly gained traction, leveraging over~\$20 billion in stake.
However, restaking introduces a new attack vector where validators can coordinate to misbehave across multiple services simultaneously, extracting digital assets while forfeiting their stake only once.

Previous work focused either on preventing coordinated misbehavior or on protecting services if all other services are \emph{Byzantine} and might unjustly cause slashing due to bugs or malice.
The first model overlooks how a single Byzantine service can collapse the network, while the second ignores shared-stake benefits.

To bridge the gap, we analyze the system as a strategic game of coordinated misbehavior, when a given fraction of the services are Byzantine.
We introduce \emph{elastic} restaking networks, where validators can allocate portions of their stake that may cumulatively exceed their total stake, and when allocations are lost, the remaining stake stretches to cover remaining allocations.
We show that elastic networks exhibit superior robustness compared to previous approaches, and demonstrate a synergistic effect where an elastic restaking network enhances its blockchain's security, contrary to community concerns of an opposite effect in existing networks.
We then design incentives for tuning validators' allocations.
  
Our elastic restaking system and incentive design have immediate practical implications for deployed restaking networks.
\end{abstract}



\maketitle

\section{Introduction}
\label{section:introduction}


Blockchains are distributed-computing protocols executed by a set of validators to facilitate digital-asset ownership. 
To secure the system in a decentralized fashion, without privileged entities, many blockchains~(e.g.,~\cite{buterin2014ethereum,yakovenko2018solana,rocket2018snowflake}) require validators to deposit \emph{stake} (collateral), which can be \emph{slashed} (lost)~\cite{buterin2014slasher} if they misbehave.
This approach, known as \emph{cryptoeconomic security}, is effective if the potential slashing is greater than any possible gains from misbehavior.

In addition to simple asset transfers, many blockchains support \emph{smart contracts}, which are stateful programs enabling automated interactions~\cite{buterin2014ethereum}.
To overcome their native limitations, many decentralized \emph{services} employ external validators alongside smart contracts.
Examples include \emph{rollups}~\cite{kalodner2018arbitrum,optimism}, which offload computations; \emph{bridges}~\cite{mccorry2021sok}, which transfer assets and data among blockchains; \emph{data availability solutions}~\cite{celestia,sheng2021aced}, which offload data storage; and \emph{oracle networks}~\cite{breidenbach2021chainlink,eskandari2021sok}, which import external data.
These services rely on cryptoeconomic security as well, requiring their external validators to deposit slashable stake.

To improve the efficiency of stake usage across the ecosystem, \emph{restaking networks} have emerged.
They allow validators to deposit stake and allocate it to multiple services, any of which can slash it.
A restaking network can either include the underlying blockchain's stake~\cite{eigenlayer2024restaking,symbiotic} or not~\cite{babylon2023bitcoinStaking}. 
There have been concerns about restaking risking the underlying blockchain's security~\cite{parasol2023ethereum,kubinec2024eigenlayer,mallesh2024eigenlayer,chitra2024much}, but nevertheless restaking has gained significant traction, with EigenLayer~\cite{restakingInfo} and other restaking networks~\cite{defillama2025restaking} collectively holding over \$20 billion in deposits.

While restaking networks make stake more accessible and allow validators to earn rewards from each service they validate, they introduce new security challenges.
When multiple services share the same stake, each additional service creates another opportunity for validators to extract value while risking their stake only once.
This gives rise to a strategic game where a coalition of validators can \emph{attack} by misbehaving in a subset of services.

Previous work~(\S\ref{sec:related}) took two distinct approaches. 
One focused on preventing coordinated misbehavior; following this approach implies over-allocation of stake is desirable, but that may leave the network vulnerable to even a single \emph{Byzantine} fault---a service that unjustly causes slashing due to bugs or malice. 
The second approach focused on protecting services if all other services are Byzantine; following this approach means not to use restaking, losing its robustness benefits. 

In this paper we present \emph{elastic restaking}~(\S\ref{section:model}), a restaking network architecture for handling both validator strategic behavior and Byzantine service faults.
In elastic restaking, validators deposit stake and allocate a portion to each service such that the sum of portions may be larger than their total stake.
Each service has an \emph{attack threshold}, the fraction of stake that must be used to attack it, and an \emph{attack prize}, the value that can be extracted from the service.

We analyze the system as a strategic \emph{cryptoeconomic security game} that proceeds as follows:
Each validator decides how much stake to use to attack each service, up to their allocated stake to that service.
Notably, validators can choose to use only a portion of their allocated stake, providing them with more granular attack strategies, a realistic but novel aspect of our model.
Each validator then loses the sum of the used portions, up to their entire stake.
If attacking validators dedicate enough stake to attack a service (above its threshold), they share the service's attack prize proportionally to the cost they paid.
Each validator's utility is their share of the prizes minus their lost stake.
We say the network is \emph{cryptoeconomically secure} if not using any stake to attack is a strong\footnote{We use a modified version of a strong Nash equilibrium where we require that there exists no coalition such that all its members non-strictly improve their utility by deviating (as opposed to the strict requirement of~\citet{aumann1959acceptable}).} Nash equilibrium~\cite{aumann1959acceptable}.

But even if cryptoeconomic security holds, the system might be brittle. 
We therefore extend this game by introducing another realistic but novel notion of restaking-network \emph{robustness}. 
First, we consider an adversary with a budget~$\adversaryBudget$ who uses it to subsidize validators to attack the network.
That is, the adversary supplements the total prize that attacking validators' can gain in the security game provided they attack at least one service.
We say the network is \emph{$\adversaryBudget$-cryptoeconomically robust} if not using any stake to attack is a strong Nash equilibrium in the resultant game.

We also consider the restaking network's robustness against \emph{Byzantine} services. 
Byzantine services can arbitrarily slash all stake allocated to them, reducing the total stake securing the network and potentially degrading its cryptoeconomic robustness.
In our model, the adversary first chooses some fraction of services to be Byzantine, and we then consider the \mbox{$\beta$-cryptoeconomic} robustness of the resulting network.

Unlike previous work that slashed an entire validator's stake, to support partial stake allocation we present \emph{elastic slashing}: when a validator's stake is slashed, the remaining stake is stretched to cover the rest of the validator's allocations. 
This makes elastic restaking networks strictly more expressive than previous models~(\S\ref{section:model:elastic_restaking_networks}).

Before addressing robustness, we analyze when networks are secure~(\S\ref{section:security_analysis}), meaning no coalition of validators will attack services.
Security holds when not attacking is a strong Nash equilibrium in the network's cryptoeconomic security game.
This equilibrium occurs precisely when there are no profitable attacks---those where the total prizes exceed the collective stake losses of the attacking validators.
To verify security, we develop sufficient conditions that generalize previous work~\cite{eigenlayer2024restaking}: a network is secure if (1)~each service has more stake allocated than it would need in isolation and (2)~for each validator, the sum of potential prize fractions across services is less than their stake.
While these conditions are useful, they only give us a partial picture.

We show that searching for profitable attacks in general restaking networks is NP-complete.
Hence, the complementary problem of checking security is co-NP-complete, and there is no efficient algorithm for it (unless~${\textit{P}=\textit{NP}}$).
We thus focus on symmetric networks, which are simpler to analyze yet rich enough to demonstrate the key mechanisms that govern restaking network robustness.
We develop an efficient algorithm to identify profitable attacks in symmetric networks.
We demonstrate our algorithm by calculating the minimum stake requirements for security in sample networks.
The implementation of our algorithm is available online~\anon{\cite{barzur2025code}}.


Next, we analyze robustness~(\S\ref{section:theoretical_robustness_analysis}) and follow a similar approach to our security analysis.
First, we present a simple yet non-efficiently computable condition for cryptoeconomic robustness:
A network is~$\adversaryBudget$-cryptoeconomically robust if there is no~\emph{$\adversaryBudget$-costly} attack, that is, there is no attack for which the total costs minus the total prizes is less than~$\adversaryBudget$.
We then extend our efficient algorithm to find profitable attacks in the symmetric case to find $\adversaryBudget$-costly attacks. 

We gain two significant insights by using our algorithm for several sample networks. 
First, elastic networks are in many cases more robust than existing restaking networks.
Second, we demonstrate a synergistic effect where a restaking network (like EigenLayer) can benefit the blockchain it is built on (Ethereum) by increasing its robustness: 
Consider a restaking network with a \emph{base} service (like Ethereum) to which all stake is allocated. 
Compare that with splitting the restaking network into two, a network without the base service and a (degenerate) restaking network with only the base service.
We find concrete cases where, using the same amount of stake overall, the combined restaking network is more robust compared to the two separate networks.

For asymmetric restaking networks, we resort to a computational approach using \emph{mixed-integer programming}~\cite{junger200950}~(\S\ref{section:mip_robustness_analysis}), as the heterogeneity of real restaking networks requires more general analysis methods. 
We solve the program with a state-of-the-art solver~\cite{hall2023highs} and validate our theoretical analysis for symmetric networks.
Furthermore, we illustrate similar effects to those of symmetric networks, suggesting that the mechanisms underlying these effects apply broadly beyond the symmetric settings we analyze.
However, the full complexity of asymmetric networks warrants further research.

We call the ratio between the sum of the validator's allocations to their stake its \emph{restaking degree}. 
Our analysis above shows that a certain restaking degree results in optimal robustness. 
The system designer should therefore encourage the validators to restake at this degree. 
We present the \emph{network formation game}~(\S\ref{section:incentives_for_a_target_restaking_degree}), in which services distribute rewards to their validators and validators choose their allocations to maximize their rewards.
We design a reward scheme that leads to a Nash equilibrium in which validators keep their restaking degree at a network-wide target value.

In conclusion~(\S\ref{section:conclusion}), our main contributions are: 
\begin{enumerate}
    \item presentation of elastic restaking networks, which are more expressive than atomic ones; 
    \item formalization of the security and robustness games; 
    \item proof that determining whether a network is secure is NP-complete; 
    \item efficient algorithms for security and robustness analysis in symmetric networks; 
    \item demonstration that elastic networks have superior robustness and may benefit their underlying blockchains; 
    \item robustness analysis in general networks using mixed-integer programming; and 
    \item a mechanism to incentivize a desired restaking degree.
\end{enumerate}

Our work raises further questions, e.g., on alternative slashing algorithms that maximize robustness, but is immediately applicable to improve the security of numerous deployed systems. 


\section{Related Work}\label{sec:related}


\paragraph{Restaking Networks}
\citet{eigenlayer2024restaking} introduced the first formal model for restaking networks, establishing sufficient conditions for cryptoeconomic security.
Their model requires validators to commit their entire stake to each service they validate, creating what we call \emph{atomic} restaking networks.
Their analysis focuses solely on coordinated misbehavior by validators, proving conditions under which no profitable attacks exist.
We build upon their security game framework but extend it in several crucial ways.
First, our elastic model allows validators to commit portions of their stake and potentially exceed their total stake across allocations.
We also consider allocation-divisible attacks where validators can use portions of their allocated stake, reflecting real-world services like Ethereum~\cite{buterin2014ethereum} where validators can be slashed for only a portion of their stake if that portion misbehaves.
Most importantly, we consider both network robustness and Byzantine services, two critical aspects absent from their initial model.

\citet{durvasula2024robust} expanded EigenLayer's analysis in two directions.
First, they examined cascading failures, showing how initial stake losses can trigger further attacks.
They show that any cascade of attacks following an initial stake loss is equivalent to a single attack, and that sufficient stake reserves can ensure the network is robust to such cascades.
Second, they studied how services might protect themselves by assuming all other services are Byzantine.
Our analysis differs from the analysis of \citeauthor{durvasula2024robust} in several ways.
(1) While we share their focus on robustness, our definitions of robustness differ.
In their model, some stake is first lost, and then the remaining stake is used to attack services;
in our model, stake is first used to attack services, and then an adversary reimburses the stake loss.
(2) Rather than considering only extremes (no services or all services being Byzantine), in this paper, we model scenarios where a weighted fraction of services are Byzantine, as is common in distributed-systems analysis.
(3) While they focus on analyzing the robustness of a given restaking network, we compare different structures to identify which are more robust.

\citet{chitra2024much} also analyze restaking networks and incentivizing allocation, but they do not address service faults and they make two additional assumptions: First, they assume coalition profits from an attack drop with the number of attacked services, whereas we consider the worse case without diminishing returns. 
Second, they assume honest validators can immediately  rebalance their remaining allocations after an attack; this is a strong assumption that neglects blockchain congestion and censorship attacks~\cite{mccorry2019smart, judmayer2021sok}, whereas our elastic restaking mechanism achieves this automatically. 
We note that unlike \citeauthor{chitra2024much} we neglect validator costs, since services often require validators to run only a single server, regardless of how much stake they have (even millions of dollars worth)~\cite{eigenda2025system,eoracle2024installation,witness2024node}. 


Community concerns~\cite{parasol2023ethereum,kubinec2024eigenlayer,mallesh2024eigenlayer,chitra2024much} that a single Byzantine service could compromise both EigenLayer and Ethereum, are perhaps what led EigenLayer to propose a significant revision~\cite{eigenlayer2024introducing}: 
Validators partition their stake among services without exceeding total stake.
In addition, they suggest services to consider both allocated and total validator stake for the services' operation, though this provides little benefit since attackers can accumulate nominal (non-slashable) stake through loans.
Setting this aside, while their model shares with ours the possibility of partial allocations, it differs crucially.
Their approach aims to eliminate stake reuse between services, while our elastic model demonstrates that carefully managed stake reuse can enhance overall network security.

Mamageishvili and Sudakov~\cite{mamageishvili2025cost} analyze the efficiency tradeoffs between restaking and vanilla Proof-of-Stake protocols by comparing their stake requirements, showing that restaking can provide significant savings.
While they focus on efficiency comparisons, our work purposes a more robust mechanism and analyzes the security and robustness of restaking networks against coordinated attacks and Byzantine failures.

\paragraph{Liquid Restaking Tokens}
\emph{Liquid restaking tokens (LRTs)}~\cite{gogol2024sok} are fungible tokens that represent restaked positions, allowing holders to maintain liquidity while their stake secures multiple services.
While recent work has examined LRTs' market risks \citep{alexander2024leveraged} and financial properties \citep{neuder2024risks}, we focus on the cryptoeconomic security and robustness of their underlying restaking networks.

\paragraph{Security Through Incentives}
The study of security from the perspective of incentives is common in the blockchain literature~\cite{liu2019survey}.
Examples span the consensus-layer: incentive-compatible protocol design~\cite{pass2017fruitchains,abraham2023colordag}, selfish mining~\cite{eyal2018majority,sapirshtein2017optimal,carlsten2016instability}, and other attacks~\cite{eyal2015miner,mirkin2020bdos,judmayer2021sok,yaish2023uncle,karakostas2024blockchain};
payment channels: attack discovery~\cite{brugger2023checkmate}, and secure design~\cite{tsabary2021mad,aumayr2022sleepy,aumayr2024securing};
and applications: attack discovery~\cite{cui2022vrust,babel2023lanturn,li2023demystifying}, and secure design~\cite{dong2017betrayal,tsabary2024ledgerhedger,wadhwa2024data}.

\paragraph{Systemic Risk}
Previous work on systemic risk in financial networks, where entities are connected by debt obligations, has studied both factors affecting risk propagation~\cite{acemoglu2015networks,acemoglu2015systemic,glasserman2016contagion} and frameworks for measuring these risks~\cite{brunnermeier2012risk,chen2013axiomatic,battiston2016price}.
Our model extends these ideas to restaking networks where security dependencies arise from shared stake rather than debt obligations, though with different dynamics since stake can be reused across multiple services simultaneously.


\section{Restaking Networks and Elastic Restaking}
\label{section:model}


We begin by presenting the components of a restaking network: validators allocate stake to services, which secure assets~(\S\ref{section:model:restaking_networks}).
We then present how a coalition validators can \emph{attack} services, and the \emph{cryptoeconomic security game} that arises~(\S\ref{section:model:cryptoeconomic_security_game}).
Later, we present the \emph{cryptoeconomic robustness game} that arises when an adversary with a budget pays validators to attack services~(\S\ref{section:model:cryptoeconomic_robustness_game}).
Finally, we consider robustness against Byzantine services that slash their validators, and leave the network more vulnerable in the cryptoeconomic robustness game~(\S\ref{section:model:robustness_against_byzantine_services}).


\subsection{Principals and Stake Allocation}
\label{section:model:restaking_networks}


A restaking network comprises a set of~$n$ services~${\allServices=\{\service_1, \service_2, \dots, \service_n\}}$ and a set of~$m$ validators~${\allValidators=\{\validator_1, \validator_2, \dots, \validator_m\}}$.
Each validator~$\validator \in \allValidators$ has a \emph{stake}~$\stake{\validator} \in \positiveRealNumbers$.
Each validator~$\validator \in \allValidators$ also has an \emph{allocation}~$\allocation{\validator}{\service}$ in the closed interval~$\left[ 0, \stake{\validator} \right]$ to each service~$\service \in \allServices$.
The allocation~$\allocation{\validator}{\service}$ represents validator~$\validator$'s stake dedicated to service~$\service$, determining their maximum possible loss from misbehavior or service failure, and affecting their reward from validating the service.
Formally,~${\allStakes: \allValidators \to \positiveRealNumbers}$ and~${\allAllocations: \allValidators \times \allServices \to \nonNegativeRealNumbers}$ are the stake and allocation functions.

This creates a weighted bipartite graph~$(\allValidators, \allServices, \allAllocations)$ where validators and services are the two sets of vertices.
The weight of an edge from a validator~$\validator$ to a service~$\service$ is the validator's allocation to the service~$\allocation{\validator}{\service}$.
A weight can be zero, meaning the validator does not allocate any stake to the service and does not validate it.
And the sum of the weights of the edges from a validator to all services can exceed the validator's stake.

A network is~\emph{atomic} if validators can only allocate their entire stake or none to a service.
That is, for each validator~$\validator \in \allValidators$ and service~$\service \in \allServices$,~$\allocation{\validator}{\service} \in \{0, \stake{\validator}\}$.
Otherwise, the network is~\emph{elastic}.

A validator's \emph{restaking degree} measures how heavily encumbered their stake is to the services they allocate to.
\begin{definition}[Restaking Degree]
  \label{definition:restaking_degree}
  In a restaking network~$\networkState$, the restaking degree of a validator~$\validator$ is the ratio of the sum of their allocations and their stake, that is,
  \begin{equation}
    \label{equation:model:restaking_degree}
    \restakingDegree{\validator} = \frac{\sum_{\service \in \allServices} \allocation{\validator}{\service}}{\stake{\validator}} .
  \end{equation}
\end{definition}
In symmetric restaking networks, where all validators share the same restaking degree, we refer to this common restaking degree as the network's restaking degree, denoted~$\restakingDegreeSymmetric$.

Each service~$\service \in \allServices$ has an \emph{attack prize}~$\attackPrize{\service} \in \positiveRealNumbers$ and an \emph{attack threshold}~$\attackThreshold{\service} \in [0,1]$.
When validators collectively allocate more than~$\attackThreshold{\service}$ of service~$\service$'s stake, they can misbehave and extract assets worth~$\attackPrize{\service}$ from it.
Formally,~$\allAttackThresholds: \allServices \to [0,1]$ and~$\allAttackPrizes: \allServices \to \positiveRealNumbers$ are the attack threshold and prize functions.

Together with the previous elements, a restaking network is defined by the tuple~${\networkState=(\allValidators, \allServices, \allStakes, \allAllocations, \allAttackThresholds, \allAttackPrizes)}$.


\subsection{The Cryptoeconomic Security Game}
\label{section:model:cryptoeconomic_security_game}


The cryptoeconomic security game is a game played between the validators~$\allValidators$.
Each validator~$\validator \in \allValidators$ can choose to use~$\attackStake{\validator}{\service} \in \left[ 0, \allocation{\validator}{\service} \right]$ of their stake to attack service~$\service \in \allServices$.
We call~$\allAttackStakes: \allValidators \times \allServices \to \nonNegativeRealNumbers$ the~\emph{attacking stake function} or simply an~\emph{attack}.
Formally, the strategy space for all validators is all legal attacking stake functions, that is,~${\allStrategySpacesSecurityGame = \left\{ \allAttackStakes: \allValidators \times \allServices \to \nonNegativeRealNumbers \middle| \attackStake{\validator}{\service} \leq \allocation{\validator}{\service} \right\}}$.

We call such attacks \emph{allocation-divisible}, as validators can choose to use only portions of their allocations.
If in an attack, validators either use their allocations in their entirety or not at all, we call the attack~\emph{allocation-indivisible}.
That is, if for all validators~$\validator \in \allValidators$ and services~$\service \in \allServices$,~${\attackStake{\validator}{\service} \in \{0, \allocation{\validator}{\service}\}}$.

For an attacking stake function~$\allAttackStakes$, let~$\attackedServices$ be all attacked services, services for which enough stake is dedicated to attacking them.
\begin{definition}[Attacked Services]
    \label{definition:attacked_services}
    Given an attack~$\allAttackStakes$, the set of attacked services is
    \begin{equation}
        \attackedServices = \left\{ \service \in \allServices \middle| \sum_{\validator \in \allValidators} \attackStake{\validator}{\service}
        \geq \attackThreshold{\service} \cdot \sum_{\validator \in \allValidators} \allocation{\validator}{\service} \right\} .
    \end{equation}
\end{definition}

As the same stake may secure several services, calculating the cost of using the stake to attack the services is more involved than simply summing the~$\attackStake{\validator}{\service}$ values.
A validator can only be slashed up to the stake they have, even if the sum of their allocations exceeds it.
Denote by~$\validatorAttackCost{\validator}{\allAttackStakes}$ the cost of validator~$\validator$ for the attack~$\allAttackStakes$:
The sum of the portions of the stake they use to attack the services, capped at the validator's stake, namely,
\begin{equation}
    \label{equation:validator_attack_cost}
    \validatorAttackCost{\validator}{\allAttackStakes} = \min \left( \stake{\validator}, \sum_{\service \in \attackedServices} \attackStake{\validator}{\service} \right) .
\end{equation}
Then, denote by~$\attackCost{\allAttackStakes}$ the total cost of the attack:
The sum of the costs of the validators in the coalition, namely,
\begin{equation}
    \label{equation:total_attack_cost}
    \attackCost{\allAttackStakes} = \sum_{\validator \in \allValidators} \validatorAttackCost{\validator}{\allAttackStakes} .
\end{equation}
And denote by~$\totalAttackPrize{\allAttackStakes}$ the prize of the attack:
The sum of the prizes of the attacked services, namely,
\begin{equation}
    \label{equation:total_attack_prize}
    \totalAttackPrize{\allAttackStakes} = \sum_{\service \in \attackedServices} \attackPrize{\service} .
\end{equation}
If the set~$\attackedServices$ is empty, the prize is~$0$.

We are now ready to present the utilities of players in the cryptoeconomic security game.
All validators lose the cost of the stake they use, and split the prizes (if any) among themselves according to the cost of each validator.
If the cost was~0 (perhaps the result of a service with no stake allocated to it), we simply split it evenly.
Denote by~$\validatorPrizeShareSecurityGame{\validator}{\allAttackStakes}$ the share of validator~$\validator$ out of the total prize of the attack:
\begin{equation}
    \label{equation:validator_prize_share_security_game}
    \validatorPrizeShareSecurityGame{\validator}{\allAttackStakes}
    = \begin{cases}
        \frac{\validatorAttackCost{\validator}{\allAttackStakes}}{\attackCost{\allAttackStakes}} & \text{if } \attackCost{\allAttackStakes} > 0 ; \\
        \frac{1}{|\allValidators|} & \text{if } \attackCost{\allAttackStakes} = 0 .
    \end{cases}
\end{equation}
Then, given an attack~$\allAttackStakes$, the utility of validator~$\validator$ is
\begin{equation}
    \label{equation:validator_utility_security_game}
    \validatorUtilitySecurityGame{\validator}{\allAttackStakes} = \validatorPrizeShareSecurityGame{\validator}{\allAttackStakes} \cdot \totalAttackPrize{\allAttackStakes} - \validatorAttackCost{\validator}{\allAttackStakes} .
\end{equation}

To define when a network is considered \emph{cryptoeconomically secure}, we use a modified notion of a strong Nash equilibrium.
Instead of requiring that there exists no coalition that can deviate and strictly increase the utility of each of its participants~\cite{aumann1959acceptable}, we require that no coalition can non-strictly increase their utilities.
Our notion is equivalent to the following definition.
\begin{definition}[Strong* Nash Equilibrium]
    \label{definition:strong_nash_equilibrium}
    Let~$(P, \Sigma, u)$ be a strategic form game.
    A strategy profile~$\sigma_{\textit{sne}} \in \Sigma$ is a strong* Nash equilibrium if for all coalitions of players~${P' \subseteq P}$ all possible deviations from~$\sigma_{\textit{sne}}$ leading to an alternative strategy profile~$\sigma \in \Sigma$ result in at least one player~$p \in P'$ being strictly worse off:~$u_p(\sigma) < u_p(\sigma_{\textit{sne}})$.
\end{definition}
For brevity, we refer to this modified notion as simply a strong Nash equilibrium throughout the rest of the paper.

Now, we are ready to present the condition under which a restaking network is considered \emph{cryptoeconomically secure}:
\begin{definition}[Restaking Network Cryptoeconomic Security]
    \label{definition:restaking_network_security}
    Let~$\networkState$ be a restaking network and consider the attacking stake function~$\allAttackStakesAt{0}$ such that for all validators~$\validator \in \allValidators$ and services~$\service \in \allServices$: $\attackStake{\validator}{\service} = 0$.
    Then,~$\networkState$ is \emph{cryptoeconomically secure} (or simply~\emph{secure}) if~$\allAttackStakesAt{0}$ is a strong Nash equilibrium of the cryptoeconomic security game for~$\networkState$ and no services are attacked, that is,~$\attackedServicesAt{0} = \emptyset$.
\end{definition}

We now precisely define the conditions under which an attack is considered \emph{profitable}, which will be useful when analyzing the cryptoeconomic security game.
\begin{definition}[Attack Profitability]
    \label{definition:attack_profitability}
    An attack~$\allAttackStakes$ is~\emph{profitable} if it results with at least one attacked service, namely,~$\attackedServices \neq \emptyset$, and
    \begin{equation}
        \attackCost{\allAttackStakes} \leq \totalAttackPrize{\allAttackStakes} .
    \end{equation}
\end{definition}


\subsection{The Cryptoeconomic Robustness Game}
\label{section:model:cryptoeconomic_robustness_game}


The cryptoeconomic robustness game is similar to the cryptoeconomic security game except one key difference.
An adversary has a budget~$\adversaryBudget \in \nonNegativeRealNumbers$ for attacking the network and if there is at least one attacked service, the adversary pays their budget to validators.
Thus, the prizes from attacking services may only partially reimburse the cost of the stake used in the attack.

The set of players and their strategies remains the same as in the cryptoeconomic security game, but the utilities are different.
Given an attack~$\allAttackStakes$, the utility of validator~$\validator$ is
\begin{equation}
    \label{equation:validator_utility_robustness_game}
    \validatorUtilitySecurityGame{\validator}{\allAttackStakes}
    = \begin{cases}
        \validatorPrizeShareSecurityGame{\validator}{\allAttackStakes} \left( \totalAttackPrize{\allAttackStakes} + \adversaryBudget \right) - \validatorAttackCost{\validator}{\allAttackStakes} & \text{if } \attackedServices \neq \emptyset ; \\ 
        - \validatorAttackCost{\validator}{\allAttackStakes} & \text{otherwise} .
    \end{cases}
\end{equation}

Complementary to the cryptoeconomic security game, we present the condition under which a restaking network is considered \emph{cryptoeconomically robust}.
\begin{definition}[Restaking Network Cryptoeconomic Robustness]
    \label{definition:restaking_network_robustness}
    Let~$\networkState$ be a restaking network and consider the attacking stake function~$\allAttackStakesAt{0}$ such that for all validators~$\validator \in \allValidators$ and services~$\service \in \allServices$.~$\attackStake{\validator}{\service} = 0$.
    Then,~$\networkState$ is \emph{$\adversaryBudget$-cryptoeconomically robust} (or~\emph{$\adversaryBudget$-budget robust}) if~$\allAttackStakesAt{0}$ is a strong Nash equilibrium of the cryptoeconomic robustness game for~$\networkState$ with an adversary budget of~$\adversaryBudget$ and no services are attacked, that is,~$\attackedServicesAt{0} = \emptyset$.
\end{definition}

In addition, we define a \emph{$\adversaryBudget$-costly} attack, which will be useful when analyzing the cryptoeconomic robustness game.
\begin{definition}[$\adversaryBudget$-costly Attack]
    \label{definition:beta_costly_attack}
    An attack is \emph{$\adversaryBudget$-costly} if it results with at least one attacked service, i.e.,~$\attackedServices \neq \emptyset$, and
    \begin{equation}
        \attackCost{\allAttackStakes} \leq \totalAttackPrize{\allAttackStakes} + \adversaryBudget .
    \end{equation}
\end{definition}
Note that a $0$-costly attack is a profitable attack.

\begin{figure}[t]
    \centering
    \begin{subfigure}[b]{0.22\textwidth}
        \centering
        \begin{tikzpicture}[
            validator/.style={circle, draw, minimum size=0.6cm, font=\scriptsize},
            service/.style={circle, draw, minimum size=0.6cm, font=\scriptsize, inner sep=0.035cm},
            scale=0.5
        ]
            \node[align=center] at (0,3) {$\allValidators$ \\ $(\allStakes)$};
            \node[align=center] at (1.5, 3) {---};
            \node[align=center] at (3,3) {$\allServices$ \\ $(\allAttackThresholds | \allAttackPrizes)$};
            
            \node[validator] (V1) at (0,1) {20};
            \node[validator] (V2) at (0,-1) {20};
            
            \node[service] (S) at (3,0) {$\frac{1}{2}~|~5$};

            \draw (V1) -- (S) ;
            \draw (V2) -- (S) ;
        \end{tikzpicture}
        \caption{Initial state}
        \label{figure:robustness_notion_comparison_network}
    \end{subfigure}
    \hfill
    \begin{subfigure}[b]{0.22\textwidth}
        \centering
        \begin{tikzpicture}[
            validator/.style={circle, draw, minimum size=0.6cm, font=\scriptsize},
            service/.style={circle, draw, minimum size=0.6cm, font=\scriptsize, inner sep=0.035cm},
            scale=0.5
        ]
            \node[align=center] at (0,3) {$\allValidators$ \\ $(\allStakes)$};
            \node[align=center] at (1.5, 3) {---};
            \node[align=center] at (3,3) {$\allServices$ \\ $(\allAttackThresholds | \allAttackPrizes)$};
            
            \node[validator,fill=red!20] (V1) at (0,1) {20};
            \node[validator] (V2) at (0,-1) {20};
            
            \node[service,fill=red!20] (S) at (3,0) {$\frac{1}{2}~|~5$};

            \draw (V1) -- (S) ;
            \draw (V2) -- (S) ;
        \end{tikzpicture}
        \caption{Attack in our model}
        \label{figure:robustness_notion_comparison_our_attack}
    \end{subfigure}
    \hfill
    \begin{subfigure}[b]{0.22\textwidth}
        \centering
        \begin{tikzpicture}[
            validator/.style={circle, draw, minimum size=0.6cm, font=\scriptsize},
            service/.style={circle, draw, minimum size=0.6cm, font=\scriptsize, inner sep=0.035cm},
            scale=0.5
        ]
            \node[align=center] at (0,3) {$\allValidators$ \\ $(\allStakes)$};
            \node[align=center] at (1.5, 3) {---};
            \node[align=center] at (3,3) {$\allServices$ \\ $(\allAttackThresholds | \allAttackPrizes)$};
            
            \node[validator] (V1) at (0,1) {5};
            \node[validator] (V2) at (0,-1) {5};
            
            \node[service] (S) at (3,0) {$\frac{1}{2}~|~5$};

            \draw (V1) -- (S) ;
            \draw (V2) -- (S) ;
        \end{tikzpicture}
        \caption{Initial stake loss~\cite{durvasula2024robust}}
        \label{figure:robustness_notion_comparison_their_initial_loss}
    \end{subfigure}
    \hfill
    \begin{subfigure}[b]{0.22\textwidth}
        \centering
        \begin{tikzpicture}[
            validator/.style={circle, draw, minimum size=0.6cm, font=\scriptsize},
            service/.style={circle, draw, minimum size=0.6cm, font=\scriptsize, inner sep=0.035cm},
            scale=0.5
        ]
            \node[align=center] at (0,3) {$\allValidators$ \\ $(\allStakes)$};
            \node[align=center] at (1.5, 3) {---};
            \node[align=center] at (3,3) {$\allServices$ \\ $(\allAttackThresholds | \allAttackPrizes)$};
            
            \node[validator,fill=red!20] (V1) at (0,1) {5};
            \node[validator] (V2) at (0,-1) {5};
            
            \node[service,fill=red!20] (S) at (3,0) {$\frac{1}{2}~|~5$};

            \draw (V1) -- (S) ;
            \draw (V2) -- (S) ;
        \end{tikzpicture}
        \caption{Attack after stake loss~\cite{durvasula2024robust}}
        \label{figure:robustness_notion_comparison_their_attack}
    \end{subfigure}
    \caption{Comparison of our robustness notion with the one of \citet{durvasula2024robust}.}
    \label{figure:robustness_notion_comparison}
    \Description{Illustration showing the order of operations in our robustness notion and the one in previous work.}
\end{figure}

The robustness notion in our model diverges from the one of \citet{durvasula2024robust}.
While they consider an initial stake loss followed by an attack, we consider an attack that may be partially reimbursed by an adversary.
For example, suppose a service has~40 units of stake and requires validators to attack with half of the service's stake to capture a prize of~5 units in an atomic restaking network~(Fig.~\ref{figure:robustness_notion_comparison_network}).
In their model, an attack becomes profitable only after the network suffers an initial stake loss of~30 units~(Fig.~\ref{figure:robustness_notion_comparison_their_initial_loss}), which reduces the service's total stake to~10 units, making it vulnerable to validators with~5 units who can capture the prize~(Fig.~\ref{figure:robustness_notion_comparison_their_attack}).
In contrast, our model enables validators to use~20 units of stake to attack the service from the outset~(Fig.~\ref{figure:robustness_notion_comparison_our_attack}).
They then capture~5 units of stake and the adversary directly reimburses the validators for their losses---15 units of stake, which is significantly lower than the~30 units required in their model.
Thus, although both models ultimately balance the attack cost with the prize, our approach realistically requires a smaller adversarial investment than the initial stake losses needed in their model.


\subsection{Elastic Restaking Against Byzantine Services}
\label{section:model:robustness_against_byzantine_services}


We also aim to capture the robustness of a restaking network to Byzantine services.
A Byzantine service causes a mass slashing of all the stake that was allocated to it, as if all validators attacked the Byzantine service with their entire allocations~\cite{durvasula2024robust}.
In practice, this could be the result of a benign design flaw, or a malicious service design.

Consider a restaking network~$\networkStateAt{0} = (\allValidatorsAt{0}, \allServicesAt{0}, \allStakesAt{0}, \allAllocationsAt{0}, \allAttackThresholdsAt{0}, \allAttackPrizesAt{0})$.
An adversary chooses a subset~$\byzantineServices \subseteq \allServicesAt{0}$ of the services to be Byzantine, causing the network to transition to a new state, denoted by~${\networkStateAt{1} = \networkStateAt{0} \networkAdvance \byzantineServices}$.
The transition occurs as follows.

Let~${\networkStateAt{1} = (\allValidatorsAt{1}, \allServicesAt{1}, \allStakesAt{1}, \allAllocationsAt{1}, \allAttackThresholdsAt{1}, \allAttackPrizesAt{1})}$ be the new state.
First, validators remain the same, namely,~$\allValidatorsAt{1} = \allValidatorsAt{0}$.
Second, Byzantine services are removed from the network; the new set of services is~$\allServicesAt{1} = \allServicesAt{0} \setminus \byzantineServices$.
Third, each validator~$\validator \in \allValidatorsAt{0}$ is slashed for the stake they allocated to the Byzantine services~$\byzantineServices$, capped by their total stake~$\stakeAt{0}{\validator}$.
To specify these dynamics, we use the notation of function restriction.
Let~$f: A \to B$ be a function from set~$A$ to set~$B$ and let set~$C \subseteq A$ be a subset of~$A$.
Then, the function restriction of~$f$ to~$C$ is the function~$\restr{f}{C}: C \to B$ defined as~$\restr{f}{C}(x) = f(x)$ for all~$x \in C$.
The new stake is given by    
\begin{multline}
    \label{equation:stake_after_byzantine_services_cause_slashing}
    \stakeAt{1}{\validator}
    = \stakeAt{0}{\validator} - \validatorAttackCostAt{0}{\validator}{\byzantineServices}{\restr{\allAllocationsAt{0}}{\allValidatorsAt{0} \times \byzantineServices}} \\
    \underset{\eqref{equation:validator_attack_cost}}{=} \stakeAt{0}{\validator} - \min \left( \stakeAt{0}{\validator}, \sum_{\service \in \byzantineServices} \allocationAt{0}{\validator}{\service} \right) \\
    = \max \left( 0, \stakeAt{0}{\validator} - \sum_{\service \in \byzantineServices} \allocationAt{0}{\validator}{\service} \right) .
\end{multline}
Since a validator cannot allocate more stake to a service than their entire stake, allocations are adjusted in the following way.
Allocations of validators with sufficient stake remain the same, while allocations of validators with insufficient stake are reduced to be equal to the remaining stake.
Formally, the new allocation function is given by
\begin{equation}
    \label{equation:allocation_after_byzantine_services_cause_slashing}
    \allocationAt{1}{\validator}{\service} = \min \left( \allocationAt{0}{\validator}{\service}, \stakeAt{1}{\validator} \right) .
\end{equation}
And lastly, attack thresholds and attack prizes of Byzantine services are removed, and the new attack thresholds and attack prizes are given by~$\allAttackThresholdsAt{1} = \restr{\allAttackThresholdsAt{0}}{\allServicesAt{1}}$ and~$\allAttackPrizesAt{1} = \restr{\allAttackPrizesAt{0}}{\allServicesAt{1}}$.

\begin{figure}[t]
    \centering
    \begin{subfigure}[b]{0.22\textwidth}
        \centering
        \begin{tikzpicture}[
            validator/.style={circle, draw, minimum size=0.5cm},
            service/.style={circle, draw, minimum size=0.5cm},
            scale=0.5
        ]
            \node[align=center] at (0,4) {$\allValidators$ \\ $(\allStakes)$};
            \node[align=center] at (1.5, 4) {--- \\ $\allAllocations$};
            \node[align=center] at (3,4) {$\allServices$ \\ $( )$};
            
            \node[validator] (V) at (0,0) {2};
            
            \node[service] (S1) at (3,2) {};
            \node[service] (S2) at (3,0) {};
            \node[service] (S3) at (3,-2) {};
            
            \draw (V) -- (S1) node[pos=0.3, above left] {1};
            \draw (V) -- (S2) node[pos=0.3, above] {1};
            \draw (V) -- (S3) node[pos=0.3, below left] {1};
        \end{tikzpicture}
        \caption{Before slashing}
        \label{figure:expressiveness_example_before}
    \end{subfigure}
    \hfill
    \begin{subfigure}[b]{0.22\textwidth}
        \centering
        \begin{tikzpicture}[
            validator/.style={circle, draw, minimum size=0.5cm},
            service/.style={circle, draw, minimum size=0.5cm},
            scale=0.5
        ]
            \node[align=center] at (0,4) {$\allValidators$ \\ $(\allStakes)$};
            \node[align=center] at (1.5, 4) {--- \\ $\allAllocations$};
            \node[align=center] at (3,4) {$\allServices$ \\ $( )$};
            
            \node[validator] (V) at (0,0) {1};
            
            \node[service] (S2) at (3,0) {};
            \node[service] (S3) at (3,-2) {};
            
            \draw (V) -- (S2) node[pos=0.3, above] {1};
            \draw (V) -- (S3) node[pos=0.3, below left] {1};
            
            \node[service,dashed,fill=gray!20] (S1) at (3,2) {};
        \end{tikzpicture}
        \caption{After $\service_1$ slashing}
        \label{figure:expressiveness_example_after}
    \end{subfigure}
    \hfill
    \begin{subfigure}[b]{0.22\textwidth}
        \centering
        \begin{tikzpicture}[
            validator/.style={circle, draw, minimum size=0.5cm},
            service/.style={circle, draw, minimum size=0.5cm},
            scale=0.5
        ]
            \node[align=center] at (0,4) {$\allValidators$ \\ $(\allStakes)$};
            \node[align=center] at (1.5, 4) {--- \\ $\allAllocations$};
            \node[align=center] at (3,4) {$\allServices$ \\ $( )$};
            
            \node[validator] (V) at (0,0) {5};
            
            \node[service] (S1) at (3,2) {};
            \node[service] (S2) at (3,0) {};
            \node[service] (S3) at (3,-2) {};
            
            \draw (V) -- (S1) node[pos=0.3, above left] {3};
            \draw (V) -- (S2) node[pos=0.3, above] {3};
            \draw (V) -- (S3) node[pos=0.3, below left] {1};
        \end{tikzpicture}
        \caption{Before slashing}
        \label{figure:byzantine_service_example_before}
    \end{subfigure}
    \hfill
    \begin{subfigure}[b]{0.22\textwidth}
        \centering
        \begin{tikzpicture}[
            validator/.style={circle, draw, minimum size=0.5cm},
            service/.style={circle, draw, minimum size=0.5cm},
            scale=0.5
        ]
            \node[align=center] at (0,4) {$\allValidators$ \\ $(\allStakes)$};
            \node[align=center] at (1.5, 4) {--- \\ $\allAllocations$};
            \node[align=center] at (3,4) {$\allServices$ \\ $( )$};
            
            \node[validator] (V) at (0,0) {2};
            
            \node[service] (S2) at (3,0) {};
            \node[service] (S3) at (3,-2) {};
            
            \draw (V) -- (S2) node[pos=0.3, above] {2};
            \draw (V) -- (S3) node[pos=0.3, below left] {1};
            
            \node[service,dashed,fill=gray!20] (S1) at (3,2) {};
        \end{tikzpicture}
        \caption{After $\service_1$ slashing}
        \label{figure:byzantine_service_example_after}
    \end{subfigure}
    \caption{Illustration of 2 elastic restaking networks stretching stake after 1 allocation is slashed.}
    \label{figure:elastic_restaking_networks_examples}
    \Description{Illustration showing two elastic restaking networks with stake adjustments after one allocation is slashed.}
\end{figure}

Let us consider two examples.
Take the network in Fig.~\ref{figure:expressiveness_example_before} with a Byzantine service~$\service_1$.
After the service causes a mass slashing, the network transitions to the state in Fig.~\ref{figure:expressiveness_example_after}.
The validator loses~$1$ unit of stake while allocations to remaining services remain the same since there's sufficient stake remaining.
Now take the network in Fig.~\ref{figure:byzantine_service_example_before}.
In this case, the validator would lose~$3$ units of stake from~$\service_1$'s slashing, leaving only~$2$ units of stake.
Since a validator cannot allocate more than their remaining stake, their allocation to~$\service_2$ would be reduced to~$2$~(Fig.~\ref{figure:byzantine_service_example_after}).

Following service failures, we check their impact on the security of the resultant network.
In general, the more Byzantine services required to reach an insecure network, the more robust the network.
But, it is necessary to account for the different magnitudes of the services that coexist in the network.
We assume that the adversary can choose up to a weighted fraction~$\maxByzantineServices$ of the services to be Byzantine, where each service is weighted by the ratio of its attack prize to its attack threshold;
this is the stake required to secure the service in isolation.

Some restaking networks may contain what we call a~\emph{base} service:
A service that cannot be made Byzantine.
In the EigenLayer restaking model, Ethereum is a base service.
If Ethereum fails, all EigenLayer's infrastructure collapses, and the restaking network would no longer be functional.
Thus, we restrict the adversary's choice of Byzantine services to only include services that are not base services.
Let~$\baseServices(\networkState)$ be the set of base services in~$\networkState$.
For brevity, we omit this detail in the notation of a restaking network~$\networkState$, and unless stated otherwise, we assume that there are no base services.

Formally, for a restaking network~$\networkState=(\allValidators, \allServices, \allStakes, \allAllocations, \allAttackThresholds, \allAttackPrizes)$, the adversary can choose any subset in
\begin{equation}
    \byzantineSubsets{\maxByzantineServices}
    = \left\{ \byzantineServices \subseteq \allServices \setminus \baseServices(\networkState) \middle| \sum_{\service \in \byzantineServices} \frac{\attackPrize{\service}}{\attackThreshold{\service}} \leq \maxByzantineServices \right\} .
\end{equation}

We are now ready to define the robustness of a network to both adversarial subsidy and Byzantine services.
\begin{definition}[$(\maxByzantineServices, \adversaryBudget)$-robust Network]
    \label{definition:overall_robust_network}
    A network~$\networkState$ is \emph{$(\maxByzantineServices, \adversaryBudget)$-robust} if  for all~$\byzantineServices \in \byzantineSubsets{\maxByzantineServices}$ the network~${\networkState \networkAdvance \byzantineServices}$ is $\adversaryBudget$-budget robust.
\end{definition}


\section{Elastic Restaking Networks Are More Expressive}
\label{section:model:elastic_restaking_networks}


Elastic restaking networks allow validators to allocate only a portion of their stake to a service and simultaneously have more stake allocated to services than their total stake.
We show that elastic networks allow us to express behavior that cannot be simulated in atomic networks.

For example, consider the previous example, illustrated in Fig.~\ref{figure:expressiveness_example_before}, where an elastic restaking network stretches its stake to cover remaining allocations.
The next proposition shows that atomic restaking networks cannot express the behavior in the example, since the allocations to the remaining services are already determined.
This holds even if we allow the validator to partition their stake and treat each portion as an individual validator with their own allocations.
\begin{proposition}
    \label{proposition:expressiveness_of_atomic_networks}
    Let~$x \in \positiveRealNumbers$.
    There exists no atomic restaking network~${\networkState = (\allValidators, \allServices, \allStakes, \allAllocations, \allAttackThresholds, \allAttackPrizes)}$ that satisfies the following conditions:
    (1) The total stake in the network is less than~$x$ times the number of services;
    (2) each service has exactly~$x$ units of stake allocated to it; and
    (3) after any service fails and slashes its allocated stake, each remaining service maintains exactly~$x$ units of stake.
\end{proposition}
The proof is deferred to Appendix~\ref{appendix:proofs_from_section_elastic_restaking_networks_are_more_expressive}.
The proposition yields the following corollary.
\begin{corollary}
    Elastic restaking networks are strictly more expressive than atomic ones.
\end{corollary}
\begin{proof}
    First, any atomic restaking network is trivially an elastic restaking network where validators happen to only make all-or-nothing allocations.
    Second, there exist behaviors possible in elastic networks that are impossible in atomic networks: Figures~\ref{figure:expressiveness_example_before} and~\ref{figure:expressiveness_example_after} show a network where each service maintains equal stake before and after failures, which Proposition~\ref{proposition:expressiveness_of_atomic_networks} proves is impossible for any atomic network.
\end{proof}


\section{Security Analysis}
\label{section:security_analysis}


We first show that in the restaking network security game not attacking is a strong Nash equilibrium, if and only if there are no profitable attacks in the network.
We identify sufficient conditions for security in elastic restaking networks, which are analogous to conditions previously identified by EigenLayer~(\S\ref{section:security_analysis:sufficient_conditions_for_security}).
However, to learn about a network's robustness---which is one of the major goals in this paper---sufficient conditions are not enough;
we must accurately determine whether a network is secure or not with respect to a given adversary.
We prove that in the general case this is NP-hard~(\S\ref{section:security_analysis:searching_for_attacks_is_np_complete}) and solve the symmetric case~(\S\ref{section:security_analysis:symmetric_case}).
We defer all proofs to Appendix~\ref{appendix:proofs_from_section_security_analysis}.


We begin by presenting a computable condition for restaking network security.
\begin{proposition}
    \label{proposition:restaking_network_security}
    A restaking network~$\networkState$ is cryptoeconomically secure if and only if there exists no profitable attack.
\end{proposition}


\subsection{Sufficient Conditions for Security}
\label{section:security_analysis:sufficient_conditions_for_security}




A sufficient condition for a network to be secure was identified by EigenLayer~\cite{eigenlayer2024restaking}
under the (very strong) assumption that misbehaving validators are slashed not only for stake allocated and used for misbehavior, but for all their stake.
Instead of our cost function~(\eqreft{equation:validator_attack_cost}), the cost for a misbehaving validator is their entire stake:
\begin{equation}
    \label{equation:eigenlayer_condition_simplified_assumption}
    \validatorAttackCost{\validator}{\allAttackStakes}
    = \begin{cases}
        \stake{\validator} & \text{if } \sum_{\service \in \attackedServices} \attackStake{\validator}{\service} > 0 ; \\
        0 & \text{otherwise} .
    \end{cases}
\end{equation}

This is the case for atomic restaking networks when only allocation-indivisible attacks are considered, which was the case considered in previous work~\cite{eigenlayer2024restaking,durvasula2024robust}.
We extend this result to include allocation-divisible attacks in elastic restaking networks using the above cost function.

\begin{theorem}[EigenLayer Condition]
    \label{theorem:eigenlayer_condition}
    A network~$\networkState$ is secure if a misbehaving validator is slashed for their stake~(\eqreft{equation:eigenlayer_condition_simplified_assumption}), and for all validators~$\validator \in \allValidators$:
    \begin{equation}
        \sum_{\service \in \allServices} \frac{\allocation{\validator}{\service}}{\sum_{\validator' \in \allValidators} \allocation{\validator'}{\service}} \cdot \frac{\attackPrize{\service}}{\attackThreshold{\service}} < \stake{\validator} .
    \end{equation}
\end{theorem}




The previous result does not apply in our model, where slashing of misbehaving validators is more nuanced.
For example, consider a network with one validator~$\validator$ with~$\stake{\validator} = 2$ and one service~$\service$ with~$\attackPrize{\service} = 1$ and~$\attackThreshold{\service} = 1$.
If the validator allocates only one unit of stake to the service, i.e.,~$\allocation{\validator}{\service} = 1$, the network is not secure, as the attack~$\allAttackStakes$ where~$\allAttackStakes(\validator, \service) = 1$ is profitable.
Since validator~$\validator$ controls all the stake that secures service~$\service$, and uses their entire allocation to attack it as~$\attackStake{\validator}{\service} = 1$,~${\attackedServices = \{\service\}}$.
And since the cost of the attack is 1 unit of stake, while the prize is also 1 unit, the attack is profitable.
Nonetheless, the condition of Theorem~\ref{theorem:eigenlayer_condition} is satisfied, as
\begin{equation}
    \frac{\allocation{\validator}{\service}}{\sum_{\validator' \in \allValidators} \allocation{\validator'}{\service}} \cdot \frac{\attackPrize{\service}}{\attackThreshold{\service}} = \frac{1}{1} \cdot \frac{1}{1} = 1 < \stake{\validator} = 2 .
\end{equation}

To overcome this issue, we generalize the condition of Theorem~\ref{theorem:eigenlayer_condition}, where networks may be elastic and attacks may be allocation-divisible.
We propose the following sufficient condition for network security.
\begin{proposition}[Generalized EigenLayer Condition]
    \label{proposition:generalized_eigenlayer_condition}
    A network~$\networkState$ is secure if all validators~$\validator \in \allValidators$ should be slashed by less than their total stake:
    \begin{equation}
        \sum_{\service \in \allServices} \frac{\allocation{\validator}{\service}}{\sum_{\validator' \in \allValidators} \allocation{\validator'}{\service}} \cdot \frac{\attackPrize{\service}}{\attackThreshold{\service}} < \stake{\validator} ,
    \end{equation}
    and all services~$\service \in \allServices$ have sufficient stake to cover their prizes:
    \begin{equation}
        \sum_{\validator \in \allValidators} \allocation{\validator}{\service} > \frac{\attackPrize{\service}}{\attackThreshold{\service}} .
    \end{equation}
\end{proposition}


\subsection{Searching for Attacks is NP-Complete}
\label{section:security_analysis:searching_for_attacks_is_np_complete}


If a network does not fulfill the sufficient conditions, to check whether it is cryptoeconomically secure we ask whether there exists a profitable attack. 
However, in general, we show this problem is NP-complete, namely: (1) The problem is in NP and (2) there exists a polynomial-time reduction from some known NP-complete problem.

We first prove for allocation-indivisible attacks.

\begin{proposition}
    \label{proposition:allocation_indivisible_attack_np_complete}
    Determining whether a restaking network has a profitable allocation-indivisible attack is NP-complete.
\end{proposition}

At first glance, it may seem that allowing for allocation-divisible attacks makes the problem easier, similarly to how searching for a Subset Sum problem would not be hard if we were allowed to take fractional values of the elements.
And indeed, when we allow allocation-divisible attacks, the previous reduction does not work, as all validators can allocate~$\frac{\sspTarget}{\sspElementSum}$ of their stake to each service, to get a profitable attack.

But, perhaps surprisingly, even when we allow for allocation-divisible attacks, the problem is NP-complete.
In the following proposition, we show a reduction from the Subset Sum problem to the problem of searching for an allocation-divisible attack.

\begin{proposition}
    \label{proposition:allocation_divisible_attack_np_complete}
    Determining whether a retaking network has a profitable allocation-divisible attack is NP-complete.
\end{proposition}

Since a network that has no profitable attack is secure, the complement of the problem we considered is verifying the security of a network;
we immediately get the following corollary.
\begin{corollary}
    Determining whether an~\emph{elastic} restaking network is secure is co-NP-complete.
\end{corollary}

Both reductions we show are in fact to an~\emph{atomic} restaking network.
So, in addition, we get that the problem of searching for attacks and the complementary problem of verifying security cannot be eased by considering atomic restaking networks alone.


\begin{figure*}[t]
    \begin{subfigure}[b]{0.32\textwidth}
        \centering
        \begin{tikzpicture}
        \pgfplotsset{height=0.6\textwidth,width=1\textwidth}
        \begin{axis}[
            xlabel={Restaking Degree},
            ylabel={Minimum Stake},
            xmin=1,
            xmax=10,
            ymin=1.75,
            ymax=3.25
        ]
        \addplot+[fill opacity=0.9, draw opacity=0.9, mark=none]
        table [col sep=comma, x=restaking_degree, y=min_stake_threshold_0.33] {data/figure3_n10.csv};
        
        \addplot+[fill opacity=0.9, draw opacity=0.9, mark=none]
        table [col sep=comma, x=restaking_degree, y=min_stake_threshold_0.50] {data/figure3_n10.csv};
        \end{axis}
        \end{tikzpicture}
        \caption{$|\allValidators| = |\allServices| = 10$.}
        \label{figure:security_analysis_sample_networks_n10}
        \Description{Graph showing stake required for cryptoeconomic security for different restaking degrees with 10 validators and services.}
    \end{subfigure}
    ~
    \begin{subfigure}[b]{0.32\textwidth}
        \centering
        \begin{tikzpicture}
        \pgfplotsset{height=0.6\textwidth,width=1\textwidth}
        \begin{axis}[
            xlabel={Restaking Degree},
            ylabel={Minimum Stake},
            xmin=1,
            xmax=11,
            ymin=1.75,
            ymax=3.25
        ]
        \addplot+[fill opacity=0.9, draw opacity=0.9, mark=none]
        table [col sep=comma, x=restaking_degree, y=min_stake_threshold_0.33] {data/figure3_n11.csv};
        
        \addplot+[fill opacity=0.9, draw opacity=0.9, mark=none]
        table [col sep=comma, x=restaking_degree, y=min_stake_threshold_0.50] {data/figure3_n11.csv};
        \end{axis}
        \end{tikzpicture}
        \caption{$|\allValidators| = |\allServices| = 11$.}
        \label{figure:security_analysis_sample_networks_n11}
        \Description{Graph showing stake required for cryptoeconomic security for different restaking degrees with 11 validators and services.}
    \end{subfigure}
    ~
    \begin{subfigure}[b]{0.32\textwidth}
        \centering
        \begin{tikzpicture}
        \pgfplotsset{height=0.6\textwidth,width=1\textwidth}
        \begin{axis}[
            legend style={at={(0.98,0.5)},anchor=east},
            legend image post style={line width=1.25pt},
            xlabel={Restaking Degree},
            ylabel={Minimum Stake},
            xmin=1,
            xmax=12,
            ymin=1.75,
            ymax=3.25
        ]
        \addplot+[fill opacity=0.9, draw opacity=0.9, mark=none]
        table [col sep=comma, x=restaking_degree, y=min_stake_threshold_0.33] {data/figure3_n12.csv};
        \addlegendentry{$\symmetricAttackThreshold = 0.33$}
        
        \addplot+[fill opacity=0.9, draw opacity=0.9, mark=none]
        table [col sep=comma, x=restaking_degree, y=min_stake_threshold_0.50] {data/figure3_n12.csv};
        \addlegendentry{$\symmetricAttackThreshold = 0.50$}
        \end{axis}
        \end{tikzpicture}
        \caption{$|\allValidators| = |\allServices| = 12$.}
        \label{figure:security_analysis_sample_networks_n12}
        \Description{Graph showing stake required for cryptoeconomic security for different restaking degrees with 12 validators and services.}
    \end{subfigure}
    \caption{Stake required for cryptoeconomic security for different restaking degrees.}
    \label{figure:security_analysis_sample_networks}
    \Description{Graphs showing stake required for cryptoeconomic security for different restaking degrees across different numbers of validators and services.}
\end{figure*}


\subsection{The Symmetric Case}
\label{section:security_analysis:symmetric_case}


Given that searching for attacks is NP-complete in the general case, we now focus on \emph{symmetric} networks where the problem becomes more tractable.
This restriction enables efficient analysis while preserving the fundamental mechanisms that determine whether restaking networks are secure.

\begin{definition}[Symmetric Network]
    A restaking network~$\networkState=(\allValidators, \allServices, \allStakes, \allAllocations, \allAttackThresholds, \allAttackPrizes)$ is \emph{symmetric} if:
    (1) All validators have equal stake, that is, for any two validators~$\validator_1, \validator_2 \in \allValidators$,~${\stake{\validator_1} = \stake{\validator_2}}$;
    (2) allocations of all validators to each service are equal, that is, for any two validators~$\validator_1, \validator_2 \in \allValidators$ and any service~$\service \in \allServices$,~$\allocation{\validator_1}{\service} = \allocation{\validator_2}{\service}$; and
    (3) all attack thresholds are equal, that is, for any two services~$\service_1, \service_2 \in \allServices$,~$\attackThreshold{\service_1} = \attackThreshold{\service_2}$.
\end{definition}
For brevity, in symmetric networks, we omit validators from the notation of the stake~$\allStakes$ and allocations to services~$\allAllocations(\service)$, and omit services from the notation of the attack thresholds~$\allAttackThresholds$.

We show a two-step reduction from an attack in a symmetric network to another simpler attack with the same prize but a (non-strictly) lower cost.
This allows us to restrict the search space of profitable attacks to those of the simpler form.
The first step is that any attack can be \emph{tightened} to use only the stake that is necessary to achieve the threshold~$\symmetricAttackThreshold$.
\begin{definition}[Tight Attack]
    Consider a symmetric restaking network~${\networkState = (\allValidators, \allServices, \allStakes, \allAllocations, \allAttackThresholds, \allAttackPrizes)}$.
    An attack~$\allAttackStakes$ is \emph{tight} if for all services~$\service \in \attackedServices$
    \begin{equation}
        \label{equation:tight_attack}
        \sum_{\validator \in \allValidators} \attackStake{\validator}{\service} = \symmetricAttackThreshold \cdot | \allValidators | \cdot \symmetricAllocation{\service} .
    \end{equation}
\end{definition}
Second, a tight attack can be \emph{consolidated} by shifting attacking stake from validators with less stake to validators with more stake until it is impossible to shift more.
\begin{definition}[Consolidated Attack]
    \label{definition:consolidated_attack}
    Consider a symmetric restaking network~${\networkState = (\allValidators, \allServices, \allStakes, \allAllocations, \allAttackThresholds, \allAttackPrizes)}$.
    Let~$\floor{\symmetricAttackThreshold |\allValidators|}$ be the integer part of~$\symmetricAttackThreshold |\allValidators|$.
    An attack~$\allAttackStakes$ is \emph{consolidated} if for all services~$\service \in \attackedServices$ it holds that for all~$i \in \left\{ 1, \ldots, \floor{\symmetricAttackThreshold |\allValidators|} \right\}$
    \begin{equation}
        \label{equation:consolidated_attack}
        \attackStake{\validator_i}{\service} = \begin{cases}
            \symmetricAllocation{\service} & \text{if } i \leq \floor{\symmetricAttackThreshold |\allValidators|} ; \\
            \left( \symmetricAttackThreshold |\allValidators| - \floor{\symmetricAttackThreshold |\allValidators|} \right) \symmetricAllocation{\service} & \text{if } i = \floor{\symmetricAttackThreshold |\allValidators|} + 1 ; \\
            0 & \text{otherwise} .
        \end{cases}
    \end{equation}
\end{definition}

Note that for each subset of services~$\allServicesAt{c}$, there is exactly one consolidated attack~$\allAttackStakesAt{c}$ for which~${\allServicesAt{c} = \attackedServicesAt{c}}$, that is, it attacks exactly the services in~$\allServicesAt{c}$. 
We can efficiently calculate the cost of~$\allAttackStakesAt{c}$ using the following proposition.
\begin{proposition}
    \label{proposition:consolidated_attack_cost}
    Let~${\networkState = (\allValidators, \allServices, \allStakes, \allAllocations, \allAttackThresholds, \allAttackPrizes)}$ be a symmetric restaking network, and let~$\allAttackStakesAt{c}$ be a consolidated attack on services~$\attackedServicesAt{c}$.
    Then, the cost of~$\allAttackStakesAt{c}$,~$\attackCost{\allAttackStakesAt{c}}$, equals
    \begin{equation}
        \floor{\symmetricAttackThreshold | \allValidators |} \cdot \min \left( \symmetricStake, \sum_{\service \in \attackedServicesAt{c}} \symmetricAllocation{\service} \right) + \min \left( \symmetricStake, \left( \symmetricAttackThreshold | \allValidators | - \floor{\symmetricAttackThreshold | \allValidators |} \right) \sum_{\service \in \attackedServicesAt{c}} \symmetricAllocation{\service} \right) .
    \end{equation}
\end{proposition}

The following proposition performs the two-step reduction on profitable attacks.
\begin{proposition}
    \label{proposition:profitable_attack_consolidated}
    If there is a profitable attack in a symmetric network, then there is a profitable attack that is consolidated.
\end{proposition}
 
We reach the following corollary stating that to check cryptoeconomic security, it suffices to consider only consolidated attacks.
\begin{corollary}
    A symmetric restaking network is cryptoeconomically secure if and only if for each subset of services~$\allServicesAt{c}$, the cost of the consolidated attack~$\allAttackStakesAt{c}$ that attacks exactly the services in~$\allServicesAt{c}$ is strictly higher than its prize.
\end{corollary}
\begin{proof}
    This follows from the Proposition~\ref{proposition:restaking_network_security}, the definition of a profitable attack and the fact that if there is a profitable attack there is also a consolidated profitable attack (Proposition~\ref{proposition:profitable_attack_consolidated}), so we can restrict our search to consolidated attacks.
\end{proof}

In general, this method has exponential complexity in the number of services, but we can significantly reduce the search space by assuming that service prizes and allocations to services are also symmetric, or that there only a few values that they can take, as we see next. 


\subsection{Sample Networks}
\label{section:security_analysis:sample_networks}


We further narrow our focus to cases where all validators allocate exactly the same amount of stake to each service, so the allocation is fully defined by the restaking degree. 
We can therefore find the minimum required stake for a given restaking degree with a binary search on the restaking degree. 

We analyze symmetric cases where the number of validators and the number of services are both~$10$,~$11$, and~$12$, and each service has a prize of~$1$ and an attack threshold~$\symmetricAttackThreshold$ of either~$1/2$ or~$1/3$. 
Fig.~\ref{figure:security_analysis_sample_networks} shows the minimum stake for cryptoeconomic security with different restaking degrees. 

When~$\symmetricAttackThreshold |\allValidators|$ is an integer, the minimum stake required for cryptoeconomic security remains constant across all restaking degrees.
Specifically, it equals the prize divided by the attack threshold--the same amount of stake each service would need in isolation.
This occurs because in a consolidated attack, exactly~$\symmetricAttackThreshold |\allValidators|$ validators can fully utilize their allocations to attack services.
When~$\symmetricAttackThreshold |\allValidators|$ is not an integer, the attack requires an additional validator who can only partially use their allocations.
At low restaking degrees, this validator cannot reach their stake limit, which increases the cost of the attack.
Then, the network is secure with a lower total stake.


\section{Theoretical Robustness Analysis}
\label{section:theoretical_robustness_analysis}


Cryptoeconomic security means that correct behavior is an equilibrium, but it could be brittle, easily destabilized by an attacker with an exogenous motivation or service faults. 
We therefore expand the game to include such scenarios, allowing us to evaluate the staking-network robustness.
We again focus on the symmetric case~(\S\ref{section:theoretical_robustness_analysis:symmetric_case}) and showcase the robustness of a few sample networks~(\S\ref{section:theoretical_robustness_analysis:results}).
We defer all proofs to Appendix~\ref{appendix:proofs_from_section_theoretical_robustness_analysis}.

We begin by presenting a computable condition for restaking network robustness.
\begin{proposition}
    \label{proposition:restaking_network_robustness}
    A restaking network~$\networkState$ is~$\adversaryBudget$-cryptoeconomically robust if and only if there exists no~${\adversaryBudget}$-costly attack.
\end{proposition}


\subsection{The Symmetric Case}
\label{section:theoretical_robustness_analysis:symmetric_case}


$\adversaryBudget$-cryptoeconomic robustness is linked to the existence of~$\adversaryBudget$-costly attacks.
But since profitable attacks are a special case of~$\adversaryBudget$-costly attacks (for~$\adversaryBudget=0$), searching for those is still NP-hard.
We thus again turn to the symmetric case.

We begin by considering cryptoeconomic robustness alone, and later consider it combined with Byzantine services.

\subsubsection{Cryptoeconomic Robustness}

The two-step reduction that we have previously used to simplify profitable attacks can also be applied to~$\adversaryBudget$-costly attacks.
\begin{proposition}
    \label{proposition:beta_costly_attack_consolidated}
    If there is a~$\adversaryBudget$-costly attack in a symmetric network, then there is a~$\adversaryBudget$-costly profitable attack that is consolidated.
\end{proposition}
 
This implies the following corollary.
\begin{corollary}
    \label{corollary:symmetric_network_cryptoeconomic_robustness}
    A symmetric network is $\adversaryBudget$-cryptoeconomically robust if and only if for each non-empty subset of services~$\allServicesAt{c}$, the cost of the consolidated attack that attacks exactly the services in~$\allServicesAt{c}$ is strictly higher than its prize plus~$\adversaryBudget$.
\end{corollary}
\begin{proof}
    This follows from the Proposition~\ref{proposition:restaking_network_robustness}, the definition of a~$\adversaryBudget$-costly attack and the fact that if there is a~$\adversaryBudget$-costly attack there is also a consolidated~$\adversaryBudget$-costly attack (Proposition~\ref{proposition:beta_costly_attack_consolidated}), so we can restrict our search to consolidated attacks.
\end{proof}

Similarly to network security, this method is exponential in the number of services, but additional assumptions can reduce the search space.

\subsubsection{Cryptoeconomic Robustness with Byzantine Services}

We now consider the combination of cryptoeconomic robustness with Byzantine robustness.
As this is an even more general problem, we again restrict our analysis to the symmetric case.
The following proposition shows that a symmetric network remains symmetric after Byzantine services cause slashing.
\begin{proposition}
    \label{proposition:symmetric_network_remains_symmetric_after_byzantine_services_cause_slashing}
    Consider a symmetric restaking network~$\networkStateAt{0} = (\allValidatorsAt{0}, \allServicesAt{0}, \allStakesAt{0}, \allAllocationsAt{0}, \allAttackThresholdsAt{0}, \allAttackPrizesAt{0})$ and a subset of Byzantine services~$\byzantineServices \subseteq \allServicesAt{0}$.
    Let~$\networkStateAt{1} = (\allValidatorsAt{1}, \allServicesAt{1}, \allStakesAt{1}, \allAllocationsAt{1}, \allAttackThresholdsAt{1}, \allAttackPrizesAt{1})$ be the restaking network that remains after the Byzantine services in~$\byzantineServices$ cause slashing.
    Then~$\networkStateAt{1}$ is symmetric.
\end{proposition}

Therefore, due to Definition~\ref{definition:overall_robust_network}, to check whether a symmetric restaking network~$\networkState$ is~$(\maxByzantineServices, \adversaryBudget)$-robust we can iterate over all possible subsets~$\byzantineServices \in \byzantineSubsets{\maxByzantineServices}
$ and get the network~$\networkState \networkAdvance \byzantineServices$ and check it is~$\adversaryBudget$-cryptoeconomically robust.
For that, we can use Corollary~\ref{corollary:symmetric_network_cryptoeconomic_robustness} since thanks to the above proposition we know that~$\networkState \networkAdvance \byzantineServices$ is symmetric.

We can again rely on some assumption to limit the number of subsets we need to consider, like that all services have the same prize and allocations or that there are only a few different possible values.

In addition, when searching for the minimum~$\adversaryBudget$ such that a network is $\adversaryBudget$-cryptoeconomically robust, we can reduce the search space even further.
The following proposition shows that when there exist 2 identical services, if one of them is Byzantine then the resulting network is less robust than the original one.
\begin{proposition}
    \label{proposition:symmetric_network_less_robust_with_byzantine_services}
    Consider a symmetric restaking network~$\networkStateAt{0}$ that has 2 identical services~$\service_1$ and~$\service_2$, meaning their attack prizes are equal and the allocation of each validator to them is identical.
    Let~$\networkStateAt{1}$ be the restaking network that remains after the slashing of one Byzantine service~$\service_1$ in~$\networkStateAt{0}$, that is, $\networkStateAt{1} = \networkStateAt{0} \networkAdvance \left\{ \service_1 \right\}$.
    If~$\networkStateAt{1}$ is~$\adversaryBudget$-cryptoeconomically robust, then~$\networkStateAt{0}$ is~$\adversaryBudget$-cryptoeconomically robust.
\end{proposition}

Then, for a restaking network~$\networkState$, if all services that can be Byzantine are identical, that is, they all have the same attack prizes and allocations to them, we get the robustness is monotonically decreasing in the number of Byzantine services.
Thus, for finding the minimal~$\adversaryBudget$ such that the network is~$(\maxByzantineServices, \adversaryBudget)$-robust, it suffices to consider only the largest subset in~$\byzantineSubsets{\maxByzantineServices}$, as we do next.


\subsection{Sample Networks}
\label{section:theoretical_robustness_analysis:results}


\begin{figure*}[t]
    \begin{subfigure}[b]{0.32\textwidth}
        \centering
        \begin{tikzpicture}
        \pgfplotsset{height=0.6\textwidth,width=1\textwidth}
        \begin{axis}[
            legend image post style={line width=1.25pt},
            legend to name=robustness_analysis_sample_networks_legend,
            legend style={legend columns=5},
            xlabel={Restaking Degree},
            ylabel={Minimum Stake},
            xmin=1,
            xmax=6,
            ymin=2.5,
            ymax=10
        ]
        \addplot+[fill opacity=0.9, draw opacity=0.9, mark=none]
        table [col sep=comma, x=restaking_degree, y=min_stake_threshold_0.00] {data/figure4_y_stake_budget_0.csv};
        \addlegendentry{$\maxByzantineServices = 0.00$}
        
        \addplot+[fill opacity=0.9, draw opacity=0.9, mark=none]
        table [col sep=comma, x=restaking_degree, y=min_stake_threshold_0.07] {data/figure4_y_stake_budget_0.csv};
        \addlegendentry{$\maxByzantineServices = 0.07$}
        
        \addplot+[fill opacity=0.9, draw opacity=0.9, mark=none]
        table [col sep=comma, x=restaking_degree, y=min_stake_threshold_0.13] {data/figure4_y_stake_budget_0.csv};
        \addlegendentry{$\maxByzantineServices = 0.13$}
        
        \addplot+[fill opacity=0.9, draw opacity=0.9, mark=none]
        table [col sep=comma, x=restaking_degree, y=min_stake_threshold_0.20] {data/figure4_y_stake_budget_0.csv};
        \addlegendentry{$\maxByzantineServices = 0.20$}
        
        \addplot+[fill opacity=0.9, draw opacity=0.9, mark=none]
        table [col sep=comma, x=restaking_degree, y=min_stake_threshold_0.27] {data/figure4_y_stake_budget_0.csv};
        \addlegendentry{$\maxByzantineServices = 0.27$}
        
        \addplot+[fill opacity=0.9, draw opacity=0.9, mark=none]
        table [col sep=comma, x=restaking_degree, y=min_stake_threshold_0.33] {data/figure4_y_stake_budget_0.csv};
        \addlegendentry{$\maxByzantineServices = 0.33$}
        
        \addplot+[fill opacity=0.9, draw opacity=0.9, mark=none]
        table [col sep=comma, x=restaking_degree, y=min_stake_threshold_0.40] {data/figure4_y_stake_budget_0.csv};
        \addlegendentry{$\maxByzantineServices = 0.40$}
        
        \addplot+[fill opacity=0.9, draw opacity=0.9, mark=none]
        table [col sep=comma, x=restaking_degree, y=min_stake_threshold_0.47] {data/figure4_y_stake_budget_0.csv};
        \addlegendentry{$\maxByzantineServices = 0.47$}
        
        \addplot+[fill opacity=0.9, draw opacity=0.9, mark=none]
        table [col sep=comma, x=restaking_degree, y=min_stake_threshold_0.53] {data/figure4_y_stake_budget_0.csv};
        \addlegendentry{$\maxByzantineServices = 0.53$}
        
        \addplot+[fill opacity=0.9, draw opacity=0.9, mark=none]
        table [col sep=comma, x=restaking_degree, y=min_stake_threshold_0.60] {data/figure4_y_stake_budget_0.csv};
        \addlegendentry{$\maxByzantineServices = 0.60$}
        \end{axis}
        \end{tikzpicture}
        \caption{Budget~$\adversaryBudget=0$ and no base service.}
        \label{figure:robustness_analysis_sample_networks_y_stake_budget_0}
        \Description{Graph showing stake required for cryptoeconomic security for different restaking degrees with 10 validators and services.}
    \end{subfigure}
    ~
    \begin{subfigure}[b]{0.32\textwidth}
        \centering
        \begin{tikzpicture}
        \pgfplotsset{height=0.6\textwidth,width=1\textwidth}
        \begin{axis}[
            xlabel={Restaking Degree},
            ylabel={Minimum Stake},
            xmin=1,
            xmax=6,
            ymin=2.5,
            ymax=10
        ]
        \addplot+[fill opacity=0.9, draw opacity=0.9, mark=none]
        table [col sep=comma, x=restaking_degree, y=min_stake_threshold_0.00] {data/figure4_y_stake_budget_1.csv};
        
        \addplot+[fill opacity=0.9, draw opacity=0.9, mark=none]
        table [col sep=comma, x=restaking_degree, y=min_stake_threshold_0.07] {data/figure4_y_stake_budget_1.csv};
        
        \addplot+[fill opacity=0.9, draw opacity=0.9, mark=none]
        table [col sep=comma, x=restaking_degree, y=min_stake_threshold_0.13] {data/figure4_y_stake_budget_1.csv};
        
        \addplot+[fill opacity=0.9, draw opacity=0.9, mark=none]
        table [col sep=comma, x=restaking_degree, y=min_stake_threshold_0.20] {data/figure4_y_stake_budget_1.csv};
        
        \addplot+[fill opacity=0.9, draw opacity=0.9, mark=none]
        table [col sep=comma, x=restaking_degree, y=min_stake_threshold_0.27] {data/figure4_y_stake_budget_1.csv};
        
        \addplot+[fill opacity=0.9, draw opacity=0.9, mark=none]
        table [col sep=comma, x=restaking_degree, y=min_stake_threshold_0.33] {data/figure4_y_stake_budget_1.csv};
        
        \addplot+[fill opacity=0.9, draw opacity=0.9, mark=none]
        table [col sep=comma, x=restaking_degree, y=min_stake_threshold_0.40] {data/figure4_y_stake_budget_1.csv};
        
        \addplot+[fill opacity=0.9, draw opacity=0.9, mark=none]
        table [col sep=comma, x=restaking_degree, y=min_stake_threshold_0.47] {data/figure4_y_stake_budget_1.csv};
        
        \addplot+[fill opacity=0.9, draw opacity=0.9, mark=none]
        table [col sep=comma, x=restaking_degree, y=min_stake_threshold_0.53] {data/figure4_y_stake_budget_1.csv};
        
        \addplot+[fill opacity=0.9, draw opacity=0.9, mark=none]
        table [col sep=comma, x=restaking_degree, y=min_stake_threshold_0.60] {data/figure4_y_stake_budget_1.csv};
        \end{axis}
        \end{tikzpicture}
        \caption{Budget~$\adversaryBudget=1$ and no base service.}
        \label{figure:robustness_analysis_sample_networks_y_stake_budget_1}
        \Description{Graph showing stake required for cryptoeconomic security for different restaking degrees with 11 validators and services.}
    \end{subfigure}
    ~
    \begin{subfigure}[b]{0.32\textwidth}
        \centering
        \begin{tikzpicture}
        \pgfplotsset{height=0.6\textwidth,width=1\textwidth}
        \begin{axis}[
            xlabel={Restaking Degree},
            ylabel={Minimum Stake},
            xmin=1,
            xmax=6,
            ymin=2.5,
            ymax=10
        ]
        \addplot+[fill opacity=0.9, draw opacity=0.9, mark=none]
        table [col sep=comma, x=restaking_degree, y=min_stake_threshold_0.00] {data/figure4_y_stake_budget_2.csv};
        
        \addplot+[fill opacity=0.9, draw opacity=0.9, mark=none]
        table [col sep=comma, x=restaking_degree, y=min_stake_threshold_0.07] {data/figure4_y_stake_budget_2.csv};
        
        \addplot+[fill opacity=0.9, draw opacity=0.9, mark=none]
        table [col sep=comma, x=restaking_degree, y=min_stake_threshold_0.13] {data/figure4_y_stake_budget_2.csv};
        
        \addplot+[fill opacity=0.9, draw opacity=0.9, mark=none]
        table [col sep=comma, x=restaking_degree, y=min_stake_threshold_0.20] {data/figure4_y_stake_budget_2.csv};
        
        \addplot+[fill opacity=0.9, draw opacity=0.9, mark=none]
        table [col sep=comma, x=restaking_degree, y=min_stake_threshold_0.27] {data/figure4_y_stake_budget_2.csv};
        
        \addplot+[fill opacity=0.9, draw opacity=0.9, mark=none]
        table [col sep=comma, x=restaking_degree, y=min_stake_threshold_0.33] {data/figure4_y_stake_budget_2.csv};
        
        \addplot+[fill opacity=0.9, draw opacity=0.9, mark=none]
        table [col sep=comma, x=restaking_degree, y=min_stake_threshold_0.40] {data/figure4_y_stake_budget_2.csv};
        
        \addplot+[fill opacity=0.9, draw opacity=0.9, mark=none]
        table [col sep=comma, x=restaking_degree, y=min_stake_threshold_0.47] {data/figure4_y_stake_budget_2.csv};
        
        \addplot+[fill opacity=0.9, draw opacity=0.9, mark=none]
        table [col sep=comma, x=restaking_degree, y=min_stake_threshold_0.53] {data/figure4_y_stake_budget_2.csv};
        
        \addplot+[fill opacity=0.9, draw opacity=0.9, mark=none]
        table [col sep=comma, x=restaking_degree, y=min_stake_threshold_0.60] {data/figure4_y_stake_budget_2.csv};
        \end{axis}
        \end{tikzpicture}
        \caption{Budget~$\adversaryBudget=2$ and no base service.}
        \label{figure:robustness_analysis_sample_networks_y_stake_budget_2}
        \Description{Graph showing stake required for cryptoeconomic security for different restaking degrees with 11 validators and services.}
    \end{subfigure}

    \begin{subfigure}[b]{0.32\textwidth}
        \centering
        \begin{tikzpicture}
        \pgfplotsset{height=0.6\textwidth,width=1\textwidth}
        \begin{axis}[
            xlabel={Restaking Degree},
            ylabel={Minimum Stake},
            xmin=1,
            xmax=6,
            ymin=2.5,
            ymax=10
        ]
        \addplot+[fill opacity=0.9, draw opacity=0.9, mark=none]
        table [col sep=comma, x=restaking_degree, y=min_stake_threshold_0.00] {data/figure4_y_stake_budget_0_base_service_10_0.33.csv};

        \addplot+[fill opacity=0.9, draw opacity=0.9, mark=none]
        table [col sep=comma, x=restaking_degree, y=min_stake_threshold_0.07] {data/figure4_y_stake_budget_0_base_service_10_0.33.csv};
        
        \addplot+[fill opacity=0.9, draw opacity=0.9, mark=none]
        table [col sep=comma, x=restaking_degree, y=min_stake_threshold_0.13] {data/figure4_y_stake_budget_0_base_service_10_0.33.csv};

        \addplot+[fill opacity=0.9, draw opacity=0.9, mark=none]
        table [col sep=comma, x=restaking_degree, y=min_stake_threshold_0.20] {data/figure4_y_stake_budget_0_base_service_10_0.33.csv};
        
        \addplot+[fill opacity=0.9, draw opacity=0.9, mark=none]
        table [col sep=comma, x=restaking_degree, y=min_stake_threshold_0.33] {data/figure4_y_stake_budget_0_base_service_10_0.33.csv};
        
        \addplot+[fill opacity=0.9, draw opacity=0.9, mark=none]
        table [col sep=comma, x=restaking_degree, y=min_stake_threshold_0.40] {data/figure4_y_stake_budget_0_base_service_10_0.33.csv};
        
        \addplot+[fill opacity=0.9, draw opacity=0.9, mark=none]
        table [col sep=comma, x=restaking_degree, y=min_stake_threshold_0.47] {data/figure4_y_stake_budget_0_base_service_10_0.33.csv};
        
        \addplot+[fill opacity=0.9, draw opacity=0.9, mark=none]
        table [col sep=comma, x=restaking_degree, y=min_stake_threshold_0.53] {data/figure4_y_stake_budget_0_base_service_10_0.33.csv};
        
        \addplot+[fill opacity=0.9, draw opacity=0.9, mark=none]
        table [col sep=comma, x=restaking_degree, y=min_stake_threshold_0.60] {data/figure4_y_stake_budget_0_base_service_10_0.33.csv};
        
        \addplot+[fill opacity=0.9, draw opacity=0.9, mark=none]
        table [col sep=comma, x=restaking_degree, y=min_stake_threshold_0.67] {data/figure4_y_stake_budget_0_base_service_10_0.33.csv};
        \end{axis}
        \end{tikzpicture}
        \caption{Budget~$\adversaryBudget=0$ with a base service.}
        \label{figure:robustness_analysis_sample_networks_y_stake_budget_0_with_base_service}
        \Description{Graph showing stake required for cryptoeconomic security for different restaking degrees with 12 validators and services.}
    \end{subfigure}
    ~
    \begin{subfigure}[b]{0.32\textwidth}
        \centering
        \begin{tikzpicture}
        \pgfplotsset{height=0.6\textwidth,width=1\textwidth}
        \begin{axis}[
            xlabel={Restaking Degree},
            ylabel={Minimum Stake},
            xmin=1,
            xmax=6,
            ymin=2.5,
            ymax=10
        ]
        \addplot+[fill opacity=0.9, draw opacity=0.9, mark=none]
        table [col sep=comma, x=restaking_degree, y=min_stake_threshold_0.00] {data/figure4_y_stake_budget_1_base_service_10_0.33.csv};
        
        \addplot+[fill opacity=0.9, draw opacity=0.9, mark=none]
        table [col sep=comma, x=restaking_degree, y=min_stake_threshold_0.07] {data/figure4_y_stake_budget_1_base_service_10_0.33.csv};
        
        \addplot+[fill opacity=0.9, draw opacity=0.9, mark=none]
        table [col sep=comma, x=restaking_degree, y=min_stake_threshold_0.13] {data/figure4_y_stake_budget_1_base_service_10_0.33.csv};
        
        \addplot+[fill opacity=0.9, draw opacity=0.9, mark=none]
        table [col sep=comma, x=restaking_degree, y=min_stake_threshold_0.20] {data/figure4_y_stake_budget_1_base_service_10_0.33.csv};
        
        \addplot+[fill opacity=0.9, draw opacity=0.9, mark=none]
        table [col sep=comma, x=restaking_degree, y=min_stake_threshold_0.27] {data/figure4_y_stake_budget_1_base_service_10_0.33.csv};
        
        \addplot+[fill opacity=0.9, draw opacity=0.9, mark=none]
        table [col sep=comma, x=restaking_degree, y=min_stake_threshold_0.33] {data/figure4_y_stake_budget_1_base_service_10_0.33.csv};
        
        \addplot+[fill opacity=0.9, draw opacity=0.9, mark=none]
        table [col sep=comma, x=restaking_degree, y=min_stake_threshold_0.40] {data/figure4_y_stake_budget_1_base_service_10_0.33.csv};
        
        \addplot+[fill opacity=0.9, draw opacity=0.9, mark=none]
        table [col sep=comma, x=restaking_degree, y=min_stake_threshold_0.47] {data/figure4_y_stake_budget_1_base_service_10_0.33.csv};
        
        \addplot+[fill opacity=0.9, draw opacity=0.9, mark=none]
        table [col sep=comma, x=restaking_degree, y=min_stake_threshold_0.53] {data/figure4_y_stake_budget_1_base_service_10_0.33.csv};
        
        \addplot+[fill opacity=0.9, draw opacity=0.9, mark=none]
        table [col sep=comma, x=restaking_degree, y=min_stake_threshold_0.60] {data/figure4_y_stake_budget_1_base_service_10_0.33.csv};
        \end{axis}
        \end{tikzpicture}
        \caption{Budget~$\adversaryBudget=1$ with a base service.}
        \label{figure:robustness_analysis_sample_networks_y_stake_budget_1_with_base_service}
        \Description{Graph showing stake required for cryptoeconomic security for different restaking degrees with 12 validators and services.}
    \end{subfigure}
    ~
    \begin{subfigure}[b]{0.32\textwidth}
        \centering
        \begin{tikzpicture}
        \pgfplotsset{height=0.6\textwidth,width=1\textwidth}
        \begin{axis}[
            xlabel={Restaking Degree},
            ylabel={Minimum Stake},
            xmin=1,
            xmax=6,
            ymin=2.5,
            ymax=10
        ]
        \addplot+[fill opacity=0.9, draw opacity=0.9, mark=none]
        table [col sep=comma, x=restaking_degree, y=min_stake_threshold_0.00] {data/figure4_y_stake_budget_2_base_service_10_0.33.csv};
        
        \addplot+[fill opacity=0.9, draw opacity=0.9, mark=none]
        table [col sep=comma, x=restaking_degree, y=min_stake_threshold_0.07] {data/figure4_y_stake_budget_2_base_service_10_0.33.csv};
        
        \addplot+[fill opacity=0.9, draw opacity=0.9, mark=none]
        table [col sep=comma, x=restaking_degree, y=min_stake_threshold_0.13] {data/figure4_y_stake_budget_2_base_service_10_0.33.csv};
        
        \addplot+[fill opacity=0.9, draw opacity=0.9, mark=none]
        table [col sep=comma, x=restaking_degree, y=min_stake_threshold_0.20] {data/figure4_y_stake_budget_2_base_service_10_0.33.csv};
        
        \addplot+[fill opacity=0.9, draw opacity=0.9, mark=none]
        table [col sep=comma, x=restaking_degree, y=min_stake_threshold_0.27] {data/figure4_y_stake_budget_2_base_service_10_0.33.csv};
        
        \addplot+[fill opacity=0.9, draw opacity=0.9, mark=none]
        table [col sep=comma, x=restaking_degree, y=min_stake_threshold_0.33] {data/figure4_y_stake_budget_2_base_service_10_0.33.csv};
        
        \addplot+[fill opacity=0.9, draw opacity=0.9, mark=none]
        table [col sep=comma, x=restaking_degree, y=min_stake_threshold_0.40] {data/figure4_y_stake_budget_2_base_service_10_0.33.csv};
        
        \addplot+[fill opacity=0.9, draw opacity=0.9, mark=none]
        table [col sep=comma, x=restaking_degree, y=min_stake_threshold_0.47] {data/figure4_y_stake_budget_2_base_service_10_0.33.csv};
        
        \addplot+[fill opacity=0.9, draw opacity=0.9, mark=none]
        table [col sep=comma, x=restaking_degree, y=min_stake_threshold_0.53] {data/figure4_y_stake_budget_2_base_service_10_0.33.csv};
        
        \addplot+[fill opacity=0.9, draw opacity=0.9, mark=none]
        table [col sep=comma, x=restaking_degree, y=min_stake_threshold_0.60] {data/figure4_y_stake_budget_2_base_service_10_0.33.csv};
        \end{axis}
        \end{tikzpicture}
        \caption{Budget~$\adversaryBudget=2$ with a base service.}
        \label{figure:robustness_analysis_sample_networks_y_stake_budget_2_with_base_service}
        \Description{Graph showing stake required for cryptoeconomic security for different restaking degrees with 12 validators and services.}
    \end{subfigure}
    \\
    \vspace{0.5em}
    \pgfplotslegendfromname{robustness_analysis_sample_networks_legend}
    \caption{Minimum stake required for $(\maxByzantineServices, \adversaryBudget)$-robustness.}
    \label{figure:robustness_analysis_sample_networks}
    \Description{Graphs showing stake required for cryptoeconomic security for different restaking degrees across different numbers of validators and services.}
\end{figure*}


\begin{figure}[t]
    \centering
    \begin{tikzpicture}
    \pgfplotsset{height=0.3\textwidth,width=0.49\textwidth}
    \begin{axis}[
        legend style={at={(0.98,0.98)},anchor=north east, legend columns=3},
        legend image post style={line width=1.25pt},
        xlabel={Fraction of Byzantine Services~$\maxByzantineServices$},
        ylabel={Adversary Budget~$\adversaryBudget$},
        xmin=0,
        xmax=0.834
    ]
    \addplot+[fill opacity=0.9, draw opacity=0.9, mark options={solid}]
    table [col sep=comma, x=robustness_threshold, y=min_budget_1.00] {data/figure5.csv};
    \addlegendentry{1.00}

    \addplot+[fill opacity=0.9, draw opacity=0.9, mark options={solid}]
    table [col sep=comma, x=robustness_threshold, y=min_budget_1.25] {data/figure5.csv};
    \addlegendentry{1.25}

    \addplot+[fill opacity=0.9, draw opacity=0.9, mark options={solid}]
    table [col sep=comma, x=robustness_threshold, y=min_budget_1.50] {data/figure5.csv};
    \addlegendentry{1.50}

    \addplot+[fill opacity=0.9, draw opacity=0.9, mark options={solid}]
    table [col sep=comma, x=robustness_threshold, y=min_budget_1.75] {data/figure5.csv};
    \addlegendentry{1.75}

    \addplot+[fill opacity=0.9, draw opacity=0.9, mark options={solid}]
    table [col sep=comma, x=robustness_threshold, y=min_budget_2.00] {data/figure5.csv};
    \addlegendentry{2.00}

    \addplot+[fill opacity=0.9, draw opacity=0.9, mark options={solid}]
    table [col sep=comma, x=robustness_threshold, y=min_budget_2.25] {data/figure5.csv};
    \addlegendentry{2.25}

    \addplot+[fill opacity=0.9, draw opacity=0.9, mark options={solid}]
    table [col sep=comma, x=robustness_threshold, y=min_budget_2.50] {data/figure5.csv};
    \addlegendentry{2.50}

    \addplot+[fill opacity=0.9, draw opacity=0.9, mark options={solid}]
    table [col sep=comma, x=robustness_threshold, y=min_budget_2.75] {data/figure5.csv};
    \addlegendentry{2.75}

    \addplot+[fill opacity=0.9, draw opacity=0.9, mark options={solid}]
    table [col sep=comma, x=robustness_threshold, y=min_budget_3.00] {data/figure5.csv};
    \addlegendentry{3.00}
    
    \end{axis}
    \end{tikzpicture}
    \caption{Failure thresholds for varying restaking degrees.}
    \label{figure:robustness_analysis_varying_restaking_degrees}
    \Description{Graph showing failure thresholds for a given network with restaking degrees varying from 1 to 3.}
\end{figure}

\begin{figure}[t]
    \centering
    \begin{tikzpicture}
    \pgfplotsset{height=0.3\textwidth,width=0.49\textwidth}
    \begin{axis}[
        legend style={at={(0.98,0.98)},anchor=north east},
        legend image post style={line width=1.25pt},
        xlabel={Fraction of Byzantine Services~$\maxByzantineServices$},
        ylabel={Adversary Budget~$\adversaryBudget$},
        xmin=0,
        xmax=0.6
    ]
    \addplot+[fill opacity=0.9, draw opacity=0.9, mark options={solid}]
    table [col sep=comma, x=robustness_threshold, y=min_budget_base_only] {data/figure6.csv};
    \addlegendentry{Base}
    
    \addplot+[fill opacity=0.9, draw opacity=0.9, mark options={solid}]
    table [col sep=comma, x=robustness_threshold, y=min_budget_no_base] {data/figure6.csv};
    \addlegendentry{No base}
    
    \addplot+[fill opacity=0.9, draw opacity=0.9, mark options={solid}]
    table [col sep=comma, x=robustness_threshold, y=min_budget_total] {data/figure6.csv};
    \addlegendentry{Combined}
    \end{axis}
    \end{tikzpicture}
    \caption{Failure thresholds for a network with or without a base service and for the base service alone.}
    \label{figure:robustness_analysis_decomposition_example}
    \Description{Graph showing failure thresholds for a network with or without a base service and for the base service alone.}
\end{figure}

The specific parameters and optimal restaking degree depend on the network parameters. We analyze concrete examples to demonstrate the trade-off between robustness to Byzantine services and to an adversary budget, and the base-service benefit from restaking.

\paragraph{Robustness tradeoff}
We consider a symmetric restaking network comprising~15 validators and~15 services, where each service has an attack threshold of~$1/3$ and an attack prize of~$1$.
We examine adversary budgets of~$0$,~$1$, and~$2$, plotting the minimum stake required for \mbox{$(\maxByzantineServices, \adversaryBudget)$-robustness} across varying restaking degrees. 
Our analysis reveals distinct optimal strategies depending on the threat model.
With no adversary budget ($\adversaryBudget=0$, Fig.~\ref{figure:robustness_analysis_sample_networks_y_stake_budget_0}), lower restaking degrees provide better robustness against Byzantine services, aligning with EigenLayer's second approach.
This is because lower restaking degrees limit stake exposure to each service, reducing damage when Byzantine services slash.
With an adversary budget of~$\adversaryBudget=1$ but no Byzantine services (Fig.~\ref{figure:robustness_analysis_sample_networks_y_stake_budget_1} and Fig.~\ref{figure:robustness_analysis_sample_networks_y_stake_budget_2}, solid blue curve), higher restaking degrees yield better security, consistent with EigenLayer's first approach.
This is because higher restaking degrees mean more stake secures each service, providing better protection against adversary budgets.
When facing both threats simultaneously (Fig.~\ref{figure:robustness_analysis_sample_networks_y_stake_budget_1} and Fig.~\ref{figure:robustness_analysis_sample_networks_y_stake_budget_2}, all other curves), we obtain a convex behavior, with the optimal restaking degree depending on the robustness goal, namely the values of~$\beta$ and~$f$.

We extend our analysis by introducing a base service with threshold~$1/3$ and prize~$10$, where all validators allocate their entire stake to this service.
The results (Figures~\ref{figure:robustness_analysis_sample_networks_y_stake_budget_0_with_base_service},~\ref{figure:robustness_analysis_sample_networks_y_stake_budget_1_with_base_service}, and~\ref{figure:robustness_analysis_sample_networks_y_stake_budget_2_with_base_service}) show similar patterns regarding optimal restaking degrees, but with higher minimum stake requirements for robustness.
Furthermore, when restaking degrees are low, since all stake is allocated to the base service, validators can only allocate a small fraction of their stake to other services, requiring more total stake to achieve robustness.
This effect vanishes at higher restaking degrees.

Furthermore, we demonstrate that tuning the restaking degree can be used to tradeoff robustness to adversary budget and to Byzantine services.
We consider the same scenario as before where each validator has~$10$ units of stake and plot the maximum adversary budget given a certain fraction of Byzantine services and a restaking degree~(Fig.~\ref{figure:robustness_analysis_varying_restaking_degrees}).

A restaking degree of~$1$ results in optimal robustness against Byzantine services, but also with the least robustness to adversary budget when the fraction of Byzantine services is low.
For other restaking degrees, the robustness to adversary budget is constant when there are only few Byzantine services, up until a certain point, where the robustness quickly collapses.
Increasing the restaking degree results in higher robustness to adversary budget when there are few Byzantine services, but also with a lower fraction of Byzantine services that the network can withstand.

Note that the lines between points in Fig.~\ref{figure:robustness_analysis_varying_restaking_degrees} are only for visual guidance.
Since the number of Byzantine services is discrete, the robustness to adversary budget is not continuous.
It is a left-continuous piecewise-constant function.
This is because increasing the maximum fraction of services allowed to be Byzantine only matters once we reach a fraction which allows one more service to be Byzantine.
In addition, due to Prop.~\ref{proposition:symmetric_network_less_robust_with_byzantine_services}, we know the function is monotonically decreasing, as we observe.
For each restaking degree, the area under its function represents its safe region, that is, values~$(\maxByzantineServices,\adversaryBudget)$ such that the restaking network is~$(\maxByzantineServices,\adversaryBudget)$-robust.

\paragraph{Base-service robustness}
In addition, we observe the difference between the networks with and without the base service.
First, the minimum stake required for the base service to be robust is
${\symmetricAttackThreshold |\allValidators| \symmetricStake < \allAttackPrizes + \adversaryBudget}$, so in our case~${5 \symmetricStake < 10 + \adversaryBudget}$.
Thus, for~$\adversaryBudget = 0$, we get that the minimum stake required for the base service to be robust is~$2$.
And indeed, the difference in stake requirements between the networks with and without the base service is~$2$ when the restaking degree is minimal.

However, with~$\adversaryBudget = 2$, we observe one of the key benefits of elastic networks:
The stake required for the combined network to be robust is lower than the stake required when the network and base service are separated.
The stake required for the base service is~$2.4$.
Consider~$\maxByzantineServices = 1/3$: the network without the base service requires~$5.4$ with its best restaking degree, while the network with the base service requires~$7.4$, which is~5\% lower than the alternative, all achieving the same robustness to Byzantine services and adversary budget.

To better illustrate the benefits for a base service we further examine this scenario, comparing the robustness of the following cases: the base service when validators have~$2.4$ units of stake, the network without the base service when validators have~$5.4$ units of stake, and the combined network when validators have the sum,~$7.8$ units of stake~(Fig.~\ref{figure:robustness_analysis_decomposition_example}).
We see that when the base service is part of the combined network it enjoys higher robustness against an adversary, as long as the number of Byzantine services is not too high.

When not too many services are Byzantine, the combined network has more stake securing the base service, requiring more stake to attack and thus a higher adversary budget to reimburse losses.

While we showcase the trade-off and synergistic effect in a specific symmetric setting, these effects apply more broadly.
We only use these symmetric networks as a simple setting to isolate and clearly demonstrate the fundamental mechanisms that underlie restaking network robustness.


\section{Robustness Analysis with Mixed-Integer Programming}
\label{section:mip_robustness_analysis}


\begin{figure*}[t]
    \begin{subfigure}[b]{0.32\textwidth}
        \centering
        \begin{tikzpicture}
        \pgfplotsset{height=0.6\textwidth,width=1\textwidth}
        \begin{axis}[
            legend image post style={line width=1.25pt},
            legend to name=mip_analysis_sample_networks_legend,
            legend style={legend columns=6},
            xlabel={Restaking Degree},
            ylabel={Minimum Stake},
            xmin=1,
            xmax=3,
            ymin=2.5,
            ymax=12
        ]
        \addplot+[fill opacity=0.9, draw opacity=0.9, mark=none]
        table [col sep=comma, x=restaking_degree, y=min_stake_threshold_0.0_with_milp] {data/figure7_budget_0.csv};
        \addlegendentry{$\maxByzantineServices = 0.00$, MIP}

        \addplot+[fill opacity=0.9, draw opacity=0.9, mark=none]
        table [col sep=comma, x=restaking_degree, y=min_stake_threshold_0.33333333333333337_with_milp] {data/figure7_budget_0.csv};
        \addlegendentry{$\maxByzantineServices = 0.33$, MIP}

        \addplot+[fill opacity=0.9, draw opacity=0.9, mark=none]
        table [col sep=comma, x=restaking_degree, y=min_stake_threshold_0.6666666666666667_with_milp] {data/figure7_budget_0.csv};
        \addlegendentry{$\maxByzantineServices = 0.67$, MIP}

        \addplot+[fill opacity=0.9, draw opacity=0.9, mark=none]
        table [col sep=comma, x=restaking_degree, y=min_stake_threshold_0.0_without_milp] {data/figure7_budget_0.csv};
        \addlegendentry{$\maxByzantineServices = 0.00$}

        \addplot+[fill opacity=0.9, draw opacity=0.9, mark=none]
        table [col sep=comma, x=restaking_degree, y=min_stake_threshold_0.33333333333333337_without_milp] {data/figure7_budget_0.csv};
        \addlegendentry{$\maxByzantineServices = 0.33$}

        \addplot+[fill opacity=0.9, draw opacity=0.9, mark=none]
        table [col sep=comma, x=restaking_degree, y=min_stake_threshold_0.6666666666666667_without_milp] {data/figure7_budget_0.csv};
        \addlegendentry{$\maxByzantineServices = 0.67$}
        \end{axis}
        \end{tikzpicture}
        \caption{$\adversaryBudget = 0$.}
        \label{figure:mip_analysis_sample_networks_b0}
        \Description{Graph showing minimum stake required for different restaking degrees and adversary budgets.}
    \end{subfigure}
    ~
    \begin{subfigure}[b]{0.32\textwidth}
        \centering
        \begin{tikzpicture}
        \pgfplotsset{height=0.6\textwidth,width=1\textwidth}
        \begin{axis}[
            xlabel={Restaking Degree},
            ylabel={Minimum Stake},
            xmin=1,
            xmax=3,
            ymin=2.5,
            ymax=12
        ]
        \addplot+[fill opacity=0.9, draw opacity=0.9, mark=none]
        table [col sep=comma, x=restaking_degree, y=min_stake_threshold_0.0_with_milp] {data/figure7_budget_1.csv};

        \addplot+[fill opacity=0.9, draw opacity=0.9, mark=none]
        table [col sep=comma, x=restaking_degree, y=min_stake_threshold_0.33333333333333337_with_milp] {data/figure7_budget_1.csv};

        \addplot+[fill opacity=0.9, draw opacity=0.9, mark=none]
        table [col sep=comma, x=restaking_degree, y=min_stake_threshold_0.6666666666666667_with_milp] {data/figure7_budget_1.csv};

        \addplot+[fill opacity=0.9, draw opacity=0.9, mark=none]
        table [col sep=comma, x=restaking_degree, y=min_stake_threshold_0.0_without_milp] {data/figure7_budget_1.csv};

        \addplot+[fill opacity=0.9, draw opacity=0.9, mark=none]
        table [col sep=comma, x=restaking_degree, y=min_stake_threshold_0.33333333333333337_without_milp] {data/figure7_budget_1.csv};

        \addplot+[fill opacity=0.9, draw opacity=0.9, mark=none]
        table [col sep=comma, x=restaking_degree, y=min_stake_threshold_0.6666666666666667_without_milp] {data/figure7_budget_1.csv};
        \end{axis}
        \end{tikzpicture}
        \caption{$\adversaryBudget=1$.}
        \label{figure:mip_analysis_sample_networks_b1}
        \Description{Graph showing minimum stake required for different restaking degrees and adversary budgets.}
    \end{subfigure}
    ~
    \begin{subfigure}[b]{0.32\textwidth}
        \centering
        \begin{tikzpicture}
        \pgfplotsset{height=0.6\textwidth,width=1\textwidth}
        \begin{axis}[
            xlabel={Restaking Degree},
            ylabel={Minimum Stake},
            xmin=1,
            xmax=3,
            ymin=2.5,
            ymax=12
        ]
        \addplot+[fill opacity=0.9, draw opacity=0.9, mark=none]
        table [col sep=comma, x=restaking_degree, y=min_stake_threshold_0.0_with_milp] {data/figure7_budget_2.csv};

        \addplot+[fill opacity=0.9, draw opacity=0.9, mark=none]
        table [col sep=comma, x=restaking_degree, y=min_stake_threshold_0.33333333333333337_with_milp] {data/figure7_budget_2.csv};

        \addplot+[fill opacity=0.9, draw opacity=0.9, mark=none]
        table [col sep=comma, x=restaking_degree, y=min_stake_threshold_0.6666666666666667_with_milp] {data/figure7_budget_2.csv};

        \addplot+[fill opacity=0.9, draw opacity=0.9, mark=none]
        table [col sep=comma, x=restaking_degree, y=min_stake_threshold_0.0_without_milp] {data/figure7_budget_2.csv};

        \addplot+[fill opacity=0.9, draw opacity=0.9, mark=none]
        table [col sep=comma, x=restaking_degree, y=min_stake_threshold_0.33333333333333337_without_milp] {data/figure7_budget_2.csv};

        \addplot+[fill opacity=0.9, draw opacity=0.9, mark=none]
        table [col sep=comma, x=restaking_degree, y=min_stake_threshold_0.6666666666666667_without_milp] {data/figure7_budget_2.csv};
        \end{axis}
        \end{tikzpicture}
        \caption{$\adversaryBudget=2$.}
        \label{figure:mip_analysis_sample_networks_b2}
        \Description{Graph showing minimum stake required for different restaking degrees and adversary budgets.}
    \end{subfigure}
    \\
    \vspace{0.5em}
    \pgfplotslegendfromname{mip_analysis_sample_networks_legend}
    \caption{Minimum stake required for $(\maxByzantineServices, \adversaryBudget)$-robustness.}
    \label{figure:mip_analysis_sample_networks}
    \Description{Graph showing minimum stake required for different restaking degrees and adversary budgets.}
\end{figure*}

\begin{figure*}[t]
    \begin{subfigure}[b]{0.32\textwidth}
        \centering
        \begin{tikzpicture}
        \pgfplotsset{height=0.6\textwidth,width=1\textwidth}
        \begin{axis}[
            legend image post style={line width=1.25pt},
            legend to name=mip_analysis_sample_networks_base_legend,
            legend style={legend columns=6},
            xlabel={Restaking Degree},
            ylabel={Minimum Stake},
            xmin=1,
            xmax=3,
            ymin=5,
            ymax=30
        ]
        \addplot+[fill opacity=0.9, draw opacity=0.9, mark=none]
        table [col sep=comma, x=restaking_degree, y=min_stake_threshold_0.00] {data/figure8_y_stake_base_service_10_0.50_loss_threshold_0.csv};
        \addlegendentry{$\maxByzantineServices = 0.00$}

        \addplot+[fill opacity=0.9, draw opacity=0.9, mark=none]
        table [col sep=comma, x=restaking_degree, y=min_stake_threshold_0.10] {data/figure8_y_stake_base_service_10_0.50_loss_threshold_0.csv};
        \addlegendentry{$\maxByzantineServices = 0.10$}

        \addplot+[fill opacity=0.9, draw opacity=0.9, mark=none]
        table [col sep=comma, x=restaking_degree, y=min_stake_threshold_0.21] {data/figure8_y_stake_base_service_10_0.50_loss_threshold_0.csv};
        \addlegendentry{$\maxByzantineServices = 0.21$}

        \addplot+[fill opacity=0.9, draw opacity=0.9, mark=none]
        table [col sep=comma, x=restaking_degree, y=min_stake_threshold_0.31] {data/figure8_y_stake_base_service_10_0.50_loss_threshold_0.csv};
        \addlegendentry{$\maxByzantineServices = 0.31$}
        \end{axis}
        \end{tikzpicture}
        \caption{$\adversaryBudget=0$.}
        \label{figure:mip_analysis_sample_networks_base_budget_0}
        \Description{Graph showing minimum stake required for different restaking degrees and adversary budgets with a base service.}
    \end{subfigure}
    ~
    \begin{subfigure}[b]{0.32\textwidth}
        \centering
        \begin{tikzpicture}
        \pgfplotsset{height=0.6\textwidth,width=1\textwidth}
        \begin{axis}[
            xlabel={Restaking Degree},
            ylabel={Minimum Stake},
            xmin=1,
            xmax=3,
            ymin=5,
            ymax=30
        ]
        \addplot+[fill opacity=0.9, draw opacity=0.9, mark=none]
        table [col sep=comma, x=restaking_degree, y=min_stake_threshold_0.00] {data/figure8_y_stake_base_service_10_0.50_loss_threshold_1.csv};

        \addplot+[fill opacity=0.9, draw opacity=0.9, mark=none]
        table [col sep=comma, x=restaking_degree, y=min_stake_threshold_0.10] {data/figure8_y_stake_base_service_10_0.50_loss_threshold_1.csv};

        \addplot+[fill opacity=0.9, draw opacity=0.9, mark=none]
        table [col sep=comma, x=restaking_degree, y=min_stake_threshold_0.21] {data/figure8_y_stake_base_service_10_0.50_loss_threshold_1.csv};

        \addplot+[fill opacity=0.9, draw opacity=0.9, mark=none]
        table [col sep=comma, x=restaking_degree, y=min_stake_threshold_0.31] {data/figure8_y_stake_base_service_10_0.50_loss_threshold_1.csv};
        \end{axis}
        \end{tikzpicture}
        \caption{$\adversaryBudget=1$.}
        \label{figure:mip_analysis_sample_networks_base_budget_1}
        \Description{Graph showing minimum stake required for different restaking degrees and adversary budgets with a base service.}
    \end{subfigure}
    ~
    \begin{subfigure}[b]{0.32\textwidth}
        \centering
        \begin{tikzpicture}
        \pgfplotsset{height=0.6\textwidth,width=1\textwidth}
        \begin{axis}[
            xlabel={Restaking Degree},
            ylabel={Minimum Stake},
            xmin=1,
            xmax=3,
            ymin=5,
            ymax=30
        ]
        \addplot+[fill opacity=0.9, draw opacity=0.9, mark=none]
        table [col sep=comma, x=restaking_degree, y=min_stake_threshold_0.00] {data/figure8_y_stake_base_service_10_0.50_loss_threshold_2.csv};

        \addplot+[fill opacity=0.9, draw opacity=0.9, mark=none]
        table [col sep=comma, x=restaking_degree, y=min_stake_threshold_0.10] {data/figure8_y_stake_base_service_10_0.50_loss_threshold_2.csv};

        \addplot+[fill opacity=0.9, draw opacity=0.9, mark=none]
        table [col sep=comma, x=restaking_degree, y=min_stake_threshold_0.21] {data/figure8_y_stake_base_service_10_0.50_loss_threshold_2.csv};

        \addplot+[fill opacity=0.9, draw opacity=0.9, mark=none]
        table [col sep=comma, x=restaking_degree, y=min_stake_threshold_0.31] {data/figure8_y_stake_base_service_10_0.50_loss_threshold_2.csv};
        \end{axis}
        \end{tikzpicture}
        \caption{$\adversaryBudget=2$.}
        \label{figure:mip_analysis_sample_networks_base_budget_2}
        \Description{Graph showing minimum stake required for different restaking degrees and adversary budgets with a base service.}
    \end{subfigure}
    \\
    \vspace{0.5em}
    \pgfplotslegendfromname{mip_analysis_sample_networks_base_legend}
    \caption{Minimum stake required for $(\maxByzantineServices, \adversaryBudget)$-robustness with a base service.}
    \label{figure:mip_analysis_sample_networks_base}
    \Description{Graph showing minimum stake required for different restaking degrees and adversary budgets with a base service.}
\end{figure*}


Despite the hardness results we have shown, we can still empirically analyze the robustness in the general case for small restaking networks.
For this, we utilize Mixed-Integer Programming~(\S\ref{section:mip_robustness_analysis:preliminary}).
We introduce 2 programs: one for finding the maximum budget~$\adversaryBudget$ against which a network is $\adversaryBudget$-cryptoeconomically robust, and one for finding the maximum fraction of Byzantine services a network can withstand given an adversary budget.
We defer the details on their design and implementation to Appendix~\ref{section:mip_appendix}.
In this section, we present results for some sample networks~(\S\ref{section:mip_robustness_analysis:results}).


\subsection{Background: Mixed-Integer Programming}
\label{section:mip_robustness_analysis:preliminary}


A mixed-integer program~(MIP) is a linear optimization problem with both integer and real-valued variables~\cite{junger200950}.
It comprises a constraint matrix~$A \in \realNumbers^{m \times n}$ and vector~$b \in \realNumbers^m$, an objective vector~$c \in \realNumbers^n$, and a set~$I \subseteq \left\{ 1, \ldots, n \right\}$ of indices of integer variables.
The program is then:
\begin{equation}
    \min_{x \in \realNumbers^n} \left\{ c^\top x \middle| A x \leq b, x_i \in \integers \text{ for all } i \in I \right\} .
\end{equation}


\subsection{Sample Networks}
\label{section:mip_robustness_analysis:results}

To validate the MIPs we compare their results with our theoretical approach in a symmetric network where all validators allocate the same amount to all services.
This implies that the restaking degree fully determines validators' allocations.
Then, given an adversary budget~$\adversaryBudget$ and a maximum fraction of Byzantine services~$\maxByzantineServices$, we can calculate the minimum stake required for $(\maxByzantineServices, \adversaryBudget)$-robustness using the previous MIPs.
We use the cryptoeconomic robustness MIP if~$\maxByzantineServices = 0$ and use the budget-and-Byzantine robustness MIP if~$\maxByzantineServices > 0$.

Fig.~\ref{figure:mip_analysis_sample_networks} shows the results using both of our approaches for a restaking network of~$3$ validators and~$3$ services where the attack threshold for all services is~$1/3$ and the attack prize is~$1$.
As expected, for~$\adversaryBudget \in \{0, 1, 2\}$ and ~$\maxByzantineServices \in \{0, 1/3, 2/3\}$, the MIPs yield the same results as our theoretical approach.

Next, we turn to a network that our theoretical approach could not analyze.
Again, we assume that validators' allocations to all services are equal so the restaking degree determines the allocations.

We start with the same network with~$3$ services,~$3$ validators, attack thresholds of~$1/3$ and attack prizes of~$1$, and add a base service that all validators are maximally allocated to.
The base service has a prize of~$10$ and a threshold of~$1/2$.
Fig.~\ref{figure:mip_analysis_sample_networks_base} shows the minimum stake required for \mbox{$(\maxByzantineServices, \adversaryBudget)$-robustness} for~${\adversaryBudget \in \{0, 1, 2\}}$.

We again observe that a balanced restaking degree results in less stake required for robustness.
But, interestingly, in some cases, we see that the minimum required stake for~$\maxByzantineServices=1/3$ and~$\maxByzantineServices=1/2$ coincide.
Perhaps because of a similar effect we observed previously in the security analysis where the number of validators times the threshold is not an integer resulting in attacks that cost more to the one validator who is not consolidated.


\section{Incentives for a Target Restaking Degree}
\label{section:incentives_for_a_target_restaking_degree}


Having shown that elastic restaking networks with a properly tuned restaking degree are more robust than atomic restaking networks, we now turn our attention to incentivizing the optimal restaking degree.
We first present a scheme for service rewards to achieve a target network-wide restaking degree~$\targetRestakingDegree$~(\S\ref{section:incentives_for_a_target_restaking_degree:service_rewards}).
We then model the validators' choices of allocations to services under this scheme as a game~(\S\ref{section:incentives_for_a_target_restaking_degree:network_formation_game}).
Lastly, we analyze the game and find a Nash equilibrium in which validators allocate their stake such that their restaking degree is equal to~$\targetRestakingDegree$~(\S\ref{section:incentives_for_a_target_restaking_degree:nash_equilibrium}).


\subsection{Service Rewards}
\label{section:incentives_for_a_target_restaking_degree:service_rewards}

In current restaking networks like~\citet{eigenlayer2024restaking}, each service~$\service$ has a \emph{reward pool}~$\serviceReward{\service}$.
Formally, denote by~$\allServiceRewards$ the reward pools of all services, namely,~$\allServiceRewards: \allServices \to \positiveRealNumbers$.
Each service's reward pool is distributed to validators proportionally to their allocations to the service.
The reward of a validator~$\validator$ for a service~$\service$ is given by
\begin{equation}
    \validatorReward{\validator}{\service} = \frac{\allocation{\validator}{\service}}{\sum_{\validator' \in \allValidators} \allocation{\validator'}{\service}} \cdot \serviceReward{\service} .
\end{equation}

To achieve a target restaking degree~$\targetRestakingDegree$, we propose a scheme that rewards only validators adhering to the target restaking degree;
Formally, the reward of a validator~$\validator$ for a service~$\service$ is given by
\begin{equation}
    \label{equation:incentives_for_a_target_restaking_degree:validator_reward}
    \validatorReward{\validator}{\service}
    = \begin{cases}
        \frac{\allocation{\validator}{\service}}{\sum_{\validator' \in \allValidators} \allocation{\validator'}{\service}} \cdot \serviceReward{\service} & \text{if } \restakingDegree{\validator} \leq \targetRestakingDegree , \\
        0 & \text{otherwise} .
    \end{cases}
\end{equation}
When~$\targetRestakingDegree \geq |\allServices|$, this scheme is equivalent to the current proportional reward scheme, since no validator can exceed this restaking degree, and thus all validators satisfy the condition for receiving rewards.

Using this scheme we disincentivize allocations higher than the desired degree.
A potential alternative would have been to simply disallow allocations higher than the desired degree by ejecting or ignoring validators that exceed it.
However, such a mechanism suffers from an important drawback when it interacts with the robustness game:
Once slashing due to a Byzantine service occurs, the restaking degree of some validators will increase and may surpass the allowed limit.
Ignoring such validators will result in further loss of stake in the network.
We choose to only disincentivize over allocation alone to avoid this issue.


\subsection{Network Formation Game}
\label{section:incentives_for_a_target_restaking_degree:network_formation_game}


We analyze the network formation under the proposed reward scheme as a strategic game.
First, assume the following are fixed: the set of validators~$\allValidators$, the set of services~$\allServices$, validators' stakes~$\allStakes$, and the service reward pools~$\allServiceRewards$.

The set of players is the set of validators~$\allValidators$.
Each validator~$\validator$ chooses an allocation~$\allocation{\validator}{\service}$ for each service~$\service \in \allServices$.
So,~$\allAllocations$ specifies the strategy profile of all validators.
The utility of a validator~$\validator$ for a given strategy profile~$\allAllocations$ is the sum of rewards they receive from all services, namely,
\begin{equation}
    \label{equation:incentives_for_a_target_restaking_degree:validator_utility}
    \validatorUtility{\validator}{\allAllocations}
    = \sum_{\service \in \allServices} \validatorReward{\validator}{\service}
    \underset{\eqref{equation:incentives_for_a_target_restaking_degree:validator_reward}} = \begin{cases}
        \sum_{\service \in \allServices} \frac{\allocation{\validator}{\service} \cdot \serviceReward{\service}}{\sum_{\validator' \in \allValidators} \allocation{\validator'}{\service}} & \text{if } \restakingDegree{\validator} \leq \targetRestakingDegree , \\
        0 & \text{otherwise} .
    \end{cases}
\end{equation}


\subsection{Nash Equilibrium}
\label{section:incentives_for_a_target_restaking_degree:nash_equilibrium}


We analyze the game and show there exists a Nash equilibrium where validators allocate their stake such that their restaking degree is~$\targetRestakingDegree$.

\begin{theorem}
    \label{theorem:incentives_for_a_target_restaking_degree:nash_equilibrium}
    Assume that for each service~$\service \in \allServices$, $\serviceReward{\service} > 0$ and~$\targetRestakingDegree \cdot \frac{\serviceReward{\service}}{\sum_{\service' \in \allServices} \serviceReward{\service'}} \leq 1$.
    Then, the strategy profile
    \begin{equation}
        \nashAllocation{\validator}{\service} = \targetRestakingDegree \cdot \frac{\serviceReward{\service}}{\sum_{\service' \in \allServices} \serviceReward{\service'}} \cdot \stake{\validator}
    \end{equation}
    is a Nash equilibrium, and it results in a restaking degree of~$\targetRestakingDegree$.
\end{theorem}
We defer the proof to Appendix~\ref{appendix:proofs_from_section_incentives_for_a_target_restaking_degree}.

This equilibrium holds when for each service~$\service$, ${\targetRestakingDegree \cdot \frac{\serviceReward{\service}}{\sum_{\service' \in \allServices} \serviceReward{\service'}} \leq 1}$.
That is, there doesn't exist a service that gives a reward that is so high compared to the others such that a validator would want to allocate more than~100\% of their stake to it.


\section{Conclusion}
\label{section:conclusion}


We introduced Elastic Restaking Networks, where in case of service failure validators' stakes are stretched among the remaining services. 
We showed that proving whether there is an attack against the network is in general an NP-complete problem, but it can be efficiently solved in symmetric cases. 
This has allowed us to find the restaking degree where the network is most robust against Byzantine service faults and against an adversary with a set budget.
While our symmetric analysis provides valuable insights into fundamental mechanisms, the full complexity of asymmetric networks remains to be explored. 
This analysis can be used directly to deploy secure restaking networks; we provide a mechanism for the system designer to incentivize validators to allocate at a target restaking degree. 

Our results give rise to several questions for future work. 
One is finding the optimal slashing function, that is, how much to penalize a validator if they use the same stake to attack multiple services. 
Intuitively, this should be a monotonically increasing function, and if it is submodular then Byzantine faults are less effective, but attacks become cheaper. 
Another question is whether the mechanism design that incentivizes a target restaking degree can be decentralized. 

While we defer these questions to future work, our results already show that elastic restaking achieves better robustness than existing schemes, and in particular can improve the security of a base-service underlying blockchain. 

\begin{acks}
    This work was supported in part by the Avalanche Foundation and by IC3.
\end{acks}

\bibliographystyle{ACM-Reference-Format}
\balance
\bibliography{ref}

\appendix


\section{Proofs Deferred from Section~\ref{section:model:elastic_restaking_networks}}
\label{appendix:proofs_from_section_elastic_restaking_networks_are_more_expressive}


\begin{proposition}[Proposition~\ref{proposition:expressiveness_of_atomic_networks} restated]
    \label{proposition:expressiveness_of_atomic_networks_appendix}
    Let~$x \in \positiveRealNumbers$.
    There exists no atomic restaking network~${\networkState = (\allValidators, \allServices, \allStakes, \allAllocations, \allAttackThresholds, \allAttackPrizes)}$ that satisfies the following conditions:
    \begin{enumerate}
        \item\label{proposition:expressiveness_of_atomic_networks_stake_condition} The total stake in the network is less than~$x$ times the number of services,
        \item\label{proposition:expressiveness_of_atomic_networks_allocation_condition_before_slashing} Each service has exactly~$x$ units of stake allocated to it, and
        \item\label{proposition:expressiveness_of_atomic_networks_allocation_condition_after_slashing} After any service fails and slashes its allocated stake, each remaining service maintains exactly~$x$ units of stake.
    \end{enumerate}
\end{proposition}
\begin{proof}
    Assume towards contradiction that such an atomic network~$\networkState$ exists.
    Due to Condition~\ref{proposition:expressiveness_of_atomic_networks_stake_condition}, we have
    \begin{equation}
        \label{equation:expressiveness_of_atomic_networks_proof_stake_condition}
        \sum_{\validator \in \allValidators} \stake{\validator} < x \cdot |\allServices| ,
    \end{equation}
    and due to Condition~\ref{proposition:expressiveness_of_atomic_networks_allocation_condition_before_slashing}, we have that for any service~$\service \in \allServices$,
    \begin{equation}
        \label{equation:expressiveness_of_atomic_networks_proof_allocation_condition_before_slashing}
        \sum_{\validator \in \allValidators} \allocation{\validator}{\service} = x .
    \end{equation}

    For any service~$\service \in \allServices$ that fails, denote by~$\allValidators_{\service}$ the set of validators with stake allocated to~$\service$, that is, $\allValidators_{\service} = \left\{ \validator \in \allValidators \middle| \allocation{\validator}{\service} > 0 \right\}$.
    Since this is an atomic network, each validator~$\validator \in \allValidators_{\service}$ must allocate their entire stake to~$\service$, and if that is the case, they will lose all stake when~$\service$ fails.
    So, due to Condition~\ref{proposition:expressiveness_of_atomic_networks_allocation_condition_after_slashing}, for all services~$\service' \in \allServices \setminus \{\service\}$, the sum of allocations for all other validators must be~$x$:
    \begin{equation}
        \label{equation:expressiveness_of_atomic_networks_proof_allocation_condition_after_slashing}
        \forall \service' \in \allServices \setminus \{\service\}: \sum_{\validator \in \allValidators \setminus \allValidators_{\service}} \allocation{\validator}{\service'} = x .
    \end{equation}

    Subtracting~\eqreft{equation:expressiveness_of_atomic_networks_proof_allocation_condition_after_slashing} from~\eqreft{equation:expressiveness_of_atomic_networks_proof_allocation_condition_before_slashing}, we get that for any service~$\service' \in \allServices \setminus \{\service\}$,
    \begin{align}
        \sum_{\validator \in \allValidators} \allocation{\validator}{\service'} - \sum_{\validator \in \allValidators \setminus \allValidators_{\service}} \allocation{\validator}{\service'} &= 0 ; \\
        \sum_{\validator \in \allValidators_{\service}} \allocation{\validator}{\service'} &= 0 . 
    \end{align}
    Since this is the sum of non-negative values, for each~$\service \in \allServices$, each~$\service' \in \allServices \setminus \{\service\}$, and each~$\validator \in \allValidators_{\service}$,~${\allocation{\validator}{\service'} = 0}$.

    Assume towards a contradiction that there exists a validator~$\validator$ that is in two different sets,~$\allValidators_{\service}$ and~$\allValidators_{\service'}$.
    As we just showed, it must be that~$\allocation{\validator}{\service'} = 0$.
    But because~$\validator \in \allValidators_{\service'}$, we must also have~$\allocation{\validator}{\service'} > 0$, which is a contradiction.
    Therefore, the sets~$\left\{ \allValidators_{\service} \right\}_{\service \in \allServices}$ must be pairwise disjoint:
    \begin{equation}
        \label{equation:expressiveness_of_atomic_networks_proof_disjoint_sets}
        \forall \service, \service' \in \allServices: \allValidators_{\service} \cap \allValidators_{\service'} = \emptyset .
    \end{equation}
    And in addition, since each~$\allValidators_{\service}$ is a subset of~$\allValidators$, we have that
    \begin{equation}
        \label{equation:expressiveness_of_atomic_networks_proof_disjoint_sets_subset}
        \bigcup_{\service \in \allServices} \allValidators_{\service} \subseteq \allValidators .
    \end{equation}

    Using the fact that the network is atomic and the definition of~$\allValidators_{\service}$, we can develop~\eqreft{equation:expressiveness_of_atomic_networks_proof_allocation_condition_before_slashing} to get that for any service~$\service \in \allServices$,
    \begin{multline}
        \label{equation:expressiveness_of_atomic_networks_proof_allocation_condition_before_slashing_expanded}
        x
        \underset{\eqref{equation:expressiveness_of_atomic_networks_proof_allocation_condition_before_slashing}}{=}         \sum_{\validator \in \allValidators} \allocation{\validator}{\service}
        = \sum_{\validator \in \allValidators \setminus \allValidators_{\service}} \allocation{\validator}{\service} + \sum_{\validator \in \allValidators_{\service}} \allocation{\validator}{\service} \\
        = \sum_{\validator \in \allValidators \setminus \allValidators_{\service}} 0 + \sum_{\validator \in \allValidators_{\service}} \stake{\validator}
        = \sum_{\validator \in \allValidators_{\service}} \stake{\validator} .
    \end{multline}

    Now, we are ready to show that the total stake in the network is at least~$x \cdot |\allServices|$.
    We use the fact that the sets~$\left\{ \allValidators_{\service} \right\}_{\service \in \allServices}$ are pairwise disjoint to obtain:
    \begin{equation}
        \sum_{\validator \in \allValidators} \stake{\validator}
        \underset{\eqref{equation:expressiveness_of_atomic_networks_proof_disjoint_sets_subset}}{\geq} \sum_{\validator \in \bigcup_{\service \in \allServices} \allValidators_{\service}} \stake{\validator}
        \underset{\eqref{equation:expressiveness_of_atomic_networks_proof_disjoint_sets}}{=} \sum_{\service \in \allServices} \sum_{\validator \in \allValidators_{\service}} \stake{\validator}
        \underset{\eqref{equation:expressiveness_of_atomic_networks_proof_allocation_condition_before_slashing_expanded}}{=} \sum_{\service \in \allServices} x
        = x \cdot |\allServices| .
    \end{equation}
    But this contradicts~\eqreft{equation:expressiveness_of_atomic_networks_proof_stake_condition}.
    Therefore, no such atomic network~$\networkState$ can exist.
\end{proof}


\section{Proofs Deferred from Section~\ref{section:security_analysis}}
\label{appendix:proofs_from_section_security_analysis}


\begin{proposition}[Proposition~\ref{proposition:restaking_network_security} restated]
    \label{proposition:restaking_network_security_appendix}
    A restaking network~$\networkState$ is cryptoeconomically secure if and only if there exists no profitable attack.
\end{proposition}
\begin{proof}
    We prove the proposition in two directions.
    
    \paragraphEmph{First direction}
    Assume that the network~$\networkState$ is cryptoeconomically secure.
    By definition, the strategy profile~$\allAttackStakes_0$, where for all~$\validator \in \allValidators$ and all~$\service \in \allServices$,~$\attackStake{\validator}{\service} = 0$, is a strong Nash equilibrium and under it there are no attacked services.
    We will show this implies that there is no profitable attack.

    First, note that due to~\eqreft{equation:validator_attack_cost}, for all validators~$\validator \in \allValidators$ and attacks~$\allAttackStakes \in \strategyProfileSecurityGame$,~$\validatorAttackCost{\validator}{\allAttackStakes} \geq 0$.
    And due to~\eqreft{equation:total_attack_cost},
    \begin{equation}
        \label{equation:proof_first_direction_total_cost_more_than_validator_cost}
        \attackCost{\allAttackStakes} \geq \validatorAttackCost{\validator}{\allAttackStakes} \geq 0 .
    \end{equation}

    The cost of the attack is
    \begin{multline}
        \label{equation:proof_first_direction_attack_cost_0}
        \attackCost{\allAttackStakes_0}
        \underset{\eqref{equation:total_attack_cost}}{=} \sum_{\validator \in \allValidators} \validatorAttackCost{\validator}{\allAttackStakes_0}
        \underset{\eqref{equation:validator_attack_cost}}{=} \sum_{\validator \in \allValidators} \min \left( \stake{\validator}, \sum_{\service \in \allServices} \attackStakeAt{0}{\validator}{\service} \right) \\
        = \sum_{\validator \in \allValidators} \min \left( \stake{\validator}, \sum_{\service \in \allServices} 0 \right)
        = 0 .
    \end{multline}
    The utility of~$\validator$ under~$\allAttackStakes_0$ is
    \begin{multline}
        \label{equation:proof_first_direction_validator_utility_0}
        \validatorUtilitySecurityGame{\validator}{\allAttackStakes_0}
        \underset{\eqref{equation:validator_utility_security_game}}{=} \validatorPrizeShareSecurityGame{\validator}{\allAttackStakes_0} \cdot \totalAttackPrize{\allAttackStakesAt{0}} - \validatorAttackCost{\validator}{\allAttackStakes_0} \\
        \underset{\eqref{equation:validator_prize_share_security_game}}{=} \begin{cases}
            \frac{\validatorAttackCost{\validator}{\allAttackStakes_0}}{\attackCost{\allAttackStakes_0}} \cdot \totalAttackPrize{\allAttackStakesAt{0}} - \validatorAttackCost{\validator}{\allAttackStakes_0} & \text{if } \attackCost{\allAttackStakes_0} > 0 ; \\
            \frac{1}{|\allValidators|} \cdot \totalAttackPrize{\allAttackStakesAt{0}} - \validatorAttackCost{\validator}{\allAttackStakes_0} & \text{if } \attackCost{\allAttackStakes_0} = 0 ;
        \end{cases} \\
        \underset{\eqref{equation:proof_first_direction_attack_cost_0}}{=}
        \frac{1}{|\allValidators|} \cdot \totalAttackPrize{\allAttackStakesAt{0}} - \validatorAttackCost{\validator}{\allAttackStakes_0}
        \underset{\eqref{equation:proof_first_direction_total_cost_more_than_validator_cost}}{=} \frac{1}{|\allValidators|} \cdot \totalAttackPrize{\allAttackStakesAt{0}} .
    \end{multline}

    Due to the definition of cryptoeconomic security~(Definition~\ref{definition:restaking_network_security}), it must be that~$\attackedServicesAt{0} = \emptyset$.
    This implies~${\totalAttackPrize{\allAttackStakesAt{0}} = 0}$~(\eqreft{equation:total_attack_prize}), and so
    \begin{equation}
        \label{equation:proof_first_direction_validator_utility_0_is_0}
        \validatorUtilitySecurityGame{\validator}{\allAttackStakes_0}
        \underset{\eqref{equation:proof_first_direction_validator_utility_0}}{=} \frac{1}{|\allValidators|} \cdot \totalAttackPrize{\allAttackStakesAt{0}}
        = 0 .
    \end{equation}

    In addition, due to the definition of cryptoeconomic security~(Definition~\ref{definition:restaking_network_security}),~$\allAttackStakes_0$ is a strong Nash equilibrium of the security game of the network~$\networkState$.
    That means that for any strategy profile~${\allAttackStakes \neq \allAttackStakes_0}$, there exists a validator~$\validator \in \allValidators$ that is worse off under~$\allAttackStakes$ than under~$\allAttackStakes_0$, that is,
    \begin{equation}
        \label{equation:proof_first_direction_validator_utility_other_is_negative}
        \validatorUtilitySecurityGame{\validator}{\allAttackStakes} < \validatorUtilitySecurityGame{\validator}{\allAttackStakes_0}
        \underset{\eqref{equation:proof_first_direction_validator_utility_0_is_0}}{=} 0 .
    \end{equation}

    Developing the utility of~$\validator$ under~$\allAttackStakes$, we get that
    \begin{multline}
        \validatorUtilitySecurityGame{\validator}{\allAttackStakes}
        \underset{\eqref{equation:validator_utility_security_game}}{=} \validatorPrizeShareSecurityGame{\validator}{\allAttackStakes} \cdot \totalAttackPrize{\allAttackStakes} - \validatorAttackCost{\validator}{\allAttackStakes} \\
        \underset{\eqref{equation:validator_prize_share_security_game}}{=} \begin{cases}
            \frac{\validatorAttackCost{\validator}{\allAttackStakes}}{\attackCost{\allAttackStakes}} \cdot \totalAttackPrize{\allAttackStakes} - \validatorAttackCost{\validator}{\allAttackStakes} & \text{if } \attackCost{\allAttackStakes} > 0 ; \\
            \frac{1}{|\allValidators|}  \cdot \totalAttackPrize{\allAttackStakes} - \validatorAttackCost{\validator}{\allAttackStakes} & \text{if } \attackCost{\allAttackStakes} = 0 .
        \end{cases} \\
        \underset{\eqref{equation:proof_first_direction_total_cost_more_than_validator_cost}}{=} \begin{cases}
            \frac{\validatorAttackCost{\validator}{\allAttackStakes}}{\attackCost{\allAttackStakes}} \cdot \totalAttackPrize{\allAttackStakes} - \validatorAttackCost{\validator}{\allAttackStakes} & \text{if } \attackCost{\allAttackStakes} > 0 ; \\
            \frac{1}{|\allValidators|}  \cdot \totalAttackPrize{\allAttackStakes} & \text{if } \attackCost{\allAttackStakes} = 0 .
        \end{cases}
        \underset{\eqref{equation:proof_first_direction_validator_utility_other_is_negative}}{<} 0 .
    \end{multline}
    Since~$\totalAttackPrize{\allAttackStakes} \geq 0$, for the last inequality to hold it must be that~$\validatorAttackCost{\validator}{\allAttackStakes} > 0$. 
    Hence,
    \begin{equation}
        \frac{\validatorAttackCost{\validator}{\allAttackStakes}}{\attackCost{\allAttackStakes}} \cdot \totalAttackPrize{\allAttackStakes} - \validatorAttackCost{\validator}{\allAttackStakes} < 0 .
    \end{equation}
    And because~$\validatorAttackCost{\validator}{\allAttackStakes} \geq 0$, it must be that~$\attackCost{\allAttackStakes} > \totalAttackPrize{\allAttackStakes}$.
    Therefore, there exists no profitable attack~(Definition~\ref{definition:attack_profitability}).
    
    \paragraphEmph{Second direction}
    Assume there exists some profitable attack~$\allAttackStakes$.
    We claim it is an alternative strategy profile where some coalition deviated, and it resulted with all of them being better off and thus the strategy profile~$\allAttackStakesAt{0}$ is not a strong Nash equilibrium, meaning the network is not secure.

    By Definition~\ref{definition:attack_profitability},
    \begin{equation}
        \attackedServices \neq \emptyset,
    \end{equation}
    and
    \begin{equation}
        \label{equation:proof_first_direction_attack_cost_total_attack_prize}
        \attackCost{\allAttackStakes} \leq \totalAttackPrize{\allAttackStakes} .
    \end{equation}
    Consider the utility of validator~$\validator$ resulting from the strategy profile~$\allAttackStakes$,
    \begin{multline}
        \label{equation:proof_first_direction_validator_utility}
        \validatorUtilitySecurityGame{\validator}{\allAttackStakes}
        \underset{\eqref{equation:validator_utility_security_game}}{=} \validatorPrizeShareSecurityGame{\validator}{\allAttackStakes} \cdot \totalAttackPrize{\allAttackStakes} - \validatorAttackCost{\validator}{\allAttackStakes} \\
        \underset{\eqref{equation:validator_prize_share_security_game}}{=} \begin{cases}
            \frac{\validatorAttackCost{\validator}{\allAttackStakes}}{\attackCost{\allAttackStakes}} \cdot \totalAttackPrize{\allAttackStakes} - \validatorAttackCost{\validator}{\allAttackStakes} & \text{if } \attackCost{\allAttackStakes} > 0 ; \\
            \frac{1}{|\allValidators|}  \cdot \totalAttackPrize{\allAttackStakes} - \validatorAttackCost{\validator}{\allAttackStakes} & \text{if } \attackCost{\allAttackStakes} = 0 .
        \end{cases}
        \geq 0 ;
    \end{multline}
    in the first case it follows from~\eqreft{equation:proof_first_direction_attack_cost_total_attack_prize}, and in the second case it follows from the fact that~$\validatorAttackCost{\validator}{\allAttackStakes}$ must be zero if~${\attackCost{\allAttackStakes} = 0}$.

    Now consider the strategy profile~$\allAttackStakes_0$, where for all~$\validator \in \allValidators$ and all~$\service \in \allServices$,~$\attackStake{\validator}{\service} = 0$.
    As we showed above, the utility of~$\validator$ under~$\allAttackStakes_0$ is
    \begin{equation}
        \label{equation:proof_second_direction_validator_utility_0}
        \validatorUtilitySecurityGame{\validator}{\allAttackStakes_0}
        \underset{\eqref{equation:proof_first_direction_validator_utility_0_is_0}}{=} \frac{1}{|\allValidators|} \cdot \totalAttackPrize{\allAttackStakesAt{0}} .
    \end{equation}
    It must be either that~$\attackedServicesAt{0} \neq \emptyset$, which means that the restaking network is not secure (Definition~\ref{definition:restaking_network_security}), or that~$\attackedServicesAt{0} = \emptyset$, which means that the total attack prize~$\totalAttackPrize{\allAttackStakesAt{0}}$ is~0.

    Thus, for all~$\validator \in \allValidators$,
    \begin{equation}
        \validatorUtilitySecurityGame{\validator}{\allAttackStakes_0}
        \underset{\eqref{equation:proof_first_direction_validator_utility_0}}{=} 0
        \underset{\eqref{equation:proof_first_direction_validator_utility}}{\leq} \validatorUtilitySecurityGame{\validator}{\allAttackStakes} .
    \end{equation}
    Therefore, by Definition~\ref{definition:strong_nash_equilibrium}, the strategy profile~$\allAttackStakes_0$ is not a strong Nash equilibrium of the restaking network security game, as otherwise we must have had some validator~$\validator \in \allValidators$ such that~${\validatorUtilitySecurityGame{\validator}{\allAttackStakes_0} > \validatorUtilitySecurityGame{\validator}{\allAttackStakes}}$.
    Hence, the network is not cryptoeconomically secure. 
\end{proof}


\subsection{Proofs Deferred from Subsection~\ref{section:security_analysis:sufficient_conditions_for_security}}


\begin{theorem}[Theorem~\ref{theorem:eigenlayer_condition} restated]
    \label{theorem:eigenlayer_condition_appendix}
    A network~$\networkState$ is secure if a misbehaving validator is slashed for their stake~(\eqreft{equation:eigenlayer_condition_simplified_assumption}), and for all validators~$\validator \in \allValidators$:
    \begin{equation}
        \label{equation:eigenlayer_condition}
        \sum_{\service \in \allServices} \frac{\allocation{\validator}{\service}}{\sum_{\validator' \in \allValidators} \allocation{\validator'}{\service}} \cdot \frac{\attackPrize{\service}}{\attackThreshold{\service}} < \stake{\validator} .
    \end{equation}
\end{theorem}
\begin{proof}
    (Adapted from~\citet{eigenlayer2024restaking})
    Assume towards a contradiction that the condition in the theorem holds, but the network~${\networkState=(\allValidators, \allServices, \allStakes, \allAllocations, \allAttackThresholds, \allAttackPrizes)}$ is insecure.
    Due to Proposition~\ref{proposition:restaking_network_security}, there exists a profitable attack~$\allAttackStakes$.

    Let~$\attackingValidators$ be the set of validators that misbehave in the attack~$\allAttackStakes$, that is,
    \begin{equation}
        \label{equation:eigenlayer_condition_proof_attacking_validators}
        \attackingValidators = \left\{ \validator \in \allValidators \middle| \sum_{\service \in \allServices} \attackStake{\validator}{\service} > 0 \right\} .
    \end{equation}
   
    Due to Definition~\ref{definition:attacked_services}, for all services~$\service \in \attackedServices$,
    \begin{multline}
        \attackThreshold{\service} \cdot \sum_{\validator \in \allValidators} \allocation{\validator}{\service}
        \leq \sum_{\validator \in \allValidators} \attackStake{\validator}{\service}
        = \sum_{\validator \in \allValidators \setminus \attackingValidators} \attackStake{\validator}{\service} + \sum_{\validator \in \attackingValidators} \attackStake{\validator}{\service} \\
        \underset{\eqref{equation:eigenlayer_condition_proof_attacking_validators}}{=} \sum_{\validator \in \attackingValidators} \attackStake{\validator}{\service} .
    \end{multline}
    And since for all~$\validator \in \allValidators$ and all~$\service \in \allServices$,~$\attackStake{\validator}{\service} \leq \allocation{\validator}{\service}$,
    \begin{equation}
        \label{equation:eigenlayer_condition_proof_feasibility}
        \attackThreshold{\service} \cdot \sum_{\validator \in \allValidators} \allocation{\validator}{\service}
        \leq \sum_{\validator \in \attackingValidators} \allocation{\validator}{\service} .
    \end{equation}
    
    Starting from the left-hand side of~\eqreft{equation:eigenlayer_condition}, and using~\eqreft{equation:eigenlayer_condition_proof_feasibility}, we get
    \begin{multline}
        \label{equation:eigenlayer_condition_proof_left_hand_side}
        \sum_{\service \in \allServices} \frac{\allocation{\validator}{\service}}{\sum_{\validator' \in \allValidators} \allocation{\validator'}{\service}} \cdot \frac{\attackPrize{\service}}{\attackThreshold{\service}}
        = \sum_{\service \in \allServices} \frac{\allocation{\validator}{\service} \cdot \attackPrize{\service}}{\attackThreshold{\service} \cdot \sum_{\validator' \in \allValidators} \allocation{\validator'}{\service}} \\
        \underset{\eqref{equation:eigenlayer_condition_proof_feasibility}}{\geq} \sum_{\service \in \allServices} \frac{\allocation{\validator}{\service} \cdot \attackPrize{\service}}{\sum_{\validator' \in \attackingValidators} \allocation{\validator'}{\service}} .
    \end{multline}
    Then, summing over all validators in~$\attackingValidators$, we get
    \begin{multline}
        \label{equation:eigenlayer_condition_proof_right_hand_side}
        \sum_{\validator \in \attackingValidators} \stake{\validator}
        \underset{\eqref{equation:eigenlayer_condition}}{>} \sum_{\validator \in \attackingValidators} \sum_{\service \in \allServices} \frac{\allocation{\validator}{\service}}{\sum_{\validator' \in \attackingValidators} \allocation{\validator'}{\service}} \cdot \frac{\attackPrize{\service}}{\attackThreshold{\service}} \\
        \underset{\eqref{equation:eigenlayer_condition_proof_left_hand_side}}{\geq} \sum_{\validator \in \attackingValidators} \sum_{\service \in \allServices} \frac{\allocation{\validator}{\service} \cdot \attackPrize{\service}}{\sum_{\validator' \in \attackingValidators} \allocation{\validator'}{\service}}
        \geq \sum_{\service \in \allServices} \frac{\sum_{\validator \in \attackingValidators} \allocation{\validator}{\service}}{\sum_{\validator' \in \attackingValidators} \allocation{\validator'}{\service}} \cdot \attackPrize{\service} \\
        = \sum_{\service \in \allServices} \attackPrize{\service}
        \geq \sum_{\service \in \attackedServices} \attackPrize{\service}
        \underset{\eqref{equation:total_attack_prize}}{=} \totalAttackPrize{\allAttackStakes} .
    \end{multline}

    Due to the assumption that misbehaving validators are slashed for all their stake~(\eqreft{equation:eigenlayer_condition_simplified_assumption}), this means that the stake of each validator~$\validator \in \attackingValidators$ is fully slashed, and thus the attack cost is
    \begin{equation}
        \attackCost{\allAttackStakes} = \sum_{\validator \in \attackingValidators} \stake{\validator} .
    \end{equation}
    
    Combined with~\eqreft{equation:eigenlayer_condition_proof_right_hand_side}, we get that~$\attackCost{\allAttackStakes} > \totalAttackPrize{\allAttackStakes}$, meaning that the attack is not profitable, in contradiction to our assumption.
    Thus, the network~$\networkState$ is secure.
\end{proof}

\begin{proposition}[Proposition~\ref{proposition:generalized_eigenlayer_condition} restated]
    \label{proposition:generalized_eigenlayer_condition_appendix}
    A network~$\networkState$ is secure if all validators~$\validator \in \allValidators$ should be slashed by less than their total stake:
    \begin{equation}
        \label{equation:generalized_eigenlayer_condition_validator}
        \sum_{\service \in \allServices} \frac{\allocation{\validator}{\service}}{\sum_{\validator' \in \allValidators} \allocation{\validator'}{\service}} \cdot \frac{\attackPrize{\service}}{\attackThreshold{\service}} < \stake{\validator} ,
    \end{equation}
    and all services~$\service \in \allServices$ have sufficient stake to cover their prizes:
    \begin{equation}
        \label{equation:generalized_eigenlayer_condition_service}
        \sum_{\validator \in \allValidators} \allocation{\validator}{\service} > \frac{\attackPrize{\service}}{\attackThreshold{\service}} .
    \end{equation}
\end{proposition}
\begin{proof}
    Assume towards a contradiction that the network~$\networkState=(\allValidators, \allServices, \allStakes, \allAllocations, \allAttackThresholds, \allAttackPrizes)$ is insecure.
    Due to Proposition~\ref{proposition:restaking_network_security}, there exists a profitable attack~$\allAttackStakes$.

    Due to Definition~\ref{definition:attacked_services}, for each service~$\service \in \attackedServices$,
    \begin{equation}
        \label{equation:generalized_eigenlayer_condition_proof_feasibility}
        \attackThreshold{\service} \cdot \sum_{\validator \in \allValidators} \allocation{\validator}{\service}
        \leq \sum_{\validator \in \allValidators} \attackStake{\validator}{\service} .
    \end{equation}

    The slashed amount from validator~$\validator$ in the attack is given by~\eqreft{equation:validator_attack_cost}:
    \begin{equation}
        \label{equation:generalized_eigenlayer_condition_proof_validator_cost}
        \validatorAttackCost{\validator}{\allAttackStakes} = \min \left( \stake{\validator}, \sum_{\service \in \attackedServices} \attackStake{\validator}{\service} \right) .
    \end{equation}
    To lower-bound the cost, we need to lower-bound both of the terms in the minimum.
    For the first term, we start from the~\eqreft{equation:generalized_eigenlayer_condition_validator}, and use the fact that~$\attackStake{\validator}{\service} \leq \allocation{\validator}{\service}$ and that~$\attackedServices \subseteq \allServices$:
    \begin{equation}
        \label{equation:generalized_eigenlayer_condition_proof_first_term}
        \stake{\validator}
        \underset{\eqref{equation:generalized_eigenlayer_condition_validator}}{>} \sum_{\service \in \allServices} \frac{\allocation{\validator}{\service}}{\sum_{\validator' \in \allValidators} \allocation{\validator'}{\service}} \cdot \frac{\attackPrize{\service}}{\attackThreshold{\service}}
        \geq \sum_{\service \in \attackedServices} \frac{\attackStake{\validator}{\service}}{\sum_{\validator' \in \allValidators} \allocation{\validator'}{\service}} \cdot \frac{\attackPrize{\service}}{\attackThreshold{\service}} .
    \end{equation}
    For the second term, we start from~\eqreft{equation:generalized_eigenlayer_condition_service}, rearrange and sum over all services in~$\attackedServices$:
    \begin{align}
        \sum_{\validator' \in \allValidators} \allocation{\validator'}{\service} &\underset{\eqref{equation:generalized_eigenlayer_condition_service}}{>}  \frac{\attackPrize{\service}}{\attackThreshold{\service}} ; \\
        1 &> \frac{1}{\sum_{\validator' \in \allValidators} \allocation{\validator'}{\service}} \cdot \frac{\attackPrize{\service}}{\attackThreshold{\service}} ; \\
        \attackStake{\validator}{\service} &> \frac{\attackStake{\validator}{\service}}{\sum_{\validator' \in \allValidators} \allocation{\validator'}{\service}} \cdot \frac{\attackPrize{\service}}{\attackThreshold{\service}} ; \\
        \label{equation:generalized_eigenlayer_condition_proof_second_term}
        \sum_{\service \in \attackedServices} \attackStake{\validator}{\service} &> \sum_{\service \in \attackedServices} \frac{\attackStake{\validator}{\service}}{\sum_{\validator' \in \allValidators} \allocation{\validator'}{\service}} \cdot \frac{\attackPrize{\service}}{\attackThreshold{\service}} .
    \end{align}

    Combining~\eqreft{equation:generalized_eigenlayer_condition_proof_first_term} and~\eqreft{equation:generalized_eigenlayer_condition_proof_second_term}, and then using~\eqreft{equation:generalized_eigenlayer_condition_proof_feasibility}, we get
    \begin{multline}
        \label{equation:generalized_eigenlayer_condition_proof_cost}
        \validatorAttackCost{\validator}{\allAttackStakes}
        \underset{\eqref{equation:generalized_eigenlayer_condition_proof_validator_cost}}{=} \min \left( \stake{\validator}, \sum_{\service \in \attackedServices} \attackStake{\validator}{\service} \right) \\
        \underset{\eqref{equation:generalized_eigenlayer_condition_proof_first_term},\eqref{equation:generalized_eigenlayer_condition_proof_second_term}}{>} \sum_{\service \in \attackedServices} \frac{\attackStake{\validator}{\service}}{\sum_{\validator' \in \allValidators} \allocation{\validator'}{\service}} \cdot \frac{\attackPrize{\service}}{\attackThreshold{\service}}
        = \sum_{\service \in \attackedServices} \frac{\attackStake{\validator}{\service} \cdot \attackPrize{\service}}{\attackThreshold{\service} \cdot \sum_{\validator' \in \allValidators} \allocation{\validator'}{\service}} \\
        \underset{\eqref{equation:generalized_eigenlayer_condition_proof_feasibility}}{\geq} \sum_{\service \in \attackedServices} \frac{\attackStake{\validator}{\service} \cdot \attackPrize{\service}}{\sum_{\validator' \in \allValidators} \attackStake{\validator'}{\service}} .
    \end{multline}

    Then, summing over all validators~$\allValidators$, we get
    \begin{multline}
        \attackCost{\allAttackStakes}
        \underset{\eqref{equation:total_attack_cost}}{=} \sum_{\validator \in \allValidators} \validatorAttackCost{\validator}{\allAttackStakes}
        \underset{\eqref{equation:generalized_eigenlayer_condition_proof_cost}}{>} \sum_{\validator \in \allValidators} \sum_{\service \in \attackedServices} \frac{\attackStake{\validator}{\service} \cdot \attackPrize{\service}}{\sum_{\validator' \in \allValidators} \attackStake{\validator'}{\service}} \\
        = \sum_{\service \in \attackedServices} \frac{\sum_{\validator \in \allValidators} \attackStake{\validator}{\service}}{\sum_{\validator' \in \allValidators} \attackStake{\validator'}{\service}} \cdot \attackPrize{\service}
        = \sum_{\service \in \attackedServices} \attackPrize{\service}
        \underset{\eqref{equation:total_attack_prize}}{=} \totalAttackPrize{\allAttackStakes} .
    \end{multline}

    Overall, we get that~$\attackCost{\allAttackStakes} > \totalAttackPrize{\allAttackStakes}$, meaning that the attack is not profitable, in contradiction to our assumption.
    Thus, the network~$\networkState$ is secure.
\end{proof}


\subsection{Proofs Deferred from Subsection~\ref{section:security_analysis:searching_for_attacks_is_np_complete}}


\begin{proposition}[Proposition~\ref{proposition:allocation_indivisible_attack_np_complete} restated]
    \label{proposition:allocation_indivisible_attack_np_complete_appendix}
    Determining whether there exists a profitable allocation-indivisible attack~$\allAttackStakes$ in a restaking network~${\networkState = (\allValidators, \allServices, \allStakes, \allAllocations, \allAttackThresholds, \allAttackPrizes)}$ is NP-complete.
\end{proposition}
\begin{proof}
    First, the problem is in NP, as given an allocation-indivisible attack, we can verify that it is profitable in polynomial time using the conditions of Definition~\ref{definition:attack_profitability}.

    Next, we show a reduction from the Subset Sum problem.
    Let~$\left\{ \sspElement{1}, \ldots, \sspElement{\sspElementCount} \right\}$ and~$\sspTarget$ be an instance of the Subset Sum problem.
    Denote
    \begin{equation}
        \label{equation:reduction_indivisible_element_sum}
        \sspElementSum = \sum_{i=1}^\sspElementCount \sspElement{i} .
    \end{equation}
    Assume that
    \begin{equation}
        \label{equation:reduction_indivisible_subset_sum_assumption}
        0 < \sspTarget \leq \sspElementSum .
    \end{equation}
    Otherwise, the Subset Sum problem is trivial, as no subset can sum to the target.

    \begin{figure}[t]
        \centering
        \begin{subfigure}[t]{0.45\textwidth}
            \centering
            \begin{tikzpicture}[
                validator/.style={circle, draw, minimum size=0.6cm, font=\scriptsize},
                service/.style={circle, draw, minimum size=0.6cm, font=\scriptsize, inner sep=0.035cm},
                scale=0.5
            ]
                \node[align=center] at (0,4) {$\allValidators$ \\ $(\allStakes)$};
                \node[align=center] at (2, 4) {--- \\ $\allAllocations$};
                \node[align=center] at (4,4) {$\allServices$ \\ $(\allAttackThresholds$ | $\allAttackPrizes)$};

                \node[validator] (V1) at (0,2) {$\sspElement{1}$};
                \node[validator] (V2) at (0,0) {$\sspElement{2}$};
                \node (Vdots) at (0,-1.25) {$\vdots$};
                \node[validator] (Vn) at (0,-3) {$\sspElement{\sspElementCount}$};
                
                \node[service] (S) at (4,0) {$\frac{\sspTarget}{\sspElementSum}~|~\sspTarget$};
                
                \draw (V1) -- (S) node[pos=0.2, above] {\scriptsize $\sspElement{1}$};
                \draw (V2) -- (S) node[pos=0.2, above] {\scriptsize $\sspElement{2}$};
                \draw (Vn) -- (S) node[pos=0.2, above] {\scriptsize $\sspElement{n}$};
            \end{tikzpicture}
            \caption{Reduction for allocation-indivisible attacks}
            \label{figure:reduction_indivisible}
            \Description{Graph showing a reduction from the Subset Sum problem to an allocation-indivisible attack in a restaking network.}
        \end{subfigure}
        \\
        \begin{subfigure}[t]{0.45\textwidth}
            \centering
            \begin{tikzpicture}[
                validator/.style={circle, draw, minimum size=0.6cm, font=\scriptsize},
                service/.style={circle, draw, minimum size=0.6cm, font=\scriptsize, inner sep=0.035cm},
                scale=0.5
            ]
                \node[align=center] at (0,4) {$\allValidators$ \\ $(\allStakes)$};
                \node[align=center] at (2, 4) {--- \\ $\allAllocations$};
                \node[align=center] at (4,4) {$\allServices$ \\ $(\allAttackThresholds$ | $\allAttackPrizes)$};

                \node[validator] (V1) at (0,2) {$\sspElement{1}$};
                \node[validator] (V2) at (0,0) {$\sspElement{2}$};
                \node (Vdots) at (0,-1.25) {$\vdots$};
                \node[validator] (Vn) at (0,-3) {$\sspElement{\sspElementCount}$};
                
                \node[service] (S1) at (4,2) {$1~|~\frac{\sspElement{1}}{2}$};
                \node[service] (S2) at (4,0) {$1~|~\frac{\sspElement{2}}{2}$};
                \node (Sdots) at (4,-1.25) {$\vdots$};
                \node[service] (Sn) at (4,-3) {$1~|~\frac{\sspElement{\sspElementCount}}{2}$};
                \node[service] (Snp1) at (4,-6) {$\frac{\sspTarget}{\sspElementSum}~|~\frac{\sspTarget}{2}$};
                
                \draw (V1) -- (Snp1) node[pos=0.03, right] {\scriptsize $\sspElement{1}$};
                \draw (V2) -- (Snp1) node[pos=0.05, below] {\scriptsize $\sspElement{2}$};
                \draw (Vn) -- (Snp1) node[pos=0.04, right] {\scriptsize $\sspElement{n}$};
                
                \draw (V1) -- (S1) node[pos=0.11, above] {\scriptsize $\sspElement{1}$};
                \draw (V2) -- (S2) node[pos=0.5, above] {\scriptsize $\sspElement{2}$};
                \draw (Vn) -- (Sn) node[pos=0.11, above] {\scriptsize $\sspElement{n}$};
            \end{tikzpicture}
            \caption{Reduction for allocation-divisible attacks}
            \label{figure:reduction_divisible}
            \Description{Graph showing a reduction from the Subset Sum problem to an allocation-divisible attack in a restaking network.}
        \end{subfigure}
        \caption{Reductions from Subset Sum to finding attacks in restaking networks.}
        \label{figure:reductions}
        \Description{Graphs showing reductions from the Subset Sum problem to finding attacks in restaking networks.}
    \end{figure}

    We construct a network~(Fig.~\ref{figure:reduction_indivisible}) with a single service~$\allServices = \left\{ \service \right\}$ and~$\sspElementCount$ validators~${\left\{ \validator_1, \ldots, \validator_{\sspElementCount} \right\}}$.
    For each~$i \in \left\{ 1, \ldots, \sspElementCount \right\}$, set
    \begin{align}
        \label{equation:reduction_indivisible_stake}
        \stake{\validator_i} &= \sspElement{i} ; \\
        \label{equation:reduction_indivisible_allocation}
        \allocation{\validator_i}{\service}
        &= \stake{\validator_i}
        = \sspElement{i} .
    \end{align}
    Also, set
    \begin{align}
        \label{equation:reduction_indivisible_attack_threshold}
        \attackThreshold{\service} &= \frac{\sspTarget}{\sspElementSum} ;\\
        \label{equation:reduction_indivisible_attack_prize}
        \attackPrize{\service} &= \sspTarget .
    \end{align}

    Due to~\eqreft{equation:reduction_indivisible_subset_sum_assumption},~$0 < \attackThreshold{\service} \leq 1$, so the attack threshold is well-defined.

    We claim that the network has a profitable allocation-indivisible attack if and only if the Subset Sum problem has a solution.

    \paragraphEmph{First Direction}
    Assume there exists a subset~$\left\{ \sspElement{i_1}, \ldots, \sspElement{i_k} \right\}$ of the~$\sspElementCount$ elements that sums to~$\sspTarget$:
    \begin{equation}
        \label{equation:reduction_indivisible_if}
        \sum_{j=1}^k \sspElement{i_j} = \sspTarget .
    \end{equation}

    Consider the attack~$\allAttackStakes$ where
    \begin{equation}
        \label{equation:reduction_indivisible_if_attack_stake}
        \attackStake{\validator}{\service}
        = \begin{cases}
            \allocation{\validator}{\service} & \text{if } \validator \in \left\{ \validator_{i_1}, \ldots, \validator_{i_k} \right\} ; \\
            0 & \text{otherwise} .
        \end{cases}
    \end{equation}

    Consider the service~$\service$:
    \begin{multline}
        \attackThreshold{\service} \cdot \sum_{i=1}^\sspElementCount \allocation{\validator_i}{\service}
        \underset{\eqref{equation:reduction_indivisible_attack_threshold}}{=} \frac{\sspTarget}{\sspElementSum} \cdot \sum_{i=1}^\sspElementCount \allocation{\validator_i}{\service}
        \underset{\eqref{equation:reduction_indivisible_allocation}}{=} \frac{\sspTarget}{\sspElementSum} \cdot \sum_{i=1}^\sspElementCount \sspElement{i}
        \underset{\eqref{equation:reduction_indivisible_element_sum}}{=} \frac{\sspTarget}{\sspElementSum} \cdot \sspElementSum
        = \sspTarget \\
        \underset{\eqref{equation:reduction_indivisible_if}}{=} \sum_{j=1}^k \sspElement{i_j}
        \underset{\eqref{equation:reduction_indivisible_allocation}}{=} \sum_{j=1}^k \allocation{\validator_{i_j}}{\service}
        \underset{\eqref{equation:reduction_indivisible_if_attack_stake}}{=} \sum_{i=1}^n \attackStake{\validator_i}{\service} .
    \end{multline}
    Thus, by Definition~\ref{definition:attacked_services}, the service~$\service$ is attacked, and since it is the only service,
    \begin{equation}
        \label{equation:reduction_indivisible_if_attack_services}
        \attackedServices = \left\{ \service \right\} .
    \end{equation}

    The cost of each validator~$\validator \in \allValidators$ is
    \begin{multline}
        \label{equation:reduction_indivisible_if_validator_attack_cost}
        \validatorAttackCost{\validator}{\allAttackStakes}
        \underset{\eqref{equation:validator_attack_cost}}{=} \min \left( \stake{\validator}, \sum_{\service' \in \attackedServices} \attackStake{\validator}{\service'} \right)
        \underset{\eqref{equation:reduction_indivisible_if_attack_services}}{=} \min \left( \stake{\validator}, \attackStake{\validator}{\service} \right) \\
        \underset{\eqref{equation:reduction_indivisible_if_attack_stake}}{=} \begin{cases}
            \min \left( \stake{\validator}, \allocation{\validator}{\service} \right) & \text{if } \validator \in \left\{ \validator_{i_1}, \ldots, \validator_{i_k} \right\} ; \\
            min \left( \stake{\validator}, 0 \right) & \text{otherwise} ;
        \end{cases} \\
        \underset{\eqref{equation:reduction_indivisible_allocation}}{=} \begin{cases}
            \stake{\validator} & \text{if } \validator \in \left\{ \validator_{i_1}, \ldots, \validator_{i_k} \right\} ; \\
            0 & \text{otherwise} .
        \end{cases}
    \end{multline}
    Therefore, the attack is profitable:
    \begin{multline}
        \attackCost{\allAttackStakes}
        \underset{\eqref{equation:total_attack_cost}}{=} \sum_{i=1}^n \validatorAttackCost{\validator_i}{\allAttackStakes}
        \underset{\eqref{equation:reduction_indivisible_if_validator_attack_cost}}{=} \sum_{j=1}^k \stake{\validator_{i_j}}
        \underset{\eqref{equation:reduction_indivisible_stake}}{=} \sum_{j=1}^k \sspElement{i_j}
        \underset{\eqref{equation:reduction_indivisible_if}}{=} \sspTarget \\
        \underset{\eqref{equation:reduction_indivisible_attack_prize}}{=} \attackPrize{\service}
        \underset{\eqref{equation:reduction_indivisible_if_attack_services}}{=} \sum_{\service' \in \attackedServices} \attackPrize{\service'}
        \underset{\eqref{equation:total_attack_prize}}{=} \totalAttackPrize{\allAttackStakes} .
    \end{multline}

    \paragraphEmph{Second Direction}
    Assume that the network has a profitable allocation-indivisible attack~$\allAttackStakes$.
    
    Since the attack is allocation-indivisible,~$\attackStake{\validator}{\service} \in \left\{ 0, \allocation{\validator}{\service} \right\}$ for all~$\validator \in \allValidators$ and~$\service \in \allServices$.
    In addition, since an attack must target at least one service, it must be that
    \begin{equation}
        \label{equation:reduction_indivisible_only_if_attack_validators}
        \attackedServices = \left\{ \service \right\} .
    \end{equation}

    Denote by~$\attackingValidators = \left\{ \validator_1, \ldots, \validator_k \right\}$ the set of validators in the attack with non-zero allocations.
    Because the attack is allocation-indivisible, it holds that
    \begin{equation}
        \label{equation:reduction_indivisible_only_if_attack_stake}
        \attackStake{\validator}{\service}
        = \begin{cases}
            \allocation{\validator}{\service} & \text{if } \validator \in \left\{ \validator_{i_1}, \ldots, \validator_{i_k} \right\} ; \\
            0 & \text{otherwise} .
        \end{cases}
    \end{equation}

    Consider the subset~$\left\{ \sspElement{i_1}, \ldots, \sspElement{i_k} \right\}$, corresponding to the validators in the attack.
    We claim that this subset satisfies the Subset Sum problem.
    Since~$\service \in \attackedServices$,
    \begin{equation}
        \label{equation:reduction_indivisible_only_if_attack_feasibility}
        \attackThreshold{\service} \cdot \sum_{i=1}^\sspElementCount \allocation{\validator_i}{\service}
        \leq \sum_{i=1}^n \attackStake{\validator_i}{\service}.
    \end{equation}
    Using this inequality and~\eqreft{equation:reduction_indivisible_only_if_attack_stake}, we get
    \begin{equation}
        \label{equation:reduction_indivisible_only_if_attack_feasibility_with_allocation}
        \sum_{j=1}^k \allocation{\validator_{i_j}}{\service} 
        \underset{\eqref{equation:reduction_indivisible_only_if_attack_stake}}{=} \sum_{i=1}^n \attackStake{\validator_i}{\service} \underset{\eqref{equation:reduction_indivisible_only_if_attack_feasibility}}{\geq}
        \attackThreshold{\service} \cdot \sum_{i=1}^\sspElementCount \allocation{\validator_i}{\service} .
    \end{equation}

    Starting from the sum of the elements in the subset, we get
    \begin{multline}
        \label{equation:reduction_indivisible_only_if_greater_than_threshold}
        \sum_{j=1}^k \sspElement{i_j}
        \underset{\eqref{equation:reduction_indivisible_allocation}}{=} \sum_{j=1}^k \allocation{\validator_{i_j}}{\service}
        \underset{\eqref{equation:reduction_indivisible_only_if_attack_feasibility_with_allocation}}{\geq} \attackThreshold{\service} \cdot \sum_{i=1}^\sspElementCount \allocation{\validator_i}{\service} \\
        \underset{\eqref{equation:reduction_indivisible_attack_threshold}}{=} \frac{\sspTarget}{\sspElementSum} \cdot \sum_{i=1}^\sspElementCount \allocation{\validator_i}{\service}
        \underset{\eqref{equation:reduction_indivisible_allocation}}{=} \frac{\sspTarget}{\sspElementSum} \cdot \sum_{i=1}^\sspElementCount \sspElement{i}
        \underset{\eqref{equation:reduction_indivisible_element_sum}}{=} \frac{\sspTarget}{\sspElementSum} \cdot \sspElementSum
        = \sspTarget .
    \end{multline}

    In addition, since the attack is profitable, by Definition~\ref{definition:attack_profitability},
    \begin{equation}
        \label{equation:reduction_indivisible_only_if_attack_profitability}
        \attackCost{\allAttackStakes} \leq \totalAttackPrize{\allAttackStakes} .
    \end{equation}
    Furthermore, similar to the opposite direction, the cost of a validator~$\validator$ equals their stake if~$\validator \in \attackingValidators$ and is 0 otherwise:
    \begin{equation}
        \label{equation:reduction_indivisible_only_if_validator_attack_cost}
        \validatorAttackCost{\validator}{\allAttackStakes}
        = \begin{cases}
            \stake{\validator} & \text{if } \validator \in \attackingValidators ; \\
            0 & \text{otherwise} .
        \end{cases}
    \end{equation}
    Then, starting from the sum of the elements in the subset, and using the fact that the attack is profitable, we get
    \begin{multline}
        \label{equation:reduction_indivisible_only_if_less_than_threshold}
        \sum_{j=1}^k \sspElement{i_j}
        \underset{\eqref{equation:reduction_indivisible_stake}}{=} \sum_{j=1}^k \stake{\validator_{i_j}}
        \underset{\eqref{equation:reduction_indivisible_only_if_validator_attack_cost}}{=} \sum_{i=1}^n \validatorAttackCost{\validator_i}{\allAttackStakes}
        \underset{\eqref{equation:total_attack_cost}}{=} \attackCost{\allAttackStakes}
        \underset{\eqref{equation:reduction_indivisible_only_if_attack_profitability}}{\leq} \totalAttackPrize{\allAttackStakes} \\
        \underset{\eqref{equation:total_attack_prize}}{=} \sum_{\service' \in \attackedServices} \attackPrize{\service'}
        \underset{\eqref{equation:reduction_indivisible_only_if_attack_validators}}{=} \attackPrize{\service}
        \underset{\eqref{equation:reduction_indivisible_attack_prize}}{=} \sspTarget .
    \end{multline}

    Combining \eqreft{equation:reduction_indivisible_only_if_less_than_threshold} with \eqreft{equation:reduction_indivisible_only_if_greater_than_threshold}, we get
    \begin{equation}
        \sum_{j=1}^k \sspElement{i_j} = \sspTarget ,
    \end{equation}
    that is, the subset~$\left\{ \sspElement{i_1}, \ldots, \sspElement{i_k} \right\}$ is a solution to the Subset Sum 
    problem.
    
    Therefore, determining whether a network has a profitable allocation-indivisible attack is NP-complete.
\end{proof}

\begin{proposition}[Proposition~\ref{proposition:allocation_divisible_attack_np_complete} restated]
    \label{proposition:allocation_divisible_attack_np_complete_appendix}
    Determining whether there exists a profitable allocation-divisible attack~$\left(\attackingValidators, \attackedServices, \allAttackStakes\right)$ in a restaking network~$\networkState = (\allValidators, \allServices, \allStakes, \allAllocations, \allAttackThresholds, \allAttackPrizes)$ is NP-complete.
\end{proposition}
\begin{proof}
    First, similarly to Proposition~\ref{proposition:allocation_indivisible_attack_np_complete}, the problem is in NP, as given an allocation-divisible attack, we can verify that it is profitable in polynomial time using the condition of Definition~\ref{definition:attack_profitability}.

    Next, we show a reduction from the Subset Sum problem.
    Let~$\left\{ \sspElement{1}, \ldots, \sspElement{\sspElementCount} \right\}$ and~$\sspTarget$ be an instance of the Subset Sum problem.
    Denote by~$\sspElementSum$ the sum of the elements, namely,
    \begin{equation}
        \label{equation:reduction_divisible_element_sum}
        \sspElementSum = \sum_{i=1}^\sspElementCount \sspElement{i} .
    \end{equation}
    As in the proof of Proposition~\ref{proposition:allocation_indivisible_attack_np_complete}, assume that
    \begin{equation}
        \label{equation:reduction_divisible_subset_sum_assumption}
        0 < \sspTarget \leq \sspElementSum .
    \end{equation}

    We construct a network~(Fig.~\ref{figure:reduction_divisible}) with~$\sspElementCount$ validators:~${\allValidators = \left\{ \validator_1, \ldots, \validator_{\sspElementCount} \right\}}$; and~$\sspElementCount + 1$ services:~$\allServices = \left\{ \service_1, \ldots, \service_{\sspElementCount + 1} \right\}$. 
    For each~$i \in \left\{ 1, \ldots, \sspElementCount \right\}$ and~${t \in \left\{ 1, \ldots, \sspElementCount + 1 \right\}}$, set
    \begin{align}
        \label{equation:reduction_divisible_stake}
        \stake{\validator_i} &= \sspElement{i} ; \\
        \label{equation:reduction_divisible_allocation}
        \allocation{\validator_i}{\service_t} &= \begin{cases}
            \sspElement{i} & \text{if } t \in \{ i, \sspElementCount + 1 \} ; \\
            0 & \text{otherwise} .
        \end{cases}
    \end{align}
    Also, set
    \begin{align}
        \label{equation:reduction_divisible_attack_threshold_sTarget}
        \attackThreshold{\service_{\sspElementCount + 1}} &= \frac{\sspTarget}{\sspElementSum} ; \\
        \label{equation:reduction_divisible_attack_prize_sTarget}
        \attackPrize{\service_{\sspElementCount + 1}} &= \frac{\sspTarget}{2} . \\
    \end{align}
    In addition, set for all~$i \in \left\{ 1, \ldots, \sspElementCount \right\}$
    \begin{align}
        \label{equation:reduction_divisible_attack_threshold_si}
        \attackThreshold{\service_i} &= 1 ; \\
        \label{equation:reduction_divisible_attack_prize_si}
        \attackPrize{\service_i} &= \frac{\sspElement{i}}{2} .
    \end{align}

    We claim that the network has a profitable allocation-divisible attack if and only if the Subset Sum problem has a solution.

    \paragraphEmph{First Direction}
    Assume there exists a subset~$\left\{ \sspElement{i_1}, \ldots, \sspElement{i_k} \right\}$ that sums to~$\sspTarget$:
    \begin{equation}
        \label{equation:reduction_divisible_if_subset_sum}
        \sum_{j=1}^k \sspElement{i_j} = \sspTarget .
    \end{equation}
    
    Consider the attack~$\allAttackStakes$ such that for each~$i \in \left\{ 1, \ldots, n \right\}$,~$t \in \left\{ 1, \ldots, \sspElementCount + 1 \right\}$
    \begin{equation}
        \label{equation:reduction_divisible_if_attack_stake}
        \attackStake{\validator_i}{\service_t} = \begin{cases}
            \sspElement{i} & \text{if } i \in \left\{ i_1, \ldots, i_k \right\} \text{ and } t \in \{ i, \sspElementCount + 1 \} ; \\
            0 & \text{otherwise} .
        \end{cases}
    \end{equation}

    We claim this attack is profitable.
    We first show that~$\service_{\sspElementCount + 1} \in \attackedServices$:
    \begin{multline}
        \label{equation:reduction_divisible_if_feasibility_sTarget}
        \attackThreshold{\service_{\sspElementCount + 1}} \cdot \sum_{\validator \in \allValidators} \allocation{\validator}{\service_{\sspElementCount + 1}}
        \underset{\eqref{equation:reduction_divisible_attack_threshold_sTarget}}{=} \frac{\sspTarget}{\sspElementSum} \cdot \sum_{\validator \in \allValidators} \allocation{\validator}{\service_{\sspElementCount + 1}}
        \underset{\eqref{equation:reduction_divisible_allocation}}{=} \frac{\sspTarget}{\sspElementSum} \cdot \sum_{i=1}^\sspElementCount \sspElement{i} \\
        \underset{\eqref{equation:reduction_divisible_element_sum}}{=} \frac{\sspTarget}{\sspElementSum} \cdot \sspElementSum
        = \sspTarget
        \underset{\eqref{equation:reduction_divisible_if_subset_sum}}{=} \sum_{j=1}^k \sspElement{i_j}
        \underset{\eqref{equation:reduction_divisible_if_attack_stake}}{=} \sum_{i=1}^n \attackStake{\validator_i}{\service_{\sspElementCount + 1}} .
    \end{multline}
    Then, we show that~$\service_{i_j} \in \attackedServices$ for all~$j \in \left\{ 1, \ldots, k \right\}$:
    \begin{multline}
        \label{equation:reduction_divisible_if_feasibility_sij}
        \attackThreshold{\service_{i_j}} \cdot \sum_{\validator \in \allValidators} \allocation{\validator}{\service_{i_j}}
        \underset{\eqref{equation:reduction_divisible_attack_threshold_si}}{=} 1 \cdot \sum_{\validator \in \allValidators} \allocation{\validator}{\service_{i_j}}
        \underset{\eqref{equation:reduction_divisible_allocation}}{=} \sspElement{i_j} \\
        \underset{\eqref{equation:reduction_divisible_if_attack_stake}}{=} \sum_{j=1}^k \attackStake{\validator_{i_j}}{\service_{i_j}} .
    \end{multline}
    By~\eqreft{equation:reduction_divisible_if_feasibility_sTarget} and~\eqreft{equation:reduction_divisible_if_feasibility_sij}, we get that
    \begin{equation}
        \label{equation:reduction_divisible_if_attacked_services}
        \left\{ \service_{i_1}, \ldots, \service_{i_k} \right\} \cup \left\{ \service_{\sspElementCount + 1} \right\} \subseteq \attackedServices .
    \end{equation}

    For~$j = 1, \ldots, k$, the cost of validator~$\validator_{i_j}$ equals~$\sspElement{i_j}$:
    \begin{multline}
        \label{equation:reduction_divisible_if_validator_attack_cost}
        \validatorAttackCost{\validator_{i_j}}{\allAttackStakes}
        \underset{\eqref{equation:validator_attack_cost}}{=} \min \left( \stake{\validator_{i_j}}, \sum_{\service' \in \allServices} \attackStake{\validator_{i_j}}{\service'} \right) \\
        \underset{\eqref{equation:reduction_divisible_if_attack_stake}}{=} \min \left( \stake{\validator_{i_j}}, \attackStake{\validator_{i_j}}{\service_0} + \attackStake{\validator_{i_j}}{\service_{i_j}} \right) \\
        \underset{\eqref{equation:reduction_divisible_if_attack_stake}}{=} \min \left( \stake{\validator_{i_j}}, 2 \sspElement{i_j} \right)
        \underset{\eqref{equation:reduction_divisible_stake}}{=} \min \left( \sspElement{i_j}, 2 \sspElement{i_j} \right)
        = \sspElement{i_j} .
    \end{multline}
    For all other validators~$\validator \in \allValidators \setminus \left\{ \validator_{i_1}, \ldots, \validator_{i_k} \right\}$, the cost of the attack is 0:
    \begin{equation}
        \label{equation:reduction_divisible_if_other_validator_attack_cost}
        \validatorAttackCost{\validator}{\allAttackStakes}
        \underset{\eqref{equation:validator_attack_cost}}{=} \min \left( \stake{\validator}, \sum_{\service' \in \allServices} \attackStake{\validator}{\service'} \right)
        \underset{\eqref{equation:reduction_divisible_if_attack_stake}}{=} \min \left( \stake{\validator}, \sum_{\service' \in \allServices} 0 \right)
        = 0 .
    \end{equation}
    
    The total cost of the attack is the sum of the costs of all validators:
    \begin{equation}
        \label{equation:reduction_divisible_if_total_attack_cost}
        \attackCost{\allAttackStakes}
        \underset{\eqref{equation:total_attack_cost}}{=} \sum_{i=1}^n \validatorAttackCost{\validator_i}{\allAttackStakes}
        \underset{\eqref{equation:reduction_divisible_if_validator_attack_cost}, \eqref{equation:reduction_divisible_if_other_validator_attack_cost}}{=} \sum_{j=1}^k \sspElement{i_j}
        \underset{\eqref{equation:reduction_divisible_if_subset_sum}}{=} \sspTarget
    \end{equation}
    The prize of the attack is the sum of the prizes of the attacked services:
    \begin{multline}
        \label{equation:reduction_divisible_if_total_attack_prize}
        \totalAttackPrize{\allAttackStakes}
        \underset{\eqref{equation:total_attack_prize}}{=} \sum_{\service \in \attackedServices} \attackPrize{\service}
        \underset{\eqref{equation:reduction_divisible_if_attacked_services}}{=} \attackPrize{\service_{\sspElementCount + 1}} + \sum_{j=1}^k \attackPrize{\service_{i_j}} \\
        \underset{\eqref{equation:reduction_divisible_attack_prize_sTarget}, \eqref{equation:reduction_divisible_attack_prize_si}}{=} \frac{\sspTarget}{2} + \sum_{j=1}^k \frac{\sspElement{i_j}}{2}
        = \frac{\sspTarget}{2} + \frac{\sum_{j=1}^k \sspElement{i_j}}{2}
        \underset{\eqref{equation:reduction_divisible_if_subset_sum}}{=} \frac{\sspTarget}{2} + \frac{\sspTarget}{2}
        = \sspTarget .
    \end{multline}

    Combining the last 2 equations, we get
    \begin{equation}
        \attackCost{\allAttackStakes}
        \underset{\eqref{equation:reduction_divisible_if_total_attack_cost}}{=} \sspTarget
        \underset{\eqref{equation:reduction_divisible_if_total_attack_prize}}{=} \totalAttackPrize{\allAttackStakes} .
    \end{equation}
    This satisfies Definition~\ref{definition:attack_profitability}, and therefore the attack is profitable.

    \paragraphEmph{Second Direction}
    Assume that the network has a profitable allocation-divisible attack~$\allAttackStakes$.

    Denote by~$\allServicesAt{I} = \{ \service_{i_1}, \ldots, \service_{i_k} \}$ the (possibly empty) set of the attacked services after removing~$\service_{\sspElementCount + 1}$:
    \begin{equation}
        \label{equation:reduction_divisible_only_if_attacked_services_excluding_sTarget}
        \allServicesAt{I}
        = \{ \service_{i_1}, \ldots, \service_{i_k} \}
        = \attackedServices \setminus \left\{ \service_{\sspElementCount + 1} \right\} .
    \end{equation}
    Consider the corresponding subset of the elements in the Subset Sum problem~$\left\{ \sspElement{i_1}, \ldots, \sspElement{i_k} \right\} $.
    We claim that this subset is a solution to the Subset Sum problem.

    Recall that for all~$\validator \in \allValidators$ and~$\service \in \allServices$
    \begin{equation}
        \label{equation:reduction_divisible_only_if_attack_stake_upper_bound}
        \attackStake{\validator}{\service}
        \leq \allocation{\validator}{\service} .
    \end{equation}
    Due to the definition of attacked services it holds that for each~$j \in \left\{ 1, \ldots, k \right\}$
    \begin{equation}
        \label{equation:reduction_divisible_only_if_attack_feasibility_sij}
        \attackThreshold{\service_{i_j}} \cdot \sum_{\validator \in \allValidators} \allocation{\validator}{\service_{i_j}}
        \leq \sum_{\validator \in \allValidators} \attackStake{\validator}{\service_{i_j}} .
    \end{equation}
    Developing~$\sspElement{i_j}$ to get the left-hand side, using the above inequality, and then developing the right-hand side, we get
    \begin{multline}
        \label{equation:reduction_divisble_only_if_attack_stake_lower_bound}
        \sspElement{i_j}
        \underset{\eqref{equation:reduction_divisible_allocation}}{=} \sum_{\validator \in \allValidators} \allocation{\validator}{\service_{i_j}}
        \underset{\eqref{equation:reduction_divisible_attack_threshold_si}}{=} \attackThreshold{\service_{i_j}} \cdot \sum_{\validator \in \allValidators} \allocation{\validator}{\service_{i_j}} \\
        \underset{\eqref{equation:reduction_divisible_only_if_attack_feasibility_sij}}{\leq} \sum_{\validator \in \allValidators} \attackStake{\validator}{\service_{i_j}}
        = \attackStake{\validator_{i_j}}{\service_{i_j}} + \sum_{\validator \in \allValidators \setminus \{ \validator_{i_j} \}} \attackStake{\validator}{\service_{i_j}} \\
        \underset{\eqref{equation:reduction_divisible_only_if_attack_stake_upper_bound}}{\leq} \attackStake{\validator_{i_j}}{\service_{i_j}} + \sum_{\validator \in \attackingValidators \setminus \{ \validator_{i_j} \}} \allocation{\validator}{\service_{i_j}} \\
        \underset{\eqref{equation:reduction_divisible_allocation}}{=} \attackStake{\validator_{i_j}}{\service_{i_j}} + \sum_{\validator \in \attackingValidators \setminus \{ \validator_{i_j} \}} 0
        = \attackStake{\validator_{i_j}}{\service_{i_j}}
    \end{multline}
    Furthermore, developing the previous inequality, we get
    \begin{equation}
        \sspElement{i_j}
        \underset{\eqref{equation:reduction_divisble_only_if_attack_stake_lower_bound}}{\leq} \attackStake{\validator_{i_j}}{\service_{i_j}}
        \underset{\eqref{equation:reduction_divisible_only_if_attack_stake_upper_bound}}{\leq} \allocation{\validator_{i_j}}{\service_{i_j}}
        \underset{\eqref{equation:reduction_divisible_allocation}}{=} \sspElement{i_j} .
    \end{equation}
    And that yields that for all~$j \in \left\{ 1, \ldots, k \right\}$
    \begin{equation}
        \label{equation:reduction_divisible_only_if_attack_stake_ij}
        \attackStake{\validator_{i_j}}{\service_{i_j}} = \sspElement{i_j} .
    \end{equation}
    
    We use the previous observations to lower bound the cost of the attack.
    To do so, we start from the cost of validators in~$\{ \validator_{i_1}, \ldots, \validator_{i_k} \}$.
    For each~$j \in \left\{ 1, \ldots, k \right\}$
    \begin{multline}
        \label{equation:reduction_divisible_only_if_validator_attack_cost_ij_lower_bound}
        \validatorAttackCost{\validator_{i_j}}{\allAttackStakes}
        \underset{\eqref{equation:validator_attack_cost}}{=} \min \left( \stake{\validator_{i_j}}, \sum_{\service \in \attackedServices} \attackStake{\validator_{i_j}}{\service} \right) \\
        \geq \min \left( \stake{\validator_{i_j}}, \attackStake{\validator_{i_j}}{\service_{i_j}} \right)
        \underset{\eqref{equation:reduction_divisible_stake}}{=} \min \left( \sspElement{i_j}, \attackStake{\validator_{i_j}}{\service_{i_j}} \right) \\
        \underset{\eqref{equation:reduction_divisible_only_if_attack_stake_ij}}{=} \min \left( \sspElement{i_j}, \sspElement{i_j} \right)
        = \sspElement{i_j} .
    \end{multline}
    Overall, since the cost of each validator is at most their stake, the cost of validator~$\validator_{i_j}$ is exactly~$\sspElement{i_j}$:
    \begin{equation}
        \label{equation:reduction_divisible_only_if_validator_attack_cost_ij_bounds}
        \sspElement{i_j}
        \underset{\eqref{equation:reduction_divisible_only_if_validator_attack_cost_ij_lower_bound}}{\leq} \validatorAttackCost{\validator_{i_j}}{\allAttackStakes}
        \underset{\eqref{equation:validator_attack_cost}}{\leq} \stake{\validator_{i_j}}
        \underset{\eqref{equation:reduction_divisible_stake}}{=} \sspElement{i_j} ;
    \end{equation}
    This implies
    \begin{equation}
        \label{equation:reduction_divisible_only_if_validator_attack_cost_ij}
        \validatorAttackCost{\validator_{i_j}}{\allAttackStakes} = \sspElement{i_j} .
    \end{equation}
    The total cost of the attack is the sum of the costs of each participating validator, and it is lower bounded by summing the costs of validators in~$\{ \validator_{i_1}, \ldots, \validator_{i_k} \}$:
    \begin{equation}
        \label{equation:reduction_divisible_only_if_total_attack_cost}
        \attackCost{\allAttackStakes}
        \underset{\eqref{equation:total_attack_cost}}{=}
        \sum_{\validator \in \allValidators} \validatorAttackCost{\validator}{\allAttackStakes}
        \geq \sum_{j=1}^k \validatorAttackCost{\validator_{i_j}}{\allAttackStakes} \\
        \underset{\eqref{equation:reduction_divisible_only_if_validator_attack_cost_ij}}{=} \sum_{j=1}^k \sspElement{i_j} .
    \end{equation}

    Assume towards a contradiction that~$\service_{\sspElementCount + 1}$ is not attacked, namely,
    \begin{equation}
        \label{equation:reduction_divisible_only_if_contradiction_attacked_services}
        \attackedServices
        = \allServicesAt{I}
        = \{ \service_{i_1}, \ldots, \service_{i_k} \} .
    \end{equation}
    If we consider the prize of the attack, we get
    \begin{equation}
        \label{equation:reduction_divisible_only_if_contradiction_total_attack_prize}
        \totalAttackPrize{\allAttackStakes}
        \underset{\eqref{equation:total_attack_prize}}{=} \sum_{\service \in \attackedServices} \attackPrize{\service}
        \underset{\eqref{equation:reduction_divisible_only_if_contradiction_attacked_services}}{=} \sum_{j=1}^k \attackPrize{\service_{i_j}}
        \underset{\eqref{equation:reduction_divisible_attack_prize_si}}{=} \sum_{j=1}^k \frac{\sspElement{i_j}}{2}
        = \frac{1}{2} \cdot \sum_{j=1}^k \sspElement{i_j} .
    \end{equation}
    Due to the attack being profitable, by Definition~\ref{definition:attack_profitability},
    \begin{equation}
        \label{equation:reduction_divisible_only_if_contradiction_attack_profitability}
        \attackCost{\allAttackStakes} \leq \totalAttackPrize{\allAttackStakes} .
    \end{equation}
    However, we have the following contradiction:
    \begin{equation}
        \attackCost{\allAttackStakes}
        \underset{\eqref{equation:reduction_divisible_only_if_total_attack_cost}}{\geq} \sum_{j=1}^k \sspElement{i_j}
        > \frac{1}{2} \cdot \sum_{j=1}^k \sspElement{i_j}
        \underset{\eqref{equation:reduction_divisible_only_if_contradiction_total_attack_prize}}{=} \totalAttackPrize{\allAttackStakes}
        \underset{\eqref{equation:reduction_divisible_only_if_contradiction_attack_profitability}}{\geq} \attackCost{\allAttackStakes} .
    \end{equation}
    Therefore, it must be that~$\service_{\sspElementCount + 1}$ is attacked, and it holds that
    \begin{equation}
        \label{equation:reduction_divisible_only_if_attacked_services}
        \attackedServices
        = \allServicesAt{I} \cup \{ \service_{\sspElementCount + 1} \}
        = \{ \service_{i_1}, \ldots, \service_{i_k} \} \cup \{ \service_{\sspElementCount + 1} \} .
    \end{equation}

    Denote by~$\allValidatorsAt{I}$ the set of validators~$\{ \validator_{i_1}, \ldots, \validator_{i_k} \}$:
    \begin{equation}
        \label{equation:reduction_divisible_only_if_all_validators_at_I}
        \allValidatorsAt{I} = \{ \validator_{i_1}, \ldots, \validator_{i_k} \} .
    \end{equation}

    Now, we prove that the subset~$\left\{ \sspElement{i_1}, \ldots, \sspElement{i_k} \right\}$ is a solution to the Subset Sum problem.
    As~$\service_{\sspElementCount + 1}$ is attacked, 
    \begin{equation}
        \label{equation:reduction_divisible_only_if_feasibility_sTarget}
        \attackThreshold{\service_{\sspElementCount + 1}} \cdot \sum_{\validator \in \allValidators} \allocation{\validator}{\service_{\sspElementCount + 1}}
        \leq \sum_{\validator \in \allValidators} \attackStake{\validator}{\service_{\sspElementCount + 1}} .
    \end{equation}
    Starting from the right-hand side of~\eqreft{equation:reduction_divisible_only_if_feasibility_sTarget} and using the new notation, we get
    \begin{multline}
        \label{equation:reduction_divisible_only_if_feasibility_sTarget_rhs}
        \sum_{\validator \in \allValidators} \attackStake{\validator}{\service_{\sspElementCount + 1}}
        = \sum_{\validator \in \allValidatorsAt{I}} \attackStake{\validator}{\service_{\sspElementCount + 1}} + \sum_{\validator \in \allValidators \setminus \allValidatorsAt{I}} \attackStake{\validator}{\service_{\sspElementCount + 1}} \\
        \underset{\eqref{equation:reduction_divisible_only_if_all_validators_at_I}}{=} \sum_{j=1}^k \attackStake{\validator_{i_j}}{\service_{\sspElementCount + 1}} + \sum_{\validator \in \allValidators \setminus \allValidatorsAt{I}} \attackStake{\validator}{\service_{\sspElementCount + 1}} \\
        \underset{\eqref{equation:reduction_divisible_only_if_attack_stake_ij}}{=} \sum_{j=1}^k \sspElement{i_j} + \sum_{\validator \in \allValidators \setminus \allValidatorsAt{I}} \attackStake{\validator}{\service_{\sspElementCount + 1}} .
    \end{multline}

    Now, by using the right-hand side of~\eqreft{equation:reduction_divisible_only_if_feasibility_sTarget_rhs} and continuing to develop its left-hand side, we get
    \begin{multline}
        \label{equation:reduction_divisible_only_if_feasibility_sTarget_developed}
        \sum_{j=1}^k \sspElement{i_j} + \sum_{\validator \in \allValidators \setminus \allValidatorsAt{I}} \attackStake{\validator}{\service_{\sspElementCount + 1}}
        \underset{\eqref{equation:reduction_divisible_only_if_feasibility_sTarget_rhs}}{=} \sum_{\validator \in \allValidators} \attackStake{\validator}{\service_{\sspElementCount + 1}} \\
        \underset{\eqref{equation:reduction_divisible_only_if_feasibility_sTarget}}{\geq} \attackThreshold{\service_{\sspElementCount + 1}} \cdot \sum_{\validator \in \allValidators} \allocation{\validator}{\service_{\sspElementCount + 1}}
        \underset{\eqref{equation:reduction_divisible_attack_threshold_sTarget}}{=} \frac{\sspTarget}{\sspElementSum} \cdot \sum_{\validator \in \allValidators} \allocation{\validator}{\service_{\sspElementCount + 1}} \\
        = \frac{\sspTarget}{\sspElementSum} \cdot \sum_{i=1}^{\sspElementCount} \allocation{\validator_i}{\service_{\sspElementCount + 1}}
        \underset{\eqref{equation:reduction_divisible_allocation}}{=} \frac{\sspTarget}{\sspElementSum} \cdot \sum_{i=1}^\sspElementCount \sspElement{i}
        \underset{\eqref{equation:reduction_divisible_element_sum}}{=} \frac{\sspTarget}{\sspElementSum} \cdot \sspElementSum
        = \sspTarget .
    \end{multline}

    Because the attack is profitable, by Definition~\ref{definition:attack_profitability},
    \begin{equation}
        \label{equation:reduction_divisible_only_if_profitability}
        \attackCost{\allAttackStakes} \leq \totalAttackPrize{\allAttackStakes} .
    \end{equation}
    We will individually develop both sides of this inequality, similarly to what we did before.
    We begin with the right-hand side of~\eqreft{equation:reduction_divisible_only_if_profitability}, to get
    \begin{multline}
        \label{equation:reduction_divisible_only_if_profitability_rhs}
        \totalAttackPrize{\allAttackStakes}
        \underset{\eqref{equation:total_attack_prize}}{=} \sum_{\service \in \attackedServices} \attackPrize{\service}
        \underset{\eqref{equation:reduction_divisible_only_if_attacked_services}}{=} \attackPrize{\service_{\sspElementCount + 1}} + \sum_{j=1}^k \attackPrize{\service_{i_j}} \\
        \underset{\eqref{equation:reduction_divisible_attack_prize_sTarget}, \eqref{equation:reduction_divisible_attack_prize_si}}{=} \frac{\sspTarget}{2} + \sum_{j=1}^k \frac{\sspElement{i_j}}{2}
        = \frac{\sspTarget}{2} + \frac{1}{2} \cdot \sum_{j=1}^k \sspElement{i_j} .
    \end{multline}
    Before developing the left-hand side of~\eqreft{equation:reduction_divisible_only_if_profitability}, we first lower-bound the attack cost of each validator~${\validator \in \allValidators}$.
    Recall that for~$\validator_{i_j} \in \allValidatorsAt{I}$, we have already calculated the attack cost~(\eqreft{equation:reduction_divisible_only_if_validator_attack_cost_ij}):
    \begin{equation}
        \label{equation:reduction_divisible_only_if_profitability_validator_attack_cost_ij_restated}
        \validatorAttackCost{\validator_{i_j}}{\attackedServices}{\allAttackStakes} = \sspElement{i_j} .
    \end{equation}
    For~$\validator \in \allValidators \setminus \allValidatorsAt{I}$, we have
    \begin{multline}
        \label{equation:reduction_divisible_only_if_profitability_validator_attack_cost_others}
        \validatorAttackCost{\validator}{\allAttackStakes}
        \underset{\eqref{equation:validator_attack_cost}}{=} \min \left( \stake{\validator}, \sum_{\service \in \allServices} \attackStake{\validator}{\service} \right)
        \underset{\eqref{equation:reduction_divisible_only_if_attacked_services}}{\geq} \min \left( \stake{\validator}, \attackStake{\validator}{\service_{\sspElementCount + 1}} \right) \\
        \underset{\eqref{equation:reduction_divisible_only_if_attack_stake_upper_bound}}{\geq} \attackStake{\validator}{\service_{\sspElementCount + 1}} .
    \end{multline}
    We are now ready to develop the left-hand side of~\eqreft{equation:reduction_divisible_only_if_profitability}.
    \begin{multline}
        \label{equation:reduction_divisible_only_if_profitability_developed}
        \frac{\sspTarget}{2} + \frac{1}{2} \cdot \sum_{j=1}^k \sspElement{i_j}
        \underset{\eqref{equation:reduction_divisible_only_if_profitability_rhs}}{=} \totalAttackPrize{\allAttackStakes}
        \underset{\eqref{equation:reduction_divisible_only_if_profitability}}{\geq}
        \attackCost{\allAttackStakes}
        \underset{\eqref{equation:total_attack_cost}}{=} \sum_{\validator \in \allValidators} \validatorAttackCost{\validator}{\allAttackStakes} \\
        = \sum_{\validator \in \allValidatorsAt{I}} \validatorAttackCost{\validator}{\allAttackStakes} + \sum_{\validator \in \allValidators \setminus \allValidatorsAt{I}} \validatorAttackCost{\validator}{\allAttackStakes} \\
        \underset{\eqref{equation:reduction_divisible_only_if_all_validators_at_I}}{=} \sum_{j=1}^k \validatorAttackCost{\validator_{i_j}}{\allAttackStakes} + \sum_{\validator \in \allValidators \setminus \allValidatorsAt{I}} \validatorAttackCost{\validator}{\allAttackStakes} \\
        \underset{\eqref{equation:reduction_divisible_only_if_profitability_validator_attack_cost_ij_restated}}{=} \sum_{j=1}^k \sspElement{i_j} + \sum_{\validator \in \allValidators \setminus \allValidatorsAt{I}} \validatorAttackCost{\validator}{\allAttackStakes}
        \underset{\eqref{equation:reduction_divisible_only_if_profitability_validator_attack_cost_others}}{\geq} \sum_{j=1}^k \sspElement{i_j} + \sum_{\validator \in \allValidators \setminus \allValidatorsAt{I}} \attackStake{\validator}{\service_{\sspElementCount + 1}} .
    \end{multline}
    Switching sides and multiplying by~$2$, we get
    \begin{align}
        \frac{\sspTarget}{2} + \frac{1}{2} \cdot \sum_{j=1}^k \sspElement{i_j}
        &\geq \sum_{j=1}^k \sspElement{i_j} + \sum_{\validator \in \allValidators \setminus \allValidatorsAt{I}} \attackStake{\validator}{\service_{\sspElementCount + 1}} \\
        \sspTarget + \sum_{j=1}^k \sspElement{i_j}
        &\geq 2 \cdot \sum_{j=1}^k \sspElement{i_j} + 2 \cdot \sum_{\validator \in \allValidators \setminus \allValidatorsAt{I}} \attackStake{\validator}{\service_{\sspElementCount + 1}} \\
        \label{equation:reduction_divisible_only_if_profitability_developed_before_collpse}
        \sspTarget - \sum_{\validator \in \allValidators \setminus \allValidatorsAt{I}} \attackStake{\validator}{\service_{\sspElementCount + 1}}
        &\geq \sum_{j=1}^k \sspElement{i_j} + \sum_{\validator \in \allValidators \setminus \allValidatorsAt{I}} \attackStake{\validator}{\service_{\sspElementCount + 1}}
    \end{align}
    Combining the last inequality with~\eqreft{equation:reduction_divisible_only_if_feasibility_sTarget_developed}, we get
    \begin{equation}
        \label{equation:reduction_divisible_only_if_collpse}
        \sspTarget - \sum_{\validator \in \allValidators \setminus \allValidatorsAt{I}} \attackStake{\validator}{\service_{\sspElementCount + 1}}
        \underset{\eqref{equation:reduction_divisible_only_if_profitability_developed_before_collpse}}{\geq} \sum_{j=1}^k \sspElement{i_j} + \sum_{\validator \in \allValidators \setminus \allValidatorsAt{I}} \attackStake{\validator}{\service_{\sspElementCount + 1}}
        \underset{\eqref{equation:reduction_divisible_only_if_feasibility_sTarget_developed}}{\geq} \sspTarget .
    \end{equation}
    This yields that
    \begin{equation}
        \sum_{\validator \in \allValidators \setminus \allValidatorsAt{I}} \attackStake{\validator}{\service_{\sspElementCount + 1}} \leq 0 ;
    \end{equation}
    But since it is the sum of non-negative terms, it must be that
    \begin{equation}
        \sum_{\validator \in \allValidators \setminus \allValidatorsAt{I}} \attackStake{\validator}{\service_{\sspElementCount + 1}} = 0 .
    \end{equation}
    Plugging this into~\eqreft{equation:reduction_divisible_only_if_collpse}, we get
    \begin{equation}
        \sspTarget - 0 \geq \sum_{j=1}^k \sspElement{i_j} + 0 \geq \sspTarget .
    \end{equation}
    We get~$\sum_{j=1}^k \sspElement{i_j} = \sspTarget$. 
    So, the subset~$\left\{ \sspElement{i_1}, \ldots, \sspElement{i_k} \right\}$ is a solution to the Subset Sum problem.

    Hence, determining whether a restaking network has a profitable allocation-divisible attack is NP-complete.
\end{proof}


\subsection{Proofs Deferred from Subsection~\ref{section:security_analysis:symmetric_case}}


\begin{proposition}[Proposition~\ref{proposition:consolidated_attack_cost} restated]
    \label{proposition:consolidated_attack_cost_appendix}
    Consider a symmetric restaking network~${\networkState = (\allValidators, \allServices, \allStakes, \allAllocations, \allAttackThresholds, \allAttackPrizes)}$, and a consolidated attack~$\allAttackStakesAt{c}$ that attacks the services~$\attackedServicesAt{c}$.
    Then, the cost of~$\allAttackStakesAt{c}$ is given by
    \begin{multline}
        \attackCost{\allAttackStakesAt{c}}
        = \floor{\symmetricAttackThreshold | \allValidators |} \cdot \min \left( \symmetricStake, \sum_{\service \in \attackedServicesAt{c}} \symmetricAllocation{\service} \right) \\
        + \min \left( \symmetricStake, \left( \symmetricAttackThreshold | \allValidators | - \floor{\symmetricAttackThreshold | \allValidators |} \right) \sum_{\service \in \attackedServicesAt{c}} \symmetricAllocation{\service} \right) .
    \end{multline}
\end{proposition}
\begin{proof}
    Since~$\allAttackStakesAt{c}$ is consolidated, for all services~$\service \in \attackedServicesAt{c}$ for all~$i \in \left\{ 1, \ldots, \floor{\symmetricAttackThreshold |\allValidators|} \right\}$, it holds that
    \begin{equation}
        \label{equation:consolidated_attack_cost_consolidated_attack_definition}
        \attackStakeAt{c}{\validator_i}{\service} = \begin{cases}
            \symmetricAllocation{\service} & \text{if } i \leq \floor{\symmetricAttackThreshold |\allValidators|} ; \\
            \left( \symmetricAttackThreshold |\allValidators| - \floor{\symmetricAttackThreshold |\allValidators|} \right) \symmetricAllocation{\service} & \text{if } i = \floor{\symmetricAttackThreshold |\allValidators|} + 1 ; \\
            0 & \text{otherwise} .
        \end{cases}
    \end{equation}
    Let us consider 3 cases.
    First,~$i \leq \floor{\symmetricAttackThreshold |\allValidators|}$.
    Then, the cost of validator~$\validator_i$ is
    \begin{equation}
        \validatorAttackCost{\validator_i}{\allAttackStakesAt{c}}
        \underset{\eqref{equation:validator_attack_cost}}{=} \min \left( \symmetricStake, \sum_{\service \in \attackedServicesAt{c}} \attackStakeAt{c}{\validator_i}{\service} \right)
        \underset{\eqref{equation:consolidated_attack_cost_consolidated_attack_definition}}{=} \min \left( \symmetricStake, \sum_{\service \in \attackedServicesAt{c}} \symmetricAllocation{\service} \right) .
    \end{equation}
    Second,~$i = \floor{\symmetricAttackThreshold |\allValidators|} + 1$.
    Then, the cost of validator~$\validator_i$ is
    \begin{multline}
        \validatorAttackCost{\validator_i}{\allAttackStakesAt{c}}
        \underset{\eqref{equation:validator_attack_cost}}{=} \min \left( \symmetricStake, \sum_{\service \in \attackedServicesAt{c}} \attackStakeAt{c}{\validator_i}{\service} \right) \\
        \underset{\eqref{equation:consolidated_attack_cost_consolidated_attack_definition}}{=} \min \left( \symmetricStake, \left( \symmetricAttackThreshold |\allValidators| - \floor{\symmetricAttackThreshold |\allValidators|} \right) \sum_{\service \in \attackedServicesAt{c}} \symmetricAllocation{\service} \right) .
    \end{multline}
    Third,~$i > \floor{\symmetricAttackThreshold |\allValidators|} + 1$.
    Then, the cost of validator~$\validator_i$ is
    \begin{equation}
        \validatorAttackCost{\validator_i}{\allAttackStakesAt{c}}
        \underset{\eqref{equation:validator_attack_cost}}{=} \min \left( \symmetricStake, \sum_{\service \in \attackedServicesAt{c}} \attackStakeAt{c}{\validator_i}{\service} \right)
        \underset{\eqref{equation:consolidated_attack_cost_consolidated_attack_definition}}{=} 0 .
    \end{equation}

    Therefore, when we sum the costs of all validators, we get
    \begin{multline}
        \attackCost{\allAttackStakesAt{c}}
        \underset{\eqref{equation:total_attack_cost}}{=} \sum_{\validator \in \allValidators} \validatorAttackCost{\validator}{\allAttackStakesAt{c}} \\
        = \sum_{i=1}^{\floor{\symmetricAttackThreshold |\allValidators|}} \min \left( \symmetricStake, \sum_{\service \in \attackedServicesAt{c}} \symmetricAllocation{\service} \right) \\
        + \min \left( \symmetricStake, \left( \symmetricAttackThreshold |\allValidators| - \floor{\symmetricAttackThreshold |\allValidators|} \right) \sum_{\service \in \attackedServicesAt{c}} \symmetricAllocation{\service} \right) \\
        = \floor{\symmetricAttackThreshold |\allValidators|} \cdot \min \left( \symmetricStake, \sum_{\service \in \attackedServicesAt{c}} \symmetricAllocation{\service} \right) \\
        + \min \left( \symmetricStake, \left( \symmetricAttackThreshold |\allValidators| - \floor{\symmetricAttackThreshold |\allValidators|} \right) \sum_{\service \in \attackedServicesAt{c}} \symmetricAllocation{\service} \right) .
    \end{multline}
    As desired.
\end{proof}

\begin{proposition}[Proposition~\ref{proposition:profitable_attack_consolidated} restated]
    \label{proposition:profitable_attack_consolidated_appendix}
    If there is a profitable attack in a symmetric network, then there is a profitable attack that is consolidated.
\end{proposition}

We break the proof into two propositions.
We begin with a proposition that an attack in a symmetric network can be tightened to one with reduced cost and equal total prize.
\begin{proposition}
    \label{proposition:tight_attack}
    Consider a symmetric restaking network~${\networkState = (\allValidators, \allServices, \allStakes, \allAllocations, \allAttackThresholds, \allAttackPrizes)}$.
    Let~$\allAttackStakes$ be an attack in~$\networkState$.
    Then, there exists a tight attack~$\allAttackStakesAt{t}$ in~$\networkState$ such that~$\attackCost{\allAttackStakesAt{t}} \leq \attackCost{\allAttackStakes}$ and~$\totalAttackPrize{\allAttackStakesAt{t}} = \totalAttackPrize{\allAttackStakes}$.
\end{proposition}
\begin{proof}
    Take~$\allAttackStakes$ and for each service~$\service \in \attackedServices$, calculate the unnecessary stake
    \begin{equation}
        \textit{excess}(\service) = \symmetricAttackThreshold \cdot | \allValidators | \cdot \symmetricAllocation{\service} - \sum_{\validator \in \allValidators} \attackStake{\validator}{\service} .
    \end{equation}
    Then iterate over validators and reduce a total of this amount from the stake they use to attack~$\service$.
    For all services~$\service \in \allServices \setminus \attackedServices$, zero the attack stake.
    Denote the result by~$\allAttackStakesAt{t}$.
    By construction, for all services~$\service \in \attackedServices$, we have
    \begin{equation}
        \sum_{\validator \in \allValidators} \attackStakeAt{t}{\validator}{\service} = \sum_{\validator \in \allValidators} \attackStake{\validator}{\service} - \textit{excess}(\service)
        = \symmetricAttackThreshold \cdot | \allValidators | \cdot \symmetricAllocation{\service} ;
    \end{equation}
    and for all services~$\service \in \allServices \setminus \attackedServices$, we have
    \begin{equation}
        \sum_{\validator \in \allValidators} \attackStakeAt{t}{\validator}{\service} = 0 .
    \end{equation}
    For any $\service \in \attackedServices$ it holds that
    \begin{equation}
        \sum_{\validator \in \allValidators} \attackStakeAt{t}{\validator}{\service} = \symmetricAttackThreshold \cdot | \allValidators | \cdot \symmetricAllocation{\service}
        = \symmetricAttackThreshold \cdot \sum_{\validator \in \allValidators} \allocation{\validator}{\service} .
    \end{equation}
    Therefore,~$\service \in \attackedServicesAt{t}$.
    Similarly, for all services~$\service \in \allServices \setminus \attackedServices$, we have that~$\service \not\in \attackedServicesAt{t}$.
    Overall, we have
    \begin{equation}
        \label{equation:tight_attack_attacked_services}
        \attackedServicesAt{t} = \attackedServices .
    \end{equation}
    Hence,~$\allAttackStakesAt{t}$ is tight.

    By construction, since we only reduced attack stake, we have for all validators~$\validator \in \allValidators$ and services~$\service \in \allServices$
    \begin{equation}
        \label{equation:tight_attack_stake_comparison}
        \attackStakeAt{t}{\validator}{\service} \leq \attackStake{\validator}{\service} .
    \end{equation}
    Therefore,
    \begin{multline}
        \attackCost{\allAttackStakesAt{t}}
        \underset{\eqref{equation:total_attack_cost}}{=} \sum_{\validator \in \allValidators} \validatorAttackCost{\validator}{\allAttackStakesAt{t}}
        \underset{\eqref{equation:validator_attack_cost}}{=} \sum_{\validator \in \allValidators} \min \left( \stake{\validator}, \sum_{\service \in \allServices} \attackStakeAt{t}{\validator}{\service} \right) \\
        \underset{\eqref{equation:tight_attack_stake_comparison}}{\leq} \sum_{\validator \in \allValidators} \min \left( \stake{\validator}, \sum_{\service \in \allServices} \attackStake{\validator}{\service} \right)
        \underset{\eqref{equation:validator_attack_cost}}{=} \sum_{\validator \in \allValidators} \validatorAttackCost{\validator}{\allAttackStakes} \\
        \underset{\eqref{equation:total_attack_cost}}{=} \attackCost{\allAttackStakes} .
    \end{multline}
    Furthermore, we have
    \begin{equation}
        \label{equation:tight_attack_total_prize}
        \totalAttackPrize{\allAttackStakesAt{t}}
        \underset{\eqref{equation:total_attack_prize}}{=} \sum_{\service \in \attackedServicesAt{t}} \attackPrize{\service}
        \underset{\eqref{equation:tight_attack_attacked_services}}{=} \sum_{\service \in \attackedServices} \attackPrize{\service}
        \underset{\eqref{equation:total_attack_prize}}{=} \totalAttackPrize{\allAttackStakes} .
    \end{equation}
    Therefore,~$\allAttackStakesAt{t}$ is a tight attack with reduced cost and equal total prize.
\end{proof}

Before showing that a tight attack can be consolidated into another attack with the same prize but lower cost, we show that shifting attack stake from a validator who uses less stake to one who already uses more stake results in a lower total cost.
\begin{lemma}
    \label{lemma:attack_stake_shifting}
    Consider a symmetric restaking network~$\networkState=(\allValidators, \allServices, \allStakes, \allAllocations, \allAttackThresholds, \allAttackPrizes)$ in which there are two validators~$\validator_1$ and~$\validator_2$ with equal stake:
    \begin{equation}
        \label{equation:attack_stake_shifting_equal_stake}
        \stake{\validator_1} = \stake{\validator_2} .
    \end{equation}
    Let~$\allAttackStakesAt{1}$ be an attack where validator~$\validator_1$ uses more stake than validator~$\validator_2$:
    \begin{equation}
        \label{equation:attack_stake_shifting_v1_has_more_stake}
        \sum_{\service \in \allServices} \attackStakeAt{1}{\validator_1}{\service} \geq \sum_{\service \in \allServices} \attackStakeAt{1}{\validator_2}{\service} .
    \end{equation}
    Consider another attack~$\allAttackStakesAt{2}$ where we shift some stake from~$\validator_2$ to~$\validator_1$ and hold everything else equal, that is, for all services~$\service \in \allServices$, we have
    \begin{align}
        \label{equation:attack_stake_shifting_equal_stake_rest}
        \forall \validator \in \allValidators \setminus \{ \validator_1, \validator_2 \} , \attackStakeAt{2}{\validator}{\service} &= \attackStakeAt{1}{\validator}{\service} , \\
        \label{equation:attack_stake_shifting_v1_gets_stake}
        \attackStakeAt{2}{\validator_1}{\service} &\geq \attackStakeAt{1}{\validator_1}{\service} , \\
        \label{equation:attack_stake_shifting_v2_loses_stake}
        \attackStakeAt{2}{\validator_2}{\service} &\leq \attackStakeAt{1}{\validator_2}{\service} , \text{and} \\
        \label{equation:attack_stake_shifting_stake_conservation}
        \attackStakeAt{1}{\validator_1}{\service} + \attackStakeAt{1}{\validator_2}{\service} &= \attackStakeAt{2}{\validator_1}{\service} + \attackStakeAt{2}{\validator_2}{\service} .
    \end{align}
    Then, the total cost of~$\allAttackStakesAt{1}$ is lower than the total cost of~$\allAttackStakesAt{2}$:
    \begin{equation}
        \attackCost{\allAttackStakesAt{1}} \leq \attackCost{\allAttackStakesAt{2}} .
    \end{equation}
\end{lemma}
\begin{proof}
    Consider two cases.
    First, assume that
    \begin{equation}
        \label{equation:attack_stake_shifting_proof_case_1}
        \sum_{\service \in \allServices} \attackStakeAt{1}{\validator_1}{\service}
        \geq \stake{\validator_1} .
    \end{equation}
    It also implies that
    \begin{equation}
        \label{equation:attack_stake_shifting_proof_case_1_v1_gets_stake}
        \sum_{\service \in \allServices} \attackStakeAt{2}{\validator_1}{\service}
        \underset{\eqref{equation:attack_stake_shifting_v1_gets_stake}}{\geq} \sum_{\service \in \allServices} \attackStakeAt{1}{\validator_1}{\service}
        \underset{\eqref{equation:attack_stake_shifting_proof_case_1}}{\geq} \stake{\validator_1} .
    \end{equation}
    The attack cost of validator~$\validator_1$ in~$\allAttackStakesAt{1}$ is
    \begin{equation}
        \label{equation:attack_stake_shifting_proof_case_1_v1_cost_before}
        \validatorAttackCost{\validator_1}{\allAttackStakesAt{1}}
        \underset{\eqref{equation:validator_attack_cost}}{=} \min \left( \stake{\validator_1}, \sum_{\service \in \allServices} \attackStakeAt{1}{\validator_1}{\service} \right)
        \underset{\eqref{equation:attack_stake_shifting_proof_case_1}}{=} \stake{\validator_1} .
    \end{equation}
    The attack cost of validator~$\validator_1$ in~$\allAttackStakesAt{2}$ is
    \begin{equation}
        \label{equation:attack_stake_shifting_proof_case_1_v1_cost_after}
        \validatorAttackCost{\validator_1}{\allAttackStakesAt{2}}
        \underset{\eqref{equation:validator_attack_cost}}{=} \min \left( \stake{\validator_1}, \sum_{\service \in \allServices} \attackStakeAt{2}{\validator_1}{\service} \right)
        \underset{\eqref{equation:attack_stake_shifting_proof_case_1_v1_gets_stake}}{=} \stake{\validator_1} .
    \end{equation}
    Now, for validator~$\validator_2$, we have
    \begin{multline}
        \label{equation:attack_stake_shifting_proof_case_1_v2_cost_comparison}
        \validatorAttackCost{\validator_2}{\allAttackStakesAt{1}}
        \underset{\eqref{equation:validator_attack_cost}}{=} \min \left( \stake{\validator_2}, \sum_{\service \in \allServices} \attackStakeAt{1}{\validator_2}{\service} \right) \\
        \underset{\eqref{equation:attack_stake_shifting_v2_loses_stake}}{\geq} \min \left( \stake{\validator_2}, \sum_{\service \in \allServices} \attackStakeAt{2}{\validator_2}{\service} \right)
        \underset{\eqref{equation:validator_attack_cost}}{=} \validatorAttackCost{\validator_2}{\allAttackStakesAt{2}} .
    \end{multline}
    Overall, we see that
    \begin{multline}
        \label{equation:attack_stake_shifting_proof_case_1_cost_comparison}
        \validatorAttackCost{\validator_1}{\allAttackStakesAt{1}} + \validatorAttackCost{\validator_2}{\allAttackStakesAt{1}}
        \underset{\eqref{equation:attack_stake_shifting_proof_case_1_v1_cost_before}}{=} \stake{\validator_1} + \validatorAttackCost{\validator_2}{\allAttackStakesAt{1}} \\
        \underset{\eqref{equation:attack_stake_shifting_proof_case_1_v1_cost_after}}{=} \validatorAttackCost{\validator_1}{\allAttackStakesAt{2}} + \validatorAttackCost{\validator_2}{\allAttackStakesAt{1}}
        \underset{\eqref{equation:attack_stake_shifting_proof_case_1_v2_cost_comparison}}{\geq} \validatorAttackCost{\validator_1}{\allAttackStakesAt{2}} + \validatorAttackCost{\validator_2}{\allAttackStakesAt{2}}
    \end{multline}

    Next, consider the case where
    \begin{equation}
        \label{equation:attack_stake_shifting_proof_case_2}
        \sum_{\service \in \allServices} \attackStakeAt{1}{\validator_1}{\service}
        < \stake{\validator_1} .
    \end{equation}
    It also implies that
    \begin{multline}
        \label{equation:attack_stake_shifting_proof_case_2_v2_loses_stake}
        \sum_{\service \in \allServices} \attackStakeAt{2}{\validator_2}{\service}
        \underset{\eqref{equation:attack_stake_shifting_v2_loses_stake}}{\leq} \sum_{\service \in \allServices} \attackStakeAt{1}{\validator_2}{\service}
        \underset{\eqref{equation:attack_stake_shifting_v1_has_more_stake}}{\leq} \sum_{\service \in \allServices} \attackStakeAt{1}{\validator_1}{\service}
        \underset{\eqref{equation:attack_stake_shifting_proof_case_2}}{<} \stake{\validator_1} \\
        \underset{\eqref{equation:attack_stake_shifting_equal_stake}}{=} \stake{\validator_2} .
    \end{multline}
    The attack cost of validator~$\validator_1$ in~$\allAttackStakesAt{1}$ is
    \begin{equation}
        \label{equation:attack_stake_shifting_proof_case_2_v1_cost_before}
        \validatorAttackCost{\validator_1}{\allAttackStakesAt{1}}
        \underset{\eqref{equation:validator_attack_cost}}{=} \min \left( \stake{\validator_1}, \sum_{\service \in \allServices} \attackStakeAt{1}{\validator_1}{\service} \right)
        \underset{\eqref{equation:attack_stake_shifting_proof_case_2}}{<} \sum_{\service \in \allServices} \attackStakeAt{1}{\validator_1}{\service} .
    \end{equation}
    The attack cost of validator~$\validator_1$ in~$\allAttackStakesAt{2}$ is
    \begin{equation}
        \label{equation:attack_stake_shifting_proof_case_2_v1_cost_after}
        \validatorAttackCost{\validator_1}{\allAttackStakesAt{2}}
        \underset{\eqref{equation:validator_attack_cost}}{=} \min \left( \stake{\validator_1}, \sum_{\service \in \allServices} \attackStakeAt{2}{\validator_1}{\service} \right)
    \end{equation}
    The attack cost of validator~$\validator_2$ in~$\allAttackStakesAt{1}$ is
    \begin{equation}
        \label{equation:attack_stake_shifting_proof_case_2_v2_cost_before}
        \validatorAttackCost{\validator_2}{\allAttackStakesAt{1}}
        \underset{\eqref{equation:validator_attack_cost}}{=} \min \left( \stake{\validator_2}, \sum_{\service \in \allServices} \attackStakeAt{1}{\validator_2}{\service} \right)
        \underset{\eqref{equation:attack_stake_shifting_proof_case_2_v2_loses_stake}}{=} \sum_{\service \in \allServices} \attackStakeAt{1}{\validator_2}{\service}
    \end{equation}
    And the attack cost of validator~$\validator_2$ in~$\allAttackStakesAt{2}$ is
    \begin{equation}
        \label{equation:attack_stake_shifting_proof_case_2_v2_cost_after}
        \validatorAttackCost{\validator_2}{\allAttackStakesAt{2}}
        \underset{\eqref{equation:validator_attack_cost}}{=} \min \left( \stake{\validator_2}, \sum_{\service \in \allServices} \attackStakeAt{2}{\validator_2}{\service} \right)
        \underset{\eqref{equation:attack_stake_shifting_proof_case_2_v2_loses_stake}}{=} \sum_{\service \in \allServices} \attackStakeAt{2}{\validator_2}{\service}
    \end{equation}
    Using the fact the sum of allocations is preserved, we get
    \begin{multline}
        \label{equation:attack_stake_shifting_proof_case_2_cost_comparison}
        \validatorAttackCost{\validator_1}{\allAttackStakesAt{1}} + \validatorAttackCost{\validator_2}{\allAttackStakesAt{1}}
        \underset{\eqref{equation:attack_stake_shifting_proof_case_2_v1_cost_before}, \eqref{equation:attack_stake_shifting_proof_case_2_v2_cost_before}}{=} \sum_{\service \in \allServices} \attackStakeAt{1}{\validator_1}{\service} + \sum_{\service \in \allServices} \attackStakeAt{1}{\validator_2}{\service} \\
        = \sum_{\service \in \allServices} \left( \attackStakeAt{1}{\validator_1}{\service} + \attackStakeAt{1}{\validator_2}{\service} \right)
        \underset{\eqref{equation:attack_stake_shifting_stake_conservation}}{=} \sum_{\service \in \allServices} \left( \attackStakeAt{2}{\validator_1}{\service} + \attackStakeAt{2}{\validator_2}{\service} \right) \\
        = \sum_{\service \in \allServices} \attackStakeAt{2}{\validator_1}{\service} + \sum_{\service \in \allServices} \attackStakeAt{2}{\validator_2}{\service}
        \underset{\eqref{equation:attack_stake_shifting_proof_case_2_v2_cost_after}}{=} \sum_{\service \in \allServices} \attackStakeAt{2}{\validator_1}{\service} + \validatorAttackCost{\validator_2}{\allAttackStakesAt{2}} \\
        \geq \min \left( \stake{\validator_1}, \sum_{\service \in \allServices} \attackStakeAt{2}{\validator_1}{\service} \right) + \validatorAttackCost{\validator_2}{\allAttackStakesAt{2}} \\
        \underset{\eqref{equation:attack_stake_shifting_proof_case_2_v1_cost_after}}{=} \validatorAttackCost{\validator_1}{\allAttackStakesAt{2}} + \validatorAttackCost{\validator_2}{\allAttackStakesAt{2}} .
    \end{multline}

    Due to \eqreft{equation:attack_stake_shifting_proof_case_1_cost_comparison} and \eqreft{equation:attack_stake_shifting_proof_case_2_cost_comparison}, in both cases we have shown that 
    \begin{equation}
        \label{equation:attack_stake_shifting_proof_cost_comparison}
        \validatorAttackCost{\validator_1}{\allAttackStakesAt{1}} + \validatorAttackCost{\validator_2}{\allAttackStakesAt{1}}
        \geq \validatorAttackCost{\validator_1}{\allAttackStakesAt{2}} + \validatorAttackCost{\validator_2}{\allAttackStakesAt{2}} .
    \end{equation}
    In addition, since the only difference in allocations in the attacks is for validators~$\validator_1$ and~$\validator_2$, we have for all other validators~$\validator \in \allValidators \setminus \{ \validator_1, \validator_2 \}$
    \begin{multline}
        \label{equation:attack_stake_shifting_proof_other_validators_cost_comparison}
        \validatorAttackCost{\validator}{\allAttackStakesAt{1}} 
        \underset{\eqref{equation:validator_attack_cost}}{=} \min \left( \stake{\validator}, \sum_{\service \in \allServices} \attackStakeAt{1}{\validator}{\service} \right) \\
        \underset{\eqref{equation:attack_stake_shifting_equal_stake_rest}}{=} \min \left( \stake{\validator}, \sum_{\service \in \allServices} \attackStakeAt{2}{\validator}{\service} \right)
        \underset{\eqref{equation:validator_attack_cost}}{=} \validatorAttackCost{\validator}{\allAttackStakesAt{2}} .
    \end{multline}

    Combining with \eqreft{equation:attack_stake_shifting_proof_other_validators_cost_comparison}, we get that
    \begin{multline}
        \attackCost{\allAttackStakesAt{1}}
        \underset{\eqref{equation:total_attack_cost}}{=} \sum_{\validator \in \allValidators} \validatorAttackCost{\validator}{\allAttackStakesAt{1}} \\
        = \validatorAttackCost{\validator_1}{\allAttackStakesAt{1}} + \validatorAttackCost{\validator_2}{\allAttackStakesAt{1}} + \sum_{\validator \in \allValidators \setminus \{ \validator_1, \validator_2 \}} \validatorAttackCost{\validator}{\allAttackStakesAt{1}} \\
        \underset{\eqref{equation:attack_stake_shifting_proof_other_validators_cost_comparison}}{=} \validatorAttackCost{\validator_1}{\allAttackStakesAt{1}} + \validatorAttackCost{\validator_2}{\allAttackStakesAt{1}}+ \sum_{\validator \in \allValidators \setminus \{ \validator_1, \validator_2 \}} \validatorAttackCost{\validator}{\allAttackStakesAt{2}} \\
        \underset{\eqref{equation:attack_stake_shifting_proof_cost_comparison}}{\geq} \validatorAttackCost{\validator_1}{\allAttackStakesAt{2}} + \validatorAttackCost{\validator_2}{\allAttackStakesAt{2}} + \sum_{\validator \in \allValidators \setminus \{ \validator_1, \validator_2 \}} \validatorAttackCost{\validator}{\allAttackStakesAt{2}} \\
        = \sum_{\validator \in \allValidators} \validatorAttackCost{\validator}{\allAttackStakesAt{2}}
        \underset{\eqref{equation:total_attack_cost}}{=} \attackCost{\allAttackStakesAt{2}} .
    \end{multline}
    And therefore, the total cost of~$\allAttackStakesAt{1}$ is lower than that of~$\allAttackStakesAt{2}$.
\end{proof}

The following proposition uses the previous lemma to show that in a symmetric network, a tight attack can be consolidated into another attack with the same prize but lower cost.
\begin{proposition}
    \label{proposition:consolidated_attack}
    Consider a symmetric restaking network~${\networkState = (\allValidators, \allServices, \allStakes, \allAllocations, \allAttackThresholds, \allAttackPrizes)}$.
    Let~$\allAttackStakesAt{t}$ be a tight attack in~$\networkState$.
    Then, there exists a consolidated attack~$\allAttackStakesAt{c}$ in~$\networkState$ such that~$\attackCost{\allAttackStakesAt{c}} \leq \attackCost{\allAttackStakesAt{t}}$ and~${\totalAttackPrize{\allAttackStakesAt{c}} = \totalAttackPrize{\allAttackStakesAt{t}}}$.
\end{proposition}
\begin{proof}
    Take the attack~$\allAttackStakesAt{t}$ and find the validator with the smallest sum of attack stake~$\sum_{\service \in \allServices} \attackStakeAt{t}{\validator}{\service}$.
    Without loss of generality, assume it is~$\validator_{|\allValidators|}$.
    
    Now, iterate over~$i = |\allValidators|, |\allValidators| - 1, ..., 1$ in reverse order.
    For each~$i = 1, ..., |\allValidators|$, iterate over all validators~$\validator \in \{ \validator_1, ..., \validator_{i-1} \}$ in descending order by the sum of their attack stakes, namely,~$\sum_{\service \in \allServices} \attackStakeAt{t}{\validator}{\service}$.
    Without loss of generality, assume their order is~$\validator_1, ... \validator_{i-1}$.
    Take the attack stake of~$\validator_i$ from all services and give as much as possible to~$\validator_j$, until~$\validator_j$ is saturated or~$\validator_i$ has no more stake to give.
    If~$\validator_i$ still has some stake left, repeat the same process for~$\validator_{j+1}$.
    If~$\validator_i$ has no more stake to give, break and go to~$\validator_{i-1}$.
    After the process is done, we have a consolidated attack~$\allAttackStakesAt{c}$.
    This is due to the fact that the attack is tight, so the sum of attack costs for each service~$\service$ is exactly~$\symmetricAttackThreshold |\allValidators| \symmetricAllocation{\service}$.
    Thus, there are exactly~$\floor{\symmetricAttackThreshold |\allValidators|}$ validators that will be saturated and possibly another validator that will have some stake left.

    In the construction of the attack~$\allAttackStakesAt{c}$, we only shift stake from validator~$\validator_i$ to~$\validator_j$ such that~$j<i$.
    Because of the sorting process for each~$i$, it holds that~${\sum_{\service \in \allServices} \attackStakeAt{t}{\validator_j}{\service} \geq \sum_{\service \in \allServices} \attackStakeAt{t}{\validator_i}{\service}}$.
    Therefore, by Lemma~\ref{lemma:attack_stake_shifting}, each time we shift stake, the total cost of the attack does not increase and while the prize of the attack remains the same.
    Thus,~$\allAttackStakesAt{c}$ is a consolidated attack with the same prize but lower cost.
\end{proof}

We are now ready to prove Proposition~\ref{proposition:profitable_attack_consolidated_appendix}.
\begin{proof}[Proposition~\ref{proposition:profitable_attack_consolidated_appendix}]
    Let~$\allAttackStakes$ be a profitable attack in a symmetric network.
    This implies its prize is higher than its cost.
    By Proposition~\ref{proposition:tight_attack}, there exists a tight attack~$\allAttackStakesAt{t}$ with the same prize but lower cost.
    By Proposition~\ref{proposition:consolidated_attack}, there exists a consolidated attack~$\allAttackStakesAt{c}$ with the same prize but an even lower cost.
    Therefore,~$\allAttackStakesAt{c}$ is profitable.
\end{proof}


\section{Proofs Deferred from Section~\ref{section:theoretical_robustness_analysis}}
\label{appendix:proofs_from_section_theoretical_robustness_analysis}


\begin{proposition}[Proposition~\ref{proposition:restaking_network_robustness} restated]
    \label{proposition:restaking_network_robustness_appendix}
    A restaking network~$\networkState$ is~$\adversaryBudget$-cryptoeconomically robust if and only if there exists no~${\adversaryBudget}$-costly attack.
\end{proposition}
\begin{proof}
    We prove the proposition in two directions.
    
    \paragraphEmph{First direction}
    Assume that the network~$\networkState$ is~$\adversaryBudget$-cryptoeconomically robust.
    By definition, the strategy profile~$\allAttackStakes_0$, where for all~$\validator \in \allValidators$ and all~$\service \in \allServices$,~$\attackStake{\validator}{\service} = 0$, is a strong Nash equilibrium and under it there are no attacked services.
    We will show this implies that there is no~$\adversaryBudget$-costly attack.
    
    First, note that due to~\eqreft{equation:validator_attack_cost}, for all validators~$\validator \in \allValidators$ and attacks~$\allAttackStakes \in \strategyProfileSecurityGame$,~$\validatorAttackCost{\validator}{\allAttackStakes} \geq 0$.
    And due to~\eqreft{equation:total_attack_cost},
    \begin{equation}
        \label{equation:robustness_game_proof_first_direction_total_cost_more_than_validator_cost}
        \attackCost{\allAttackStakes} \geq \validatorAttackCost{\validator}{\allAttackStakes} \geq 0 .
    \end{equation}

    As in the security game, the cost of the attack is
    \begin{multline}
        \label{equation:robustness_game_proof_first_direction_attack_cost_0}
        \attackCost{\allAttackStakes_0}
        \underset{\eqref{equation:total_attack_cost}}{=} \sum_{\validator \in \allValidators} \attackCost{\validator}{\allAttackStakes_0}
        \underset{\eqref{equation:validator_attack_cost}}{=} \sum_{\validator \in \allValidators} \min \left( \stake{\validator}, \sum_{\service \in \allServices} \attackStakeAt{0}{\validator}{\service} \right) \\
        = \sum_{\validator \in \allValidators} \min \left( \stake{\validator}, \sum_{\service \in \allServices} 0 \right)
        = 0 .
    \end{multline}

    By Definition~\ref{definition:restaking_network_robustness}, we have
    \begin{equation}
        \label{equation:robustness_game_proof_first_direction_attacked_services_0}
        \attackedServicesAt{0} = \emptyset
    \end{equation}
    The utility of~$\validator$ under~$\allAttackStakes_0$ is
    \begin{multline}
        \label{equation:robustness_game_proof_first_direction_validator_utility_0_is_0}
        \validatorUtilitySecurityGame{\validator}{\allAttackStakesAt{0}}
        \underset{\eqref{equation:validator_utility_robustness_game}}{=} \begin{cases}
            \validatorPrizeShareSecurityGame{\validator}{\allAttackStakesAt{0}} \left( \totalAttackPrize{\allAttackStakesAt{0}} + \adversaryBudget \right) - \validatorAttackCost{\validator}{\allAttackStakesAt{0}} & \text{if } \attackedServicesAt{0} \neq \emptyset ; \\ 
            - \validatorAttackCost{\validator}{\allAttackStakesAt{0}} & \text{otherwise} ;
        \end{cases} \\
        \underset{\eqref{equation:robustness_game_proof_first_direction_total_cost_more_than_validator_cost}, \eqref{equation:robustness_game_proof_first_direction_attacked_services_0}}{=} - \validatorAttackCost{\validator}{\allAttackStakes}
        \underset{\eqref{equation:robustness_game_proof_first_direction_attack_cost_0}}{=} 0 .
    \end{multline}
    
    In addition, due to the definition of cryptoeconomic security~(Definition~\ref{definition:restaking_network_security}),~$\allAttackStakes_0$ is a strong Nash equilibrium of the security game of the network~$\networkState$.
    That means that for any strategy profile~${\allAttackStakes \neq \allAttackStakes_0}$, there exists a validator~$\validator \in \allValidators$ that is worse off under~$\allAttackStakes$ than under~$\allAttackStakes_0$, that is,
    \begin{equation}
        \label{equation:robustness_game_proof_first_direction_validator_utility_other_is_negative}
        \validatorUtilitySecurityGame{\validator}{\allAttackStakes} < \validatorUtilitySecurityGame{\validator}{\allAttackStakes_0}
        \underset{\eqref{equation:robustness_game_proof_first_direction_validator_utility_0_is_0}}{=} 0 .
    \end{equation}

    If~$\attackedServices = \emptyset$, then~$\allAttackStakes$ is not~$\adversaryBudget$-costly.
    Then, assume
    \begin{equation}
        \label{equation:robustness_game_proof_first_direction_attacked_services_not_empty}
        \attackedServices \neq \emptyset .
    \end{equation}

    Developing the utility of~$\validator$ under~$\allAttackStakes$, we get that
    \begin{multline}
        \validatorUtilitySecurityGame{\validator}{\allAttackStakes}
        \underset{\eqref{equation:validator_utility_robustness_game}}{=} \begin{cases}
            \validatorPrizeShareSecurityGame{\validator}{\allAttackStakes} \left( \totalAttackPrize{\allAttackStakes} + \adversaryBudget \right) - \validatorAttackCost{\validator}{\allAttackStakes} & \text{if } \attackedServices \neq \emptyset ; \\ 
            - \validatorAttackCost{\validator}{\allAttackStakes} & \text{otherwise} ;
        \end{cases} \\
        \underset{\eqref{equation:robustness_game_proof_first_direction_attacked_services_not_empty}}{=}\validatorPrizeShareSecurityGame{\validator}{\allAttackStakes} \left( \totalAttackPrize{\allAttackStakes} + \adversaryBudget \right) - \validatorAttackCost{\validator}{\allAttackStakes} \\
        \underset{\eqref{equation:validator_prize_share_security_game}}{=} \begin{cases}
            \frac{\validatorAttackCost{\validator}{\allAttackStakes}}{\attackCost{\allAttackStakes}} \cdot \left( \totalAttackPrize{\allAttackStakes} + \adversaryBudget \right) - \validatorAttackCost{\validator}{\allAttackStakes} & \text{if } \attackCost{\allAttackStakes} > 0 ; \\
            \frac{1}{|\allValidators|}  \cdot \left( \totalAttackPrize{\allAttackStakes} + \adversaryBudget \right) - \validatorAttackCost{\validator}{\allAttackStakes} & \text{if } \attackCost{\allAttackStakes} = 0 .
        \end{cases} \\
        \underset{\eqref{equation:robustness_game_proof_first_direction_total_cost_more_than_validator_cost}}{=} \begin{cases}
            \frac{\validatorAttackCost{\validator}{\allAttackStakes}}{\attackCost{\allAttackStakes}} \cdot \left( \totalAttackPrize{\allAttackStakes} + \adversaryBudget \right) - \validatorAttackCost{\validator}{\allAttackStakes} & \text{if } \attackCost{\allAttackStakes} > 0 ; \\
            \frac{1}{|\allValidators|}  \cdot \left( \totalAttackPrize{\allAttackStakes} + \adversaryBudget \right) & \text{if } \attackCost{\allAttackStakes} = 0 .
        \end{cases}
        \underset{\eqref{equation:robustness_game_proof_first_direction_validator_utility_other_is_negative}}{<} 0 .
    \end{multline}
    Since~$\left( \totalAttackPrize{\allAttackStakes} + \adversaryBudget \right) \geq 0$, for the last inequality to hold it must be that~$\validatorAttackCost{\validator}{\allAttackStakes} > 0$. 
    Hence,
    \begin{equation}
        \frac{\validatorAttackCost{\validator}{\allAttackStakes}}{\attackCost{\allAttackStakes}} \cdot \left( \totalAttackPrize{\allAttackStakes} + \adversaryBudget \right) - \validatorAttackCost{\validator}{\allAttackStakes} < 0 .
    \end{equation}
    And because~$\validatorAttackCost{\validator}{\allAttackStakes} \geq 0$, it must be that~$\attackCost{\allAttackStakes} > \totalAttackPrize{\allAttackStakes} + \adversaryBudget$.
    Thus,~$\allAttackStakes$ is not~$\adversaryBudget$-costly and there exists no~$\adversaryBudget$-costly attack in~$\networkState$.
    
    \paragraphEmph{Second direction}
    Assume there exists some~$\adversaryBudget$-costly attack~$\allAttackStakes$.
    We claim it is an alternative strategy profile where some coalition deviated, and it resulted with all of them being better off and thus the strategy profile~$\allAttackStakesAt{0}$ is not a strong Nash equilibrium, meaning the network is not secure.

    By Definition~\ref{definition:beta_costly_attack},
    \begin{equation}
        \label{equation:robustness_game_proof_second_direction_attacked_services_not_empty}
        \attackedServices \neq \emptyset,
    \end{equation}
    and
    \begin{equation}
        \label{equation:robustness_game_proof_first_direction_attack_cost_total_attack_prize}
        \attackCost{\allAttackStakes} \leq \totalAttackPrize{\allAttackStakes} + \adversaryBudget .
    \end{equation}
    Consider the utility of validator~$\validator$ resulting from the strategy profile~$\allAttackStakes$,
    \begin{multline}
        \label{equation:robustness_game_proof_first_direction_validator_utility}
        \validatorUtilitySecurityGame{\validator}{\allAttackStakes}
        \underset{\eqref{equation:validator_utility_robustness_game}}{=} \begin{cases}
            \validatorPrizeShareSecurityGame{\validator}{\allAttackStakes} \left( \totalAttackPrize{\allAttackStakes} + \adversaryBudget \right) - \validatorAttackCost{\validator}{\allAttackStakes} & \text{if } \attackedServices \neq \emptyset ; \\ 
            - \validatorAttackCost{\validator}{\allAttackStakes} & \text{otherwise} .
        \end{cases} \\
        \underset{\eqref{equation:robustness_game_proof_first_direction_attack_cost_total_attack_prize}}{=} \validatorPrizeShareSecurityGame{\validator}{\allAttackStakes} \left( \totalAttackPrize{\allAttackStakes} + \adversaryBudget \right) - \validatorAttackCost{\validator}{\allAttackStakes} \\
        \underset{\eqref{equation:validator_prize_share_security_game}}{=} \begin{cases}
            \frac{\validatorAttackCost{\validator}{\allAttackStakes}}{\attackCost{\allAttackStakes}} \cdot \left( \totalAttackPrize{\allAttackStakes} + \adversaryBudget \right) - \validatorAttackCost{\validator}{\allAttackStakes} & \text{if } \attackCost{\allAttackStakes} > 0 ; \\
            \frac{1}{|\allValidators|}  \cdot \left( \totalAttackPrize{\allAttackStakes} + \adversaryBudget \right) - \validatorAttackCost{\validator}{\allAttackStakes} & \text{if } \attackCost{\allAttackStakes} = 0 .
        \end{cases}
        \geq 0 ;
    \end{multline}
    in the first case it follows from~\eqreft{equation:robustness_game_proof_first_direction_attack_cost_total_attack_prize}, and in the second case it follows from the fact that~$\validatorAttackCost{\validator}{\allAttackStakes}$ must be zero if~${\attackCost{\allAttackStakes} = 0}$.

    Now consider the strategy profile~$\allAttackStakes_0$, where for all~$\validator \in \allValidators$ and all~$\service \in \allServices$,~$\attackStake{\validator}{\service} = 0$.
    For all~$\validator \in \allValidators$,
    \begin{equation}
        \validatorUtilitySecurityGame{\validator}{\allAttackStakes_0}
        \underset{\eqref{equation:robustness_game_proof_first_direction_validator_utility_0_is_0}}{=} 0
        \underset{\eqref{equation:robustness_game_proof_first_direction_validator_utility}}{\leq} \validatorUtilitySecurityGame{\validator}{\allAttackStakes} .
    \end{equation}
    Therefore, by Definition~\ref{definition:strong_nash_equilibrium}, the strategy profile~$\allAttackStakes_0$ is not a strong Nash equilibrium of the restaking network security game, as otherwise we must have had some validator~$\validator \in \allValidators$ such that~${\validatorUtilitySecurityGame{\validator}{\allAttackStakes_0} > \validatorUtilitySecurityGame{\validator}{\allAttackStakes}}$.
    Hence, the network is not $\adversaryBudget$-cryptoeconomically robust. 
\end{proof}

\begin{proposition}[Proposition~\ref{proposition:beta_costly_attack_consolidated} restated]
    \label{proposition:beta_costly_attack_consolidated_appendix}
    If there is a~$\adversaryBudget$-costly attack in a symmetric network, then there is a~$\adversaryBudget$-costly profitable attack that is consolidated.
\end{proposition}
\begin{proof}
    Let~$\allAttackStakes$ be a~$\adversaryBudget$-costly attack in a symmetric network.
    This implies that
    \begin{equation}
        \label{equation:beta_costly_attack_consolidated_costly_constraint}
        \attackCost{\allAttackStakes} \leq \totalAttackPrize{\allAttackStakes} + \adversaryBudget .
    \end{equation}
    By Proposition~\ref{proposition:tight_attack}, there exists a tight attack~$\allAttackStakesAt{t}$ such that~$\totalAttackPrize{\allAttackStakesAt{t}} = \totalAttackPrize{\allAttackStakes}$ and~$\attackCost{\allAttackStakesAt{t}} \leq \attackCost{\allAttackStakes}$.
    By Proposition~\ref{proposition:consolidated_attack}, there exists a consolidated attack~$\allAttackStakesAt{c}$ such that~$\totalAttackPrize{\allAttackStakesAt{c}} = \totalAttackPrize{\allAttackStakesAt{t}}$ and~$\attackCost{\allAttackStakesAt{c}} \leq \attackCost{\allAttackStakesAt{t}}$.
    Overall, we have
    \begin{equation}
        \label{equation:beta_costly_attack_consolidated_attack_cost_constraint}
        \attackCost{\allAttackStakesAt{c}} \leq \attackCost{\allAttackStakesAt{t}} \leq \attackCost{\allAttackStakes} ,
    \end{equation}
    and
    \begin{equation}
        \label{equation:beta_costly_attack_consolidated_prize_constraint}
        \totalAttackPrize{\allAttackStakesAt{c}} = \totalAttackPrize{\allAttackStakesAt{t}} = \totalAttackPrize{\allAttackStakes} .
    \end{equation}
    Starting from the cost of~$\allAttackPrizesAt{c}$, we get
    \begin{equation}
        \attackCost{\allAttackStakesAt{c}}
        \underset{\eqref{equation:beta_costly_attack_consolidated_attack_cost_constraint}}{\leq} \attackCost{\allAttackStakes}
        \underset{\eqref{equation:beta_costly_attack_consolidated_costly_constraint}}{\leq} \totalAttackPrize{\allAttackStakes} + \adversaryBudget
        \underset{\eqref{equation:beta_costly_attack_consolidated_prize_constraint}}{=} \totalAttackPrize{\allAttackStakesAt{c}} + \adversaryBudget .
    \end{equation}
    Therefore,~$\allAttackStakesAt{c}$ is~$\adversaryBudget$-costly.
\end{proof}

\begin{proposition}[Proposition~\ref{proposition:symmetric_network_remains_symmetric_after_byzantine_services_cause_slashing} restated]
    \label{proposition:symmetric_network_remains_symmetric_after_byzantine_services_cause_slashing_appendix}
    Consider a symmetric restaking network~$\networkStateAt{0} = (\allValidatorsAt{0}, \allServicesAt{0}, \allStakesAt{0}, \allAllocationsAt{0}, \allAttackThresholdsAt{0}, \allAttackPrizesAt{0})$ and a subset of Byzantine services~$\byzantineServices \subseteq \allServicesAt{0}$.
    Let~$\networkStateAt{1} = (\allValidatorsAt{1}, \allServicesAt{1}, \allStakesAt{1}, \allAllocationsAt{1}, \allAttackThresholdsAt{1}, \allAttackPrizesAt{1})$ be the restaking network that remains after the Byzantine services in~$\byzantineServices$ cause slashing.
    Then~$\networkStateAt{1}$ is symmetric.
\end{proposition}
\begin{proof}
    To show that~$\networkStateAt{1}$ is symmetric, we need to show that for all validators have equal stake, all allocations to a service~$\service \in \allServicesAt{1}$ are equal and that all attack thresholds are equal.
    By the way the slashing of Byzantine services is defined, the condition on attack thresholds is trivially satisfied.
    
    We first show that the stake is equal.
    For all validators~$\validator \in \allValidatorsAt{1}$,
    \begin{multline}
        \stakeAt{1}{\validator}
        \underset{\eqref{equation:stake_after_byzantine_services_cause_slashing}}{=} \max \left( 0, \stakeAt{0}{v} - \sum_{\service \in \byzantineServices} \allocationAt{0}{\validator}{\service} \right) \\
        = \max \left( 0, \symmetricStakeAt{0} - \sum_{\service \in \byzantineServices} \symmetricAllocationAt{0}{\service} \right) .
    \end{multline}
    Therefore, the stake is equal.

    We then show that the allocations are equal.
    For all validators~$\validator \in \allValidatorsAt{1}$ and all services~$\service \in \allServicesAt{1}$,
    \begin{equation}
        \allocationAt{1}{\validator}{\service}
        \underset{\eqref{equation:allocation_after_byzantine_services_cause_slashing}}{=} \min \left( \allocationAt{0}{\validator}{\service}, \stakeAt{1}{\validator} \right) = \min \left( \symmetricAllocationAt{0}{\service}, \symmetricStakeAt{1} \right) .
    \end{equation}
    Therefore, the allocations for~$\service$ are also equal.
    Hence, the network is symmetric.
\end{proof}

\begin{proposition}[Proposition~\ref{proposition:symmetric_network_less_robust_with_byzantine_services} restated]
    \label{proposition:symmetric_network_less_robust_with_byzantine_services_appendix}
    Consider a symmetric restaking network~$\networkStateAt{0} = (\allValidators, \allServicesAt{0}, \allStakesAt{0}, \allAllocationsAt{0}, \allAttackThresholds, \allAttackPrizes)$ in which there exist 2 services~$\service_1$ and~$\service_2$ such that~$\attackPrizeAt{0}{\service_1} = \attackPrizeAt{0}{\service_2}$ and~$\symmetricAllocationAt{0}{\service_1} = \symmetricAllocationAt{0}{\service_2}$.
    Let~$\networkStateAt{1} = (\allValidators, \allServicesAt{1}, \allStakesAt{1}, \allAllocationsAt{1}, \allAttackThresholds, \allAttackPrizes)$ be the restaking network that remains after slashing of one Byzantine service~$\service_1$ in~$\networkStateAt{0}$, that is, $\networkStateAt{1} = \networkStateAt{0} \networkAdvance \left\{ \service_1 \right\}$.
    Then, if~$\networkStateAt{1}$ is~$\adversaryBudget$-cryptoeconomically robust, then~$\networkStateAt{0}$ is~$\adversaryBudget$-cryptoeconomically robust.
\end{proposition}
\begin{proof}
    We prove the contrapositive.
    Assume~$\networkStateAt{0}$ is not~$\adversaryBudget$-cryptoeconomically robust.
    Then, there exists a~$\adversaryBudget$-costly attack~$\allAttackStakesAt{0}$ in~$\networkStateAt{0}$ such that~$\attackCost{\allAttackStakesAt{0}} \leq \totalAttackPrize{\allAttackStakesAt{0}} + \adversaryBudget$ and~$\attackedServicesAt{0} \neq \emptyset$.
    Assume that~$\networkStateAt{0}$ is consolidated, otherwise consolidate it and use that instead of~$\networkStateAt{0}$.

    First, let us consider the remaining stake and allocations in~$\networkStateAt{1} = \networkStateAt{0} \networkAdvance \left\{ \service_1 \right\}$.
    For all validators~$\validator \in \allValidatorsAt{1}$,
    \begin{multline}
        \label{equation:symmetric_network_less_robust_with_byzantine_services_stake_after_byzantine_services_cause_slashing}
        \stakeAt{1}{\validator}
        \underset{\eqref{equation:stake_after_byzantine_services_cause_slashing}}{=} \max \left( 0, \stakeAt{0}{v} - \sum_{\service \in \byzantineServices} \allocationAt{0}{\validator}{\service} \right)
        = \max \left( 0, \symmetricStakeAt{0} - \symmetricAllocationAt{0}{\service_1 } \right) \\
        = \symmetricStakeAt{0} - \symmetricAllocationAt{0}{\service_1} .
    \end{multline}
    For all validators~$\validator \in \allValidatorsAt{1}$ and all services~$\service \in \allServicesAt{1}$,
    \begin{multline}
        \label{equation:symmetric_network_less_robust_with_byzantine_services_allocation_after_byzantine_services_cause_slashing}
        \allocationAt{1}{\validator}{\service}
        \underset{\eqref{equation:allocation_after_byzantine_services_cause_slashing}}{=} \min \left( \allocationAt{0}{\validator}{\service}, \stakeAt{1}{\validator} \right) = \min \left( \symmetricAllocationAt{0}{\service}, \symmetricStakeAt{1} \right) \\
        = \min \left( \symmetricAllocationAt{0}{\service}, \symmetricStakeAt{0} - \symmetricAllocationAt{0}{\service_1} \right) .
    \end{multline}
    
    Now, Consider two cases.
    First, assume
    \begin{equation}
        \label{equation:symmetric_network_less_robust_with_byzantine_services_first_case}
        \attackedServicesAt{0} = \left\{ \service_1 \right\} .
    \end{equation}
    We show it implies that~$\networkStateAt{1}$ is not~$\adversaryBudget$-cryptoeconomically robust.
    
    Due to Proposition~\ref{proposition:consolidated_attack_cost}, the cost of~$\allAttackStakesAt{0}$ is
    \begin{multline}
        \label{equation:symmetric_network_less_robust_with_byzantine_services_first_case_cost}
        \attackCostAt{0}{\allAttackStakesAt{0}}
        = \floor{\symmetricAttackThreshold | \allValidators |} \cdot \min \left( \symmetricStakeAt{0}, \sum_{\service \in \attackedServicesAt{0}} \symmetricAllocationAt{0}{\service} \right) \\
        + \min \left( \symmetricStakeAt{0}, \left( \symmetricAttackThreshold | \allValidators | - \floor{\symmetricAttackThreshold | \allValidators |} \right) \sum_{\service \in \attackedServicesAt{0}} \symmetricAllocationAt{0}{\service} \right) \\
        \underset{\eqref{equation:symmetric_network_less_robust_with_byzantine_services_first_case}}{=} \floor{\symmetricAttackThreshold | \allValidators |} \cdot \min \left( \symmetricStakeAt{0}, \symmetricAllocationAt{0}{\service_1} \right) + \min \left( \symmetricStakeAt{0}, \left( \symmetricAttackThreshold | \allValidators | - \floor{\symmetricAttackThreshold | \allValidators |} \right) \symmetricAllocationAt{0}{\service_1} \right) \\
        = \floor{\symmetricAttackThreshold | \allValidators |} \cdot \symmetricAllocationAt{0}{\service_1} + \left( \symmetricAttackThreshold | \allValidators | - \floor{\symmetricAttackThreshold | \allValidators |} \right) \symmetricAllocationAt{0}{\service_1} \\
        = \symmetricAttackThreshold | \allValidators | \cdot \symmetricAllocationAt{0}{\service_1} .
    \end{multline}
    Since~$\allAttackStakesAt{0}$ targets only service~$\service_1$, we have~$\totalAttackPrizeAt{0}{\allAttackStakesAt{0}} = \attackPrize{\service_1}$.
    And because~$\allAttackStakesAt{0}$ is~$\adversaryBudget$-costly, we have
    \begin{equation}
        \label{equation:symmetric_network_less_robust_with_byzantine_services_first_case_prize_constraint}
        \attackPrize{\service_1} + \adversaryBudget \geq \attackCostAt{0}{\allAttackStakesAt{0}} 
        \underset{\eqref{equation:symmetric_network_less_robust_with_byzantine_services_first_case_cost}}{=} \symmetricAttackThreshold | \allValidators | \cdot \symmetricAllocationAt{0}{\service_1} .
    \end{equation}
    
    Consider the consolidated attack~$\allAttackStakesAt{1}$ that targets~$\service_2$ in network~$\networkStateAt{1}$.
    Due to Proposition~\ref{proposition:consolidated_attack_cost}, and developing similarly using the fact that only one service is attacked, we get:
    \begin{multline}
        \label{equation:symmetric_network_less_robust_with_byzantine_services_first_case_cost_1}
        \attackCostAt{1}{\allAttackStakesAt{1}}
        = \symmetricAttackThreshold | \allValidators | \cdot \symmetricAllocationAt{0}{\service_2}
        \underset{\eqref{equation:symmetric_network_less_robust_with_byzantine_services_allocation_after_byzantine_services_cause_slashing}}{=} \symmetricAttackThreshold | \allValidators | \cdot \min \left( \symmetricAllocationAt{0}{\service_2}, \symmetricStakeAt{0} - \symmetricAllocationAt{0}{\service_1} \right) \\
        = \symmetricAttackThreshold | \allValidators | \cdot \min \left( \symmetricAllocationAt{0}{\service_1}, \symmetricStakeAt{0} - \symmetricAllocationAt{0}{\service_1} \right) \\
        \leq \symmetricAttackThreshold | \allValidators | \cdot \symmetricAllocationAt{0}{\service_1}
        \underset{\eqref{equation:symmetric_network_less_robust_with_byzantine_services_first_case_prize_constraint}}{\leq} \attackPrize{\service_1} + \adversaryBudget .
    \end{multline}

    Since~$\allAttackStakes{1}$ targets only service~$\service_2$, we have~$\totalAttackPrizeAt{1}{\allAttackStakesAt{1}} = \attackPrize{\service_2}$.

    Combining what we have, we get
    \begin{equation}
        \totalAttackPrizeAt{1}{\allAttackStakesAt{1}} + \adversaryBudget
        \geq \attackPrize{\service_2} + \adversaryBudget
        = \attackPrize{\service_1} + \adversaryBudget
        \underset{\eqref{equation:symmetric_network_less_robust_with_byzantine_services_first_case_cost_1}}{\geq} \attackCostAt{1}{\allAttackStakesAt{1}} .
    \end{equation}
    Therefore,~$\allAttackStakesAt{1}$ is~$\adversaryBudget$-costly, and due to Proposition~\ref{proposition:restaking_network_robustness},~$\networkStateAt{1}$ is not~$\adversaryBudget$-cryptoeconomically robust.

    Now, consider the other case where
    \begin{equation}
        \label{equation:symmetric_network_less_robust_with_byzantine_services_second_case}
        \attackedServicesAt{0} \neq \left\{ \service_1 \right\} .
    \end{equation}
    Furthermore, denote by~$\attackedServicesAt{2}$ the attack which we used in the previous case, namely, the one where $\attackedServices = \left\{ \service_1 \right\}$.
    Assume that it is not~$\adversaryBudget$-costly.
    Otherwise, we can use the previous case with~$\attackedServicesAt{2}$ to deduce that~$\networkStateAt{1}$ is not~$\adversaryBudget$-cryptoeconomically robust.
    
    Now, we show that~$\networkStateAt{1}$ is not~$\adversaryBudget$-cryptoeconomically robust.
    First, since~$\attackedServicesAt{2}$ is not~$\adversaryBudget$-costly, we have
    \begin{equation}
        \attackCostAt{0}{\allAttackStakesAt{2}}
        > \totalAttackPrizeAt{0}{\allAttackStakesAt{2}} + \adversaryBudget
        \geq \totalAttackPrizeAt{0}{\allAttackStakesAt{2}} .
    \end{equation}
    As in the previous case, we have~$\totalAttackPrizeAt{0}{\allAttackStakesAt{2}} = \attackPrize{\service_2}$, and~$\attackCostAt{1}{\allAttackStakesAt{2}} = \symmetricAttackThreshold | \allValidators | \cdot \symmetricAllocationAt{0}{\service_1}$~(\eqreft{equation:symmetric_network_less_robust_with_byzantine_services_first_case_cost}).
    Therefore, we have
    \begin{equation}
        \label{equation:symmetric_network_less_robust_with_byzantine_services_second_case_prize_constraint}
        \symmetricAttackThreshold | \allValidators | \cdot \symmetricAllocationAt{0}{\service_1} > \attackPrize{\service_2} .
    \end{equation}
    
    Now, since~$\allAttackStakesAt{0}$ is a consolidated attack, for all~$i=1, ..., |\allValidators|$ and all services~$\service \in \attackedServicesAt{0}$, we have
    \begin{equation}
        \label{equation:symmetric_network_less_robust_with_byzantine_services_attack_stake_in_alpha_0}
        \attackStakeAt{0}{\validator_i}{\service}
        \underset{\eqref{equation:consolidated_attack}}{=} \begin{cases}
            \symmetricAllocationAt{0}{\service} & \text{if } i \leq \floor{\symmetricAttackThreshold |\allValidators|} ; \\
            \left( \symmetricAttackThreshold |\allValidators| - \floor{\symmetricAttackThreshold |\allValidators|} \right) \symmetricAllocationAt{0}{\service} & \text{if } i = \floor{\symmetricAttackThreshold |\allValidators|} + 1 ; \\
            0 & \text{otherwise} .
        \end{cases}
    \end{equation}
    And for all other services, we have
    \begin{equation}
        \attackStakeAt{0}{\validator_i}{\service} = 0 .
    \end{equation}
      
    Consider the attack~$\allAttackStakesAt{1}$ in network~$\networkStateAt{1}$, which is the same as~$\allAttackStakesAt{0}$, capped at their new allocations, and with service~$\service_1$ removed.
    We have for all validators~$\validator \in \allValidators$ and all services~$\service \in \allServicesAt{1} = \allServicesAt{0} \setminus \left\{ \service_1 \right\}$,
    \begin{equation}
        \label{equation:symmetric_network_less_robust_with_byzantine_services_attack_stake_in_alpha_1}
        \attackStakeAt{1}{\validator}{\service}
        = \min \left( \attackStakeAt{0}{\validator}{\service}, \symmetricAllocationAt{1}{\service} \right) .
    \end{equation}
    Due to~\eqreft{equation:symmetric_network_less_robust_with_byzantine_services_allocation_after_byzantine_services_cause_slashing}, we have for all~$i=1, ..., |\allValidators|$ and all services~$\service \in \attackedServicesAt{0} \setminus \left\{ \service_1 \right\}$,
    \begin{multline}
        \label{equation:symmetric_network_less_robust_with_byzantine_services_attack_stake_in_alpha_1_allocation}
        \attackStakeAt{1}{\validator_i}{\service} \\
        \underset{\eqref{equation:symmetric_network_less_robust_with_byzantine_services_allocation_after_byzantine_services_cause_slashing}}{=} \begin{cases}
            \min \left( \symmetricAllocationAt{0}{\service}, \symmetricAllocationAt{1}{\service} \right) & \text{if } i \leq \floor{\symmetricAttackThreshold |\allValidators|} ; \\
            \min \left( \left( \symmetricAttackThreshold |\allValidators| - \floor{\symmetricAttackThreshold |\allValidators|} \right) \symmetricAllocationAt{0}{\service}, \symmetricAllocationAt{1}{\service} \right) & \text{if } i = \floor{\symmetricAttackThreshold |\allValidators|} + 1 ; \\
            0 & \text{otherwise} ;
        \end{cases} \\
        = \begin{cases}
            \symmetricAllocationAt{1}{\service} & \text{if } i \leq \floor{\symmetricAttackThreshold |\allValidators|} ; \\
            \min \left( \left( \symmetricAttackThreshold |\allValidators| - \floor{\symmetricAttackThreshold |\allValidators|} \right) \symmetricAllocationAt{0}{\service}, \symmetricAllocationAt{1}{\service} \right) & \text{if } i = \floor{\symmetricAttackThreshold |\allValidators|} + 1 ; \\
            0 & \text{otherwise} .
        \end{cases}
    \end{multline}
    And for all other services, we have
    \begin{equation}
        \attackStakeAt{1}{\validator_i}{\service} = 0 .
    \end{equation}

    The cost of~$\allAttackStakesAt{1}$ is
    \begin{multline}
        \label{equation:symmetric_network_less_robust_with_byzantine_services_attack_cost_in_alpha_1}
        \attackCostAt{1}{\allAttackStakesAt{1}}
        \underset{\eqref{equation:total_attack_cost}}{=} \sum_{\validator \in \allValidators} \validatorAttackCostAt{1}{\allAttackStakesAt{1}}{\validator}
        \underset{\eqref{equation:validator_attack_cost}}{=} \sum_{\validator \in \allValidators} \min \left( \symmetricStakeAt{1}, \sum_{\service \in \allServicesAt{1}} \attackStakeAt{1}{\validator}{\service} \right) \\
        \underset{\eqref{equation:symmetric_network_less_robust_with_byzantine_services_attack_stake_in_alpha_1}}{\leq} \sum_{\validator \in \allValidators} \min \left( \symmetricStakeAt{1}, \sum_{\service \in \allServicesAt{1}} \attackStakeAt{0}{\validator}{\service} \right) \\
        = \sum_{\validator \in \allValidators} \min \left( \symmetricStakeAt{1}, \sum_{\service \in \allServicesAt{0}} \attackStakeAt{0}{\validator}{\service} - \attackStakeAt{0}{\validator}{\service_1} \right) \\
        \underset{\eqref{equation:symmetric_network_less_robust_with_byzantine_services_stake_after_byzantine_services_cause_slashing}}{=} \sum_{\validator \in \allValidators} \min \left( \symmetricStakeAt{0} - \symmetricAllocationAt{0}{\service_1}, \sum_{\service \in \allServicesAt{0}} \attackStakeAt{0}{\validator}{\service} - \attackStakeAt{0}{\validator}{\service_1} \right) \\
        \leq \sum_{\validator \in \allValidators} \min \left( \symmetricStakeAt{0} - \attackStakeAt{0}{\validator}{\service_1}, \sum_{\service \in \allServicesAt{0}} \attackStakeAt{0}{\validator}{\service} - \attackStakeAt{0}{\validator}{\service_1} \right) \\
        = \sum_{\validator \in \allValidators} \left( \min \left( \symmetricStakeAt{0}, \sum_{\service \in \allServicesAt{0}} \attackStakeAt{0}{\validator}{\service} \right) - \attackStakeAt{0}{\validator}{\service_1} \right) \\
        = \sum_{\validator \in \allValidators} \min \left( \symmetricStakeAt{0}, \sum_{\service \in \allServicesAt{0}} \attackStakeAt{0}{\validator}{\service} \right) - \sum_{\validator \in \allValidators} \attackStakeAt{0}{\validator}{\service_1} \\
        \underset{\eqref{equation:symmetric_network_less_robust_with_byzantine_services_attack_stake_in_alpha_0}}{=} \sum_{\validator \in \allValidators} \min \left( \symmetricStakeAt{0}, \sum_{\service \in \allServicesAt{0}} \attackStakeAt{0}{\validator}{\service} \right) - \symmetricAttackThreshold | \allValidators | \cdot \symmetricAllocationAt{0}{\service_1} \\
        \underset{\eqref{equation:validator_attack_cost}}{=} \sum_{\validator \in \allValidators} \validatorAttackCostAt{0}{\validator}{\allAttackStakesAt{0}} - \symmetricAttackThreshold | \allValidators | \cdot \symmetricAllocationAt{0}{\service_1}
        \underset{\eqref{equation:total_attack_cost}}{=} \attackCostAt{0}{\allAttackStakesAt{0}} - \symmetricAttackThreshold | \allValidators | \cdot \symmetricAllocationAt{0}{\service_1} \\
        \underset{\eqref{equation:symmetric_network_less_robust_with_byzantine_services_second_case_prize_constraint}}{<} \attackCostAt{0}{\allAttackStakesAt{0}} - \attackPrize{\service_2}
    \end{multline}

    Now, we derive the profit of~$\allAttackStakesAt{1}$.
    To do so, we need to find the attacked services in~$\allAttackStakesAt{1}$.
    And first calculate for all services~${\service \in \attackedServicesAt{0} \setminus \left\{ \service_1 \right\}}$
    \begin{multline}
        \sum_{\validator \in \allValidators} \attackStakeAt{1}{\validator}{\service} \\
        \underset{\eqref{equation:symmetric_network_less_robust_with_byzantine_services_attack_stake_in_alpha_1_allocation}}
        = \floor{\symmetricAttackThreshold |\allValidators|} \symmetricAllocationAt{1}{\service} + \min \left( \left( \symmetricAttackThreshold |\allValidators| - \floor{\symmetricAttackThreshold |\allValidators|} \right) \symmetricAllocationAt{0}{\service}, \symmetricAllocationAt{1}{\service} \right) \\ 
        \geq \floor{\symmetricAttackThreshold |\allValidators|} \symmetricAllocationAt{1}{\service} \\
        + \min \left( \left( \symmetricAttackThreshold |\allValidators| - \floor{\symmetricAttackThreshold |\allValidators|} \right) \symmetricAllocationAt{0}{\service}, \left( \symmetricAttackThreshold |\allValidators| - \floor{\symmetricAttackThreshold |\allValidators|} \right) \symmetricAllocationAt{1}{\service} \right) \\
        \geq \floor{\symmetricAttackThreshold |\allValidators|} \symmetricAllocationAt{1}{\service} + \left( \symmetricAttackThreshold |\allValidators| - \floor{\symmetricAttackThreshold |\allValidators|} \right) \symmetricAllocationAt{1}{\service}
        = \symmetricAttackThreshold |\allValidators| \symmetricAllocationAt{1}{\service} \\
        = \symmetricAttackThreshold \sum_{\validator \in \allValidators} \symmetricAllocationAt{1}{\service} .
    \end{multline}
    Hence,
    \begin{equation}
        \label{equation:symmetric_network_less_robust_with_byzantine_services_attacked_services_in_alpha_1}
        \attackedServicesAt{0} \setminus \left\{ \service_1 \right\} \subseteq \attackedServicesAt{1} .
    \end{equation}
    This implies that
    \begin{multline}
        \label{equation:symmetric_network_less_robust_with_byzantine_services_total_attack_prize_in_alpha_1}
        \totalAttackPrizeAt{0}{\allAttackStakesAt{0}}
        \underset{\eqref{equation:total_attack_prize}}{=} \sum_{\service \in \attackedServicesAt{0}} \attackPrize{\service}
        \leq \sum_{\service \in \attackedServicesAt{0} \setminus \left\{ \service_1 \right\}} \attackPrize{\service} + \attackPrize{\service_1} \\
        \underset{\eqref{equation:symmetric_network_less_robust_with_byzantine_services_attacked_services_in_alpha_1}}{\leq} \sum_{\service \in \attackedServicesAt{1}} \attackPrize{\service} + \attackPrize{\service_1}
        \underset{\eqref{equation:total_attack_prize}}{=} \totalAttackPrizeAt{1}{\allAttackStakesAt{1}} + \attackPrize{\service_1} .
    \end{multline}

    Now, recall that~$\attackedServicesAt{0} \neq \emptyset$.
    It implies that~$\attackedServicesAt{1} \neq \emptyset$ as well.
    It remains to show that~${\attackCostAt{1}{\allAttackStakes{1}} \leq \totalAttackPrizeAt{1}{\allAttackStakes{1}} + \adversaryBudget}$.
    For that we use the fact~$\allAttackStakesAt{0}$ is~$\adversaryBudget$-costly:
    \begin{equation}
        \label{equation:symmetric_network_less_robust_with_byzantine_services_attack_cost_in_alpha_0}
        \attackCostAt{0}{\allAttackStakesAt{0}} \leq \totalAttackPrizeAt{0}{\allAttackStakesAt{0}} + \adversaryBudget .
    \end{equation}
    
    We are now ready to show that~$\allAttackStakesAt{1}$ is~$\adversaryBudget$-costly:
    \begin{multline}
        \attackCostAt{1}{\allAttackStakesAt{1}}
        \underset{\eqref{equation:symmetric_network_less_robust_with_byzantine_services_attack_cost_in_alpha_1}}{<} \attackCostAt{0}{\allAttackStakesAt{0}} - \attackPrize{\service_2}
        \underset{\eqref{equation:symmetric_network_less_robust_with_byzantine_services_attack_cost_in_alpha_0}}{\leq}
        \totalAttackPrizeAt{0}{\allAttackStakesAt{0}} + \adversaryBudget - \attackPrize{\service_2} \\
        \underset{\eqref{equation:symmetric_network_less_robust_with_byzantine_services_total_attack_prize_in_alpha_1}}{\leq} \totalAttackPrizeAt{1}{\allAttackStakesAt{1}} + \attackPrize{\service_1} + \adversaryBudget - \attackPrize{\service_2}
        = \totalAttackPrizeAt{1}{\allAttackStakesAt{1}} + \adversaryBudget .
    \end{multline}
    Hence, we get that in this case too, the network~$\networkStateAt{1}$ is not~$\adversaryBudget$-cryptoeconomically robust.
    This concludes the proof.
\end{proof}


\section{Designing and solving the MIPs}
\label{section:mip_appendix}


We first formulate the problem of determining the minimum adversary budget required to attack a restaking network as a MIP~(\S\ref{section:mip_appendix:budget_only_robustness}).
Then, we formulate as a MIP the problem of determining the maximum fraction of Byzantine services such that the network remains secure given an adversary budget~(\S\ref{section:mip_appendix:budget_and_byzantine_robustness}).
Afterward, we present how we solve the MIPs~(\S\ref{section:mip_appendix:solving_the_mips}).


\subsection{MIP for Cryptoeconomic Robustness}
\label{section:mip_appendix:budget_only_robustness}


Given a restaking network~$\networkState = (\allValidators, \allServices, \allStakes, \allAllocations, \allAttackThresholds, \allAttackPrizes)$, where~$\allValidators = \left\{ \validator_1, \ldots, \validator_n \right\}$,~$\allServices = \left\{ \service_1, \ldots, \service_m \right\}$, we formulate the problem of determining whether there exists a~$\adversaryBudget$-costly allocation-divisible attack as a mixed-integer program.


\subsubsection{Variables}


For each~$j \in \left\{ 1, \ldots, m \right\}$, denote by~$\milpAttackedServiceVariable{j}$ the variable that is~$1$ if service~$\service_j$ is attacked, and~$0$ otherwise.

For each~$i \in \left\{ 1, \ldots, n \right\}$ and~$j \in \left\{ 1, \ldots, m \right\}$, denote by~$\milpAttackStakeVariable{i}{j}$ the variable that is the amount of stake of validator~$\validator_i$ that is allocated to service~$\service_j$.
It can take any value in~$\left[ 0, \allocation{\validator_i}{\service_j} \right]$.

For each~$i \in \left\{ 1, \ldots, n \right\}$, denote by~$\milpValidatorCostVariable{i}$ the variable that is the cost of validator~$\validator_i$ in the attack, namely, the minimum between the stake used by the validator to attack and their stake.
It can take any value in~$\left[ 0, \stake{\validator_i} \right]$.
For each~$i \in \left\{ 1, \ldots, n \right\}$ we introduce an auxiliary variable~$\milpValidatorCostAuxiliaryVariable{i}$ that takes values in~$\left\{ 0, 1 \right\}$.
It will be used to calculate the attack cost of validators.


\subsubsection{Constraints}


First, as at least one service must be attacked, we have
\begin{equation}
    \sum_{j=1}^m \milpAttackedServiceVariable{j}
    \geq 1 .
\end{equation}

Denote by~$\milpLargeNumberFeasibility$ a large number used to make the constraints for having sufficient stake to attack apply only to attacked services.
For an attack have sufficient stake, it must be that for each~$j \in \left\{ 1, \ldots, m \right\}$
\begin{equation}
    \sum_{i=1}^n \milpAttackStakeVariable{i}{j}
    \geq \attackThreshold{\service_j} \cdot \sum_{i=1}^n \allocation{\validator_i}{\service_j} - \milpLargeNumberFeasibility \cdot (1 - \milpAttackedServiceVariable{j}) .
\end{equation}
This way, if service~$\service_j$ is not attacked, the constraint is trivially satisfied, and if it is attacked, the constraint ensures that the attack has enough stake.
For this to hold, we must have~$\milpLargeNumberFeasibility \geq \attackThreshold{\service_j} \cdot \sum_{i=1}^n \allocation{\validator_i}{\service_j}$ for all~$j \in \left\{ 1, \ldots, m \right\}$.

The attack cost of a validator~$\validator_i$ is~$\min \left( \stake{\validator_i}, \sum_{j=1}^m \milpAttackStakeVariable{i}{j} \right)$.
Also, denote by~$\milpLargeNumberCost$ a large number used to calculate the attack cost of validators.
We then introduce the following constraints:
\begin{align}
    \milpValidatorCostVariable{i}
    &\leq \stake{\validator_i} , \\
    \milpValidatorCostVariable{i}
    &\leq \sum_{j=1}^m \milpAttackStakeVariable{i}{j} , \\
    \label{equation:milp_validator_cost_lower_bound_validator_stake}
    \milpValidatorCostVariable{i}
    &\geq \stake{\validator_i} - \milpLargeNumberCost \cdot \milpValidatorCostAuxiliaryVariable{i} , \\
    \label{equation:milp_validator_cost_lower_bound_attack_stake}
    \milpValidatorCostVariable{i}
    &\geq \sum_{j=1}^m \milpAttackStakeVariable{i}{j} - \milpLargeNumberCost \cdot (1 - \milpValidatorCostAuxiliaryVariable{i}) .
\end{align}
This way, if~$\milpValidatorCostAuxiliaryVariable{i} = 0$,~\eqreft{equation:milp_validator_cost_lower_bound_validator_stake} ensures that the attack cost of validator~$\validator_i$ must be equal to~$\stake{\validator_i}$ and~\eqreft{equation:milp_validator_cost_lower_bound_attack_stake} is trivially satisfied; and if~$\milpValidatorCostAuxiliaryVariable{i} = 1$,~\eqreft{equation:milp_validator_cost_lower_bound_attack_stake} ensures that the attack cost of validator~$\validator_i$ must be equal to~$\sum_{j=1}^m \milpAttackStakeVariable{i}{j}$ and~\eqreft{equation:milp_validator_cost_lower_bound_validator_stake} is trivially satisfied.
For this to hold, we must have~$\milpLargeNumberCost \geq \stake{\validator_i}$ and~$\milpLargeNumberCost \geq \sum_{j=1}^m \milpAttackStakeVariable{i}{j}$ for all~$i \in \left\{ 1, \ldots, n \right\}$.


\subsubsection{Constants}


We pick the constants~$\milpLargeNumberFeasibility$ and~$\milpLargeNumberCost$ as follows:
\begin{align}
    \milpLargeNumberFeasibility
    &= \max_{j \in \left\{ 1, \ldots, m \right\}} \left\{ \attackThreshold{\service_j} \cdot \sum_{i=1}^n \allocation{\validator_i}{\service_j} \right\} , \\
    \milpLargeNumberCost
    &= \max \left\{ \max_{i \in \left\{ 1, \ldots, n \right\}} \stake{\validator_i}, \max_{i \in \left\{ 1, \ldots, n \right\}} \sum_{j=1}^m \allocation{\validator_i}{\service_j} \right\} .
\end{align}


\subsubsection{Objective}


Let~$\milpVariables$ denote the tuple of all variables we defined above:
\begin{equation}
    \milpVariables = \left( \left( \milpValidatorCostVariable{i}, \milpValidatorCostAuxiliaryVariable{i} \right)_{i=1}^n, \left( \milpAttackedServiceVariable{j} \right)_{j=1}^m, \left( \milpAttackStakeVariable{i}{j} \right)_{i=1,j=1}^{n,m} \right) .
\end{equation}
The objective is to maximize the profit of the attack, namely, the total attack prize minus the total attack cost:
\begin{equation}
    \max_{\milpVariables} \sum_{j=1}^m \attackPrize{\service_j} \cdot \milpAttackedServiceVariable{j} - \sum_{i=1}^n \milpValidatorCostVariable{i} .
\end{equation}

If the optimum~$y$ we find is greater or equal to~$0$, then the network is not secure.
And if it is less than~$0$, then the network is secure and is~$(-y)$-budget robust.


\subsubsection{MIP}


Fig.~\ref{figure:budget_only_robustness_mip} summarizes the previous paragraphs.
It presents the MIP that determines the existence of a~$\adversaryBudget$-costly allocation-divisible attack in a restaking network~${\networkState = (\allValidators, \allServices, \allStakes, \allAllocations, \allAttackThresholds, \allAttackPrizes)}$.

\begin{figure*}[t]
    \centering
    \boxed{
    \begin{minipage}{0.65\textwidth}
        \begin{alignat}{2}
            & \max_{\milpVariables} \quad && \sum_{j=1}^m \attackPrize{\service_j} \cdot \milpAttackedServiceVariable{j} - \sum_{i=1}^n \milpValidatorCostVariable{i} \\
            & \text{subject to} \quad && \sum_{j=1}^m \milpAttackedServiceVariable{j} \geq 1 ; \\
            & \forall i \in \left\{ 1, \ldots, n \right\}: \quad && 0 \leq \milpValidatorCostVariable{i} \leq \stake{\validator_i}: , \\
            & &&\milpValidatorCostAuxiliaryVariable{i} \in \{0,1\} , \\
            & && \milpValidatorCostVariable{i} \leq \sum_{j=1}^m \milpAttackStakeVariable{i}{j} , \\
            & && \milpValidatorCostVariable{i} \geq \stake{\validator_i} - \milpLargeNumberCost \cdot \milpValidatorCostAuxiliaryVariable{i} , \\
            & && \milpValidatorCostVariable{i} \geq \sum_{j=1}^m \milpAttackStakeVariable{i}{j} - \milpLargeNumberCost \cdot (1 - \milpValidatorCostAuxiliaryVariable{i}) ; \\
            & \forall j \in \left\{ 1, \ldots, m \right\}: \quad && \milpAttackedServiceVariable{j} \in \{0,1\} , \\
            & && \sum_{i=1}^n \milpAttackStakeVariable{i}{j} \geq \attackThreshold{\service_j} \cdot \sum_{i=1}^n \allocation{\validator_i}{\service_j} - \milpLargeNumberFeasibility \cdot (1 - \milpAttackedServiceVariable{j}) ; \\
            & \forall i,j \in \left\{ 1, \ldots, n \right\} \times \left\{ 1, \ldots, m \right\}: \quad && 0 \leq \milpAttackStakeVariable{i}{j} \leq \allocation{\validator_i}{\service_j} .
        \end{alignat}
    \end{minipage}
    }
    \caption{MIP for budget-only robustness.}
    \label{figure:budget_only_robustness_mip}
    \Description{MIP for budget-only robustness.}
\end{figure*}


\subsection{MIP for Budget-and-Byzantine Robustness}
\label{section:mip_appendix:budget_and_byzantine_robustness}


Given a restaking network~$\networkState = (\allValidators, \allServices, \allStakes, \allAllocations, \allAttackThresholds, \allAttackPrizes)$, where~$\allValidators = \left\{ \validator_1, \ldots, \validator_n \right\}$,~$\allServices = \left\{ \service_1, \ldots, \service_m \right\}$, and an adversary budget~$\adversaryBudget$, we formulate the problem of determining the maximum fraction~$\maxByzantineServices$ of Byzantine services such that the network is~$(\maxByzantineServices, \adversaryBudget)$-robust.
This implies that for all~$\maxByzantineServices' \leq \maxByzantineServices$, the network is also~$(\maxByzantineServices', \adversaryBudget)$-robust.


\subsubsection{Variables}


Similar to the previous MIP, we define variables for whether service~$\service_j$ is attacked~$\milpAttackedServiceVariable{j}$, for the stake validator~$\validator_i$ uses to attack service~$\service_j$~$\milpAttackStakeVariable{i}{j}$, and for the attack cost of validator~$\validator_i$~$\milpValidatorCostVariable{i}$.
We also define the auxiliary variables~$\milpValidatorCostAuxiliaryVariable{i}$ to calculate the attack cost of validators.

Unlike the previous MIP, we define new variables as follows.
For each~$j \in \left\{ 1, \ldots, m \right\}$, set~$\milpByzantineServiceVariable{j}$ to be~$1$ if service~$\service_j$ is Byzantine, and~$0$ otherwise.
For each~$i \in \left\{ 1, \ldots, n \right\}$, denote by~$\milpRemainingStakeVariable{i}$ the amount of stake of validator~$\validator_i$ that remains after Byzantine services cause slashing.
It can take any value in~$\left[ 0, \stake{\validator_i} \right]$.
For each~$i \in \left\{ 1, \ldots, n \right\}$ and~$j \in \left\{ 1, \ldots, m \right\}$, denote by~$\milpRemainingAllocationVariable{i}{j}$ the amount of stake of validator~$\validator_i$ that remains allocated to service~$\service_j$ after Byzantine services cause slashing.
It can take any value in~$\left[ 0, \allocation{\validator_i}{\service_j} \right]$.

We introduce the auxiliary variable~$\milpByzantineServiceAuxiliaryVariable$ to ensure that either all services are Byzantine, or at least one service is attacked.
For each~$i \in \left\{ 1, \ldots, n \right\}$, we introduce the auxiliary variable~$\milpRemainingStakeAuxiliaryVariable{i}$.
For each~$i \in \left\{ 1, \ldots, n \right\}$ and~$j \in \left\{ 1, \ldots, m \right\}$, we introduce the auxiliary variable~$\milpRemainingAllocationAuxiliaryVariable{i}{j}$.
These take values in~$\left\{ 0, 1 \right\}$ and will be used to calculate the remaining stake and allocation of validators.


\subsubsection{Constraints}


We begin with most of the constraints that the previous MIP has.

First, as before, we define~$\milpLargeNumberFeasibility$ as a large number used to make the constraints for having sufficient stake to attack apply only to attacked services.
Then, for an attack have sufficient stake, it must be that for each~${j \in \left\{ 1, \ldots, m \right\}}$
\begin{equation}
    \sum_{i=1}^n \milpAttackStakeVariable{i}{j}
    \geq \attackThreshold{\service_j} \cdot \sum_{i=1}^n \milpRemainingAllocationVariable{i}{j} - \milpLargeNumberFeasibility \cdot (1 - \milpAttackedServiceVariable{j}) .
\end{equation}
This time we use~$\milpRemainingAllocationVariable{i}{j}$ instead of~$\allocation{\validator_i}{\service_j}$ as the remaining allocations depend on the Byzantine services.

Similarly, the attack cost of a validator~$\validator_i$ should be equal to $\min \left( \milpRemainingStakeVariable{i}, \sum_{j=1}^m \milpAttackStakeVariable{i}{j} \right)$ instead of $\min \left( \stake{\validator_i}, \sum_{j=1}^m \milpAttackStakeVariable{i}{j} \right)$.
So, we define~$\milpLargeNumberCost$ as before, and get the following constraints for each~${i \in \left\{ 1, \ldots, n \right\}}$:
\begin{align}
    \milpValidatorCostVariable{i}
    &\leq \milpRemainingStakeVariable{i} ; \\
    \milpValidatorCostVariable{i}
    &\leq \sum_{j=1}^m \milpAttackStakeVariable{i}{j} ; \\
    \milpValidatorCostVariable{i}
    &\geq \milpRemainingStakeVariable{i} - \milpLargeNumberCost \cdot \milpRemainingStakeAuxiliaryVariable{i} ; \\
    \milpValidatorCostVariable{i}
    &\geq \sum_{j=1}^m \milpAttackStakeVariable{i}{j} - \milpLargeNumberCost \cdot (1 - \milpRemainingStakeAuxiliaryVariable{i}) .
\end{align}

Another constraint we should specify is that an attack is~$\adversaryBudget$-costly.
This was present in the previous MIP implicitly, as the objective was to maximize the profit of the attack (or minimize the loss).
The total attack prize is~$\sum_{j=1}^m \attackPrize{\service_j} \cdot \milpAttackedServiceVariable{j}$.
The total attack cost is~$\sum_{i=1}^n \milpValidatorCostVariable{i}$.
So, we have
\begin{equation}
    \sum_{j=1}^m \attackPrize{\service_j} \cdot \milpAttackedServiceVariable{j} - \sum_{i=1}^n \milpValidatorCostVariable{i} \geq \adversaryBudget .
\end{equation}

Now, we specify constraints that are new to this MIP.
First, a Byzantine service cannot be attacked.
So, for each~$j \in \left\{ 1, \ldots, m \right\}$, we have
\begin{equation}
    \milpAttackedServiceVariable{j} + \milpByzantineServiceVariable{j} \leq 1 .
\end{equation}

Next, as before, at least one service must be attacked.
However, if all services are Byzantine, there is no service to attack.
So, we define~$\milpLargeNumberByzantineServices$ to be a large number used to ensure either that at least one service is attacked, or that all services are Byzantine.
We thus have the two following constraints:
\begin{align}
    \sum_{j=1}^m \milpAttackedServiceVariable{j} &\geq 1 - \milpLargeNumberByzantineServices \cdot \milpByzantineServiceAuxiliaryVariable , \\
    \sum_{j=1}^m \milpByzantineServiceVariable{j} &\geq \left| \allServices \right| - \milpLargeNumberByzantineServices \cdot (1 - \milpByzantineServiceAuxiliaryVariable) .
\end{align}

Next, we specify constraints for the remaining stake of validators.
The remaining stake of validator~$\validator_i$ is equal to~$\max \left( 0, \stake{\validator_i} - \sum_{j=1}^m \allocation{\validator_i}{\service_j} \cdot \milpByzantineServiceVariable{j} \right)$.
Denote~$\milpLargeNumberRemainingStake$ as a large number used to calculate the remaining stake of validators.
We thus have the following constraints for each~$i \in \left\{ 1, \ldots, n \right\}$:
\begin{align}
    \milpRemainingStakeVariable{i} &\geq \stake{\validator_i} - \sum_{j=1}^m \allocation{\validator_i}{\service_j} \cdot \milpByzantineServiceVariable{j} , \\
    \milpRemainingStakeVariable{i} &\geq 0 , \\
    \milpRemainingStakeVariable{i} &\leq \stake{\validator_i} - \sum_{j=1}^m \allocation{\validator_i}{\service_j} \cdot \milpByzantineServiceVariable{j} + \milpLargeNumberRemainingStake \cdot \milpRemainingStakeAuxiliaryVariable{i} , \\
    \milpRemainingStakeVariable{i} &\leq \milpLargeNumberRemainingStake \cdot (1 - \milpRemainingStakeAuxiliaryVariable{i}) .
\end{align}

Lastly, we specify constraints for the remaining allocation of validators.
For each~$i \in \left\{ 1, \ldots, n \right\}$ and~$j \in \left\{ 1, \ldots, m \right\}$, the remaining allocation of validator~$\validator_i$ to service~$\service_j$ is equal to~$\min \left( \allocation{\validator_i}{\service_j}, \milpRemainingStakeVariable{i} \right)$.
Denote~$\milpLargeNumberRemainingAllocation$ as a large number used to calculate the remaining allocation of validators.
We thus have the following constraints for each~$i \in \left\{ 1, \ldots, n \right\}$ and~$j \in \left\{ 1, \ldots, m \right\}$:
\begin{align}
    \milpRemainingAllocationVariable{i}{j} &\leq \allocation{\validator_i}{\service_j} , \\
    \milpRemainingAllocationVariable{i}{j} &\leq \milpRemainingStakeVariable{i} , \\
    \milpRemainingAllocationVariable{i}{j} &\geq \allocation{\validator_i}{\service_j} - \milpLargeNumberRemainingAllocation \cdot \milpRemainingAllocationAuxiliaryVariable{i}{j} , \\
    \milpRemainingAllocationVariable{i}{j} &\geq \milpRemainingStakeVariable{i} - \milpLargeNumberRemainingAllocation \cdot (1 - \milpRemainingAllocationAuxiliaryVariable{i}{j}) .
\end{align}


\subsubsection{Constants}


We pick the constants~$\milpLargeNumberFeasibility$,~$\milpLargeNumberCost$,~$\milpLargeNumberRemainingStake$,~$\milpLargeNumberRemainingAllocation$, and~$\milpLargeNumberByzantineServices$ as follows:
\begin{align}
    \milpLargeNumberFeasibility &= \max_{j \in \left\{ 1, \ldots, m \right\}} \left\{ \attackThreshold{\service_j} \cdot \sum_{i=1}^n \allocation{\validator_i}{\service_j} \right\} , \\
    \milpLargeNumberCost = \milpLargeNumberRemainingStake &= \max \left\{ \max_{i \in \left\{ 1, \ldots, n \right\}} \stake{\validator_i}, \max_{i \in \left\{ 1, \ldots, n \right\}} \sum_{j=1}^m \allocation{\validator_i}{\service_j} \right\} , \\
    \milpLargeNumberRemainingAllocation &= \max_{i \in \left\{ 1, \ldots, n \right\}} \stake{\validator_i} , \\
    \milpLargeNumberByzantineServices &= \left| \allServices \right| .
\end{align}


\subsubsection{Objective}


Let~$\milpVariables$ denote the concatenation of all variables we defined:
\begin{multline}
    \milpVariables = \left( \milpByzantineServiceAuxiliaryVariable, \left( \milpValidatorCostVariable{i}, \milpValidatorCostAuxiliaryVariable{i}, \milpRemainingStakeVariable{i}, \milpRemainingStakeAuxiliaryVariable{i} \right)_{i=1}^n, \left( \milpAttackedServiceVariable{j}, \milpByzantineServiceVariable{j} \right)_{j=1}^m, \right. \\
    \left. \left( \milpAttackStakeVariable{i}{j}, \milpRemainingAllocationVariable{i}{j}, \milpRemainingAllocationAuxiliaryVariable{i}{j} \right)_{i=1,j=1}^{n,m} \right) .
\end{multline}

We search for the maximum fraction of Byzantine services such that the network remains secure.
This is equivalent to searching for the minimum fraction of Byzantine services such that the network can be attacked.
We thus minimize the following objective function:
\begin{equation}
    \sum_{j=1}^m \frac{\attackPrize{\service_j}}{\attackThreshold{\service_j}} \milpByzantineServiceVariable{j} .
\end{equation}
A larger Byzantine service is more damaging than a smaller Byzantine service.
To negate this, we weight each service by the ratio of its attack prize to its attack threshold, that is the stake required to secure the service if it were the only one.


\subsubsection{MIP}


Fig.~\ref{figure:budget_and_byzantine_robustness_mip} summarizes the previous paragraphs.
It presents the MIP that, for a given restaking network~$\networkState$ and an adversary budget~$\adversaryBudget$, determines the maximum fraction of Byzantine services~$\maxByzantineServices$ such that the network is~$(\maxByzantineServices, \adversaryBudget)$-robust.
\begin{figure*}[tp]
    \centering
    \boxed{
    \begin{minipage}{0.65\textwidth}
        \begin{alignat}{2}
            & \min_{\milpVariables} \quad && \sum_{j=1}^m \frac{\attackPrize{\service_j}}{\attackThreshold{\service_j}} \milpByzantineServiceVariable{j} \\
            & \text{subject to} \quad && \sum_{j=1}^m \milpAttackedServiceVariable{j} \geq 1 - \milpLargeNumberByzantineServices \cdot \milpByzantineServiceAuxiliaryVariable , \\
            & && \sum_{j=1}^m \milpByzantineServiceVariable{j} \geq \left| \allServices \right| - \milpLargeNumberByzantineServices \cdot (1 - \milpByzantineServiceAuxiliaryVariable) , \\
            & && \sum_{j=1}^m \attackPrize{\service_j} \cdot \milpAttackedServiceVariable{j} - \sum_{i=1}^n \milpValidatorCostVariable{i} \geq \adversaryBudget ; \\
            & \forall i \in \left\{ 1, \ldots, n \right\}: \quad && 0 \leq \milpValidatorCostVariable{i} \leq \milpRemainingStakeVariable{i} , \\
            & && \milpValidatorCostAuxiliaryVariable{i} \in \{0,1\} , \\
            & && \milpValidatorCostVariable{i} \leq \sum_{j=1}^m \milpAttackStakeVariable{i}{j} , \\
            & && \milpValidatorCostVariable{i} \geq \milpRemainingStakeVariable{i} - \milpLargeNumberCost \cdot \milpValidatorCostAuxiliaryVariable{i} , \\
            & && \milpValidatorCostVariable{i} \geq \sum_{j=1}^m \milpAttackStakeVariable{i}{j} - \milpLargeNumberCost \cdot (1 - \milpValidatorCostAuxiliaryVariable{i}) ; \\
            & && 0 \leq \milpRemainingStakeVariable{i} \leq \stake{\validator_i} , \\
            & && \milpRemainingStakeAuxiliaryVariable{i} \in \{0,1\} , \\
            & && \milpRemainingStakeVariable{i} \geq \stake{\validator_i} - \sum_{j=1}^m \allocation{\validator_i}{\service_j} \cdot \milpByzantineServiceVariable{j} , \\
            & && \milpRemainingStakeVariable{i} \leq \stake{\validator_i} - \sum_{j=1}^m \allocation{\validator_i}{\service_j} \cdot \milpByzantineServiceVariable{j} + \milpLargeNumberRemainingStake \cdot \milpRemainingStakeAuxiliaryVariable{i} , \\
            & && \milpRemainingStakeVariable{i} \leq \milpLargeNumberRemainingStake \cdot (1 - \milpRemainingStakeAuxiliaryVariable{i}) ; \\
            & \forall j \in \left\{ 1, \ldots, m \right\}: \quad && \milpAttackedServiceVariable{j} \in \{0,1\} , \\
            & && \milpByzantineServiceVariable{j} \in \{0,1\} , \\
            & && \milpAttackedServiceVariable{j} + \milpByzantineServiceVariable{j} \leq 1 , \\
            & && \sum_{i=1}^n \milpAttackStakeVariable{i}{j} \geq \attackThreshold{\service_j} \cdot \sum_{i=1}^n \milpRemainingAllocationVariable{i}{j} - \milpLargeNumberFeasibility \cdot (1 - \milpAttackedServiceVariable{j}) ; \\
            & \forall i,j \in \left\{ 1, \ldots, n \right\} \times \left\{ 1, \ldots, m \right\}: \quad && 0 \leq \milpAttackStakeVariable{i}{j} \leq \milpRemainingAllocationVariable{i}{j} , \\
            & && \milpRemainingAllocationAuxiliaryVariable{i}{j} \in \{0,1\} , \\
            & && \milpRemainingAllocationVariable{i}{j} \leq \allocation{\validator_i}{\service_j} , \\
            & && \milpRemainingAllocationVariable{i}{j} \leq \milpRemainingStakeVariable{i} , \\
            & && \milpRemainingAllocationVariable{i}{j} \geq \allocation{\validator_i}{\service_j} - \milpLargeNumberRemainingAllocation \cdot \milpRemainingAllocationAuxiliaryVariable{i}{j} , \\
            & & & \milpRemainingAllocationVariable{i}{j} \geq \milpRemainingStakeVariable{i} - \milpLargeNumberRemainingAllocation \cdot (1 - \milpRemainingAllocationAuxiliaryVariable{i}{j}) .
        \end{alignat}
    \end{minipage}
    }
    \caption{MIP for budget-and-byzantine robustness.}
    \label{figure:budget_and_byzantine_robustness_mip}
    \Description{MIP for budget-and-byzantine robustness.}
\end{figure*}


\subsection{Solving the MIPs}
\label{section:mip_appendix:solving_the_mips}

We solve the MIPs in Python~\anon{\cite{barzur2025code}}, dynamically generating any instance using NumPy~\cite{harris2020array} and then calling SciPy~\cite{2020SciPy-NMeth} to numerically solve the instance.
Under the hood, SciPy uses the dual revised simplex method~\cite{huangfu2018parallelizing} implemented in the library HiGHS~\cite{hall2023highs}.

We solve the MIPs with a precision of~$10^{-6}$, meaning that the solution we find is feasible, and the objective value is within~$10^{-6}$ of the true optimum.

For running time optimization, instead of solving the complete Robustness MIP for symmetric networks, we iterate over all possible fractions of Byzantine services, and for each fraction, simulate the network state caused by the Byzantine services and solve the Budget Robustness MIP.
This is only possible for symmetric networks, for which we can choose any services to be Byzantine according to the desired fraction as all would lead to the same network state.
But for asymmetric networks, different subsets of Byzantine services may lead to different network states, so we must use the complete Robustness MIP.


\section{Proofs Deferred from Section~\ref{section:incentives_for_a_target_restaking_degree}}
\label{appendix:proofs_from_section_incentives_for_a_target_restaking_degree}


\begin{theorem}[Theorem~\ref{theorem:incentives_for_a_target_restaking_degree:nash_equilibrium} restated]
    \label{theorem:incentives_for_a_target_restaking_degree:nash_equilibrium_appendix}
    Assume that for each service~$\service \in \allServices$, $\serviceReward{\service} > 0$ and~$\targetRestakingDegree \cdot \frac{\serviceReward{\service}}{\sum_{\service' \in \allServices} \serviceReward{\service'}} \leq 1$.
    Then, the strategy profile
    \begin{equation}
        \label{equation:incentives_for_a_target_restaking_degree:nash_allocations}
        \nashAllocation{\validator}{\service} = \targetRestakingDegree \cdot \frac{\serviceReward{\service}}{\sum_{\service' \in \allServices} \serviceReward{\service'}} \cdot \stake{\validator}
    \end{equation}
    is a Nash equilibrium, and it results in a restaking degree of~$\targetRestakingDegree$.
\end{theorem}
\begin{proof}
    We first show that in this strategy profile, all validators have a restaking degree of~$\targetRestakingDegree$.
    \begin{multline}
        \label{equation:incentives_for_a_target_restaking_degree:nash_restaking_degree}
        \restakingDegree{\validator}
        \underset{\eqref{equation:model:restaking_degree}}{=} \frac{\sum_{\service \in \allServices} \nashAllocation{\validator}{\service}}{\stake{\validator}}
        \underset{\eqref{equation:incentives_for_a_target_restaking_degree:nash_allocations}}{=} \frac{\sum_{\service \in \allServices} \targetRestakingDegree \cdot \frac{\serviceReward{\service}}{\sum_{\service' \in \allServices} \serviceReward{\service'}} \cdot \stake{\validator}}{\stake{\validator}} \\
        = \targetRestakingDegree \cdot \frac{\sum_{\service \in \allServices} \serviceReward{\service}}{\sum_{\service' \in \allServices} \serviceReward{\service'}}
        = \targetRestakingDegree .
    \end{multline}

    Next, we show that this strategy profile is a Nash equilibrium.
    To do so, we use~$f \cup g$ to denote a piecewise combination of~$f$ and~$g$. Formally, Let $f: A \to C$ and $g: B \to C$ such that~$A \cap B = \emptyset$. Then $f \cup g: A \cup B \to C$ is defined as $(f \cup g)(x) = f(x)$ for $x \in A$ and $(f \cup g)(x) = g(x)$ for $x \in B$.

    Fix a validator~$\validator$, and consider the strategy profile~$\allNashAllocations_{-\validator}$ of all validators except~$\validator$, namely,~$\allNashAllocations_{-\validator} = \restr{\allNashAllocations}{\left(\allValidators \setminus \left\{ \validator \right\}\right) \times \allServices}$.
    We need to show that for validator~$\validator$ it holds for any possible strategy~$\allAllocations_\validator: \left\{ \validator \right\} \times \allServices \to \positiveRealNumbers$ that
    \begin{equation}
        \label{equation:incentives_for_a_target_restaking_degree:nash_equilibrium}
        \validatorUtility{\validator}{\allNashAllocations} \geq \validatorUtility{\validator}{\allAllocations_\validator \cup \allNashAllocations_{-\validator}} .
    \end{equation}
    To do so, we develop the term on the right-hand side.

    But first, let~$\allServices = \{\service_1, ... \service_n\}$, and for all~${i \in [n]}$ denote~${\allocationStrategy{i} = \allAllocations_\validator(\validator, \service_i)}$.
    
    Now, let's develop the term on the right-hand side of~\eqreft{equation:incentives_for_a_target_restaking_degree:nash_equilibrium}.
    Consider 2 cases.
    First, if~$\sum_{i=1}^n \allocationStrategy{i} > \targetRestakingDegree \cdot \stake{\validator}$, then~$\restakingDegree{\validator} > \targetRestakingDegree$ and~$\validatorUtility{\validator}{\allAllocations_\validator \cup \allNashAllocations_{-\validator}} = 0$~(\eqreft{equation:incentives_for_a_target_restaking_degree:validator_utility}), and~\eqreft{equation:incentives_for_a_target_restaking_degree:nash_equilibrium} holds.

    Second, assume that~$\sum_{i=1}^n \allocationStrategy{i} \leq \targetRestakingDegree \cdot \stake{\validator}$, meaning that
    \begin{equation}
        \label{equation:incentives_for_a_target_restaking_degree:restaking_degree}
        \restakingDegree{\validator} \leq \targetRestakingDegree .
    \end{equation}
    Let
    \begin{equation}
        \label{equation:incentives_for_a_target_restaking_degree:all_allocations}
        \allAllocations = \allAllocations_\validator \cup \allNashAllocations_{-\validator} .
    \end{equation}
    We now get that
    \begin{multline}
        \validatorUtility{\validator}{\allAllocations_\validator \cup \allNashAllocations_{-\validator}}
        \underset{\eqref{equation:incentives_for_a_target_restaking_degree:all_allocations}}{=} \validatorUtility{\validator}{\allAllocations} \\
        \underset{\eqref{equation:incentives_for_a_target_restaking_degree:validator_utility}}{=} \begin{cases}
            \sum_{i=1}^n \frac{\allocation{\validator}{\service_i}}{\sum_{\validator' \in \allValidators} \allocation{\validator'}{\service_i}} \cdot \serviceReward{\service_i} & \text{if } \restakingDegree{\validator} \leq \targetRestakingDegree , \\
            0 & \text{otherwise} ;
        \end{cases} \\
        \underset{\eqref{equation:incentives_for_a_target_restaking_degree:restaking_degree}}{=} \sum_{i=1}^n \frac{\allocation{\validator}{\service_i}}{\sum_{\validator' \in \allValidators} \allocation{\validator'}{\service_i}} \cdot \serviceReward{\service_i} \\
        = \sum_{i=1}^n \frac{\allocation{\validator}{\service_i}}{\allocation{\validator}{\service_i} + \sum_{\validator' \in \allValidators \setminus \{ \validator \}} \allocation{\validator'}{\service_i}} \cdot \serviceReward{\service_i} \\
        \underset{\eqref{equation:incentives_for_a_target_restaking_degree:all_allocations}}{=} \sum_{i=1}^n \frac{\allocationStrategy{i}}{\allocationStrategy{i} + \sum_{\validator' \in \allValidators \setminus \{ \validator \}} \nashAllocation{\validator'}{\service_i}} \cdot \serviceReward{\service_i} \\
        = \sum_{i=1}^n \frac{1}{1 + \frac{1}{\allocationStrategy{i}} \cdot \sum_{\validator' \in \allValidators \setminus \{ \validator \}} \nashAllocation{\validator'}{\service_i}} \cdot \serviceReward{\service_i} .
    \end{multline}

    For simplicity, let
    \begin{equation}
        \label{equation:incentives_for_a_target_restaking_degree:nash_constant}
        \nashConstant{i} = \sum_{\validator' \in \allValidators \setminus \{ \validator \}} \nashAllocation{\validator'}{\service_i} ;
    \end{equation}
    these are non-negative constants with respect to the strategy of~$\validator$.
    We can then rewrite the utility of~$\validator$ as
    \begin{equation}
        \label{equation:incentives_for_a_target_restaking_degree:validator_utility_2}
        \validatorUtility{\validator}{\allAllocations_\validator \cup \allNashAllocations_{-\validator}}
        = \sum_{i=1}^n \frac{1}{1 + \frac{1}{\allocationStrategy{i}} \cdot \nashConstant{i}} \cdot \serviceReward{\service_i}
        = \sum_{i=1}^n \left( 1 + \frac{\nashConstant{i}}{\allocationStrategy{i}} \right)^{-1} \cdot \serviceReward{\service_i} .
    \end{equation}
    
    Now, we show that this utility is maximized when~$\allocationStrategy{i} = \nashAllocation{\validator}{\service_i}$ for all~$i \in [n]$.
    The term~$\validatorUtility{\validator}{\allAllocations_\validator \cup \allNashAllocations_{-\validator}}$ is a continuous function of the variables~$\left\{ \allocationStrategy{i} \right\}_{i=1}^n$ in a compact set defined by the inequalities:
    \begin{align}
        \forall i \in [n]; \quad \allocationStrategy{i} &\geq 0 , \text{and}\\
        \sum_{i=1}^n \allocationStrategy{i} &\leq \targetRestakingDegree \cdot \stake{\validator} .
    \end{align}
    The discontinuities where~$\allocationStrategy{i} = 0$ can be removed by substituting the result of~$\left( 1 + \frac{\nashConstant{i}}{\allocationStrategy{i}} \right)^{-1}$ to~$0$ at these points since this is the limit when~$\allocationStrategy{i}$ approaches~$0$.
    The function is continuous on a compact set, and thus attains a maximum.
    We now show that the maximum is attained when~$\allocationStrategy{i} = \nashAllocation{\validator}{\service_i}$.

    First, consider the case where~$\sum_{i=1}^n \allocationStrategy{i} < \targetRestakingDegree \cdot \stake{\validator}$.
    It must be that there is some~$i$ such that~$\allocationStrategy{i} < \nashAllocation{\validator}{\service_i}$, or otherwise the restaking degree of the validator would be at least~$\targetRestakingDegree$.
    This also implies that~$\allocationStrategy{i} < \stake{\validator}$.
    Without loss of generality, let~$i=n$.

    Pick~$\varepsilon$ such that~$\varepsilon < \stake{\validator} - \allocationStrategy{n}$.
    Consider an alternative strategy profile~$\allAllocations'_\validator$, and denote its value for all~$i \in [n]$ as~$\allocationStrategyPrime{i}$, which we choose to be
    \begin{equation}
        \label{equation:incentives_for_a_target_restaking_degree:all_allocations_prime}
        \allocationStrategyPrime{i}
        = \begin{cases}
            \allocationStrategy{i} + \varepsilon & \text{if } i = n , \\
            \allocationStrategy{i} & \text{otherwise} ;
        \end{cases}
    \end{equation}

    This profile is well-defined due to our choice of~$\varepsilon$, and it gives a strictly higher utility to~$\validator$ than~$\allAllocations_\validator$:
    \begin{multline}
        \validatorUtility{\validator}{\allAllocations'_\validator \cup \allNashAllocations_{-\validator}}
        \underset{\eqref{equation:incentives_for_a_target_restaking_degree:validator_utility_2}}{=} \sum_{i=1}^n \left( 1 + \frac{\nashConstant{i}}{\allocationStrategyPrime{i}} \right)^{-1} \cdot \serviceReward{\service_i} \\
        = \left( 1 + \frac{\nashConstant{n}}{\allocationStrategyPrime{n}} \right)^{-1} \cdot \serviceReward{\service_n} + \sum_{i=1}^{n-1} \left( 1 + \frac{\nashConstant{i}}{\allocationStrategyPrime{i}} \right)^{-1} \cdot \serviceReward{\service_i} \\
        \underset{\eqref{equation:incentives_for_a_target_restaking_degree:all_allocations_prime}}{=} \left( 1 + \frac{\nashConstant{n}}{\allocationStrategy{n} + \varepsilon} \right)^{-1} \cdot \serviceReward{\service_n} + \sum_{i=1}^{n-1} \left( 1 + \frac{\nashConstant{i}}{\allocationStrategy{i}} \right)^{-1} \cdot \serviceReward{\service_i} \\
        > \left( 1 + \frac{\nashConstant{n}}{\allocationStrategy{n}} \right)^{-1} \cdot \serviceReward{\service_n} + \sum_{i=1}^{n-1} \left( 1 + \frac{\nashConstant{i}}{\allocationStrategy{i}} \right)^{-1} \cdot \serviceReward{\service_i} \\
        = \sum_{i=1}^n \left( 1 + \frac{\nashConstant{i}}{\allocationStrategy{i}} \right)^{-1} \cdot \serviceReward{\service_i}
        \underset{\eqref{equation:incentives_for_a_target_restaking_degree:validator_utility_2}}{=} \validatorUtility{\validator}{\allAllocations_\validator \cup \allNashAllocations_{-\validator}} .
    \end{multline}

    So, the strategy we considered~$\allAllocations_\validator$ is not a maximum.
    We now restrict our search for the maximum to the set of strategy profiles where
    \begin{equation}
        \label{equation:incentives_for_a_target_restaking_degree:sum_allocations}
        \sum_{i=1}^n \allocationStrategy{i} = \targetRestakingDegree \cdot \stake{\validator} .
    \end{equation}
    
    By isolating the service~$\service_n$ in~\eqreft{equation:incentives_for_a_target_restaking_degree:sum_allocations}, we get that
    \begin{equation}
        \label{equation:incentives_for_a_target_restaking_degree:allocation_service}
        \allocationStrategy{n} = \targetRestakingDegree \cdot \stake{\validator} - \sum_{i=1}^{n-1} \allocationStrategy{i} .
    \end{equation}

    Now, let~$\utilityFunction$ be the utility of validator~$\validator$ as a function of~$\left\{ \allocationStrategy{i} \right\}_{i=1}^{n-1}$.
    Formally, we get that
    \begin{equation}
        \utilityFunction(\allocationStrategy{1}, \dots, \allocationStrategy{n-1}) = \validatorUtility{\validator}{\allAllocations_\validator \cup \allNashAllocations_{-\validator}}
    \end{equation}
    with the constraints
    \begin{align}
        \forall i \in [n-1] , \quad \allocationStrategy{i} \geq 0 ; \text{and} \\
        \sum_{i=1}^{n-1} \allocationStrategy{i} \leq \targetRestakingDegree \cdot \stake{\validator} .
    \end{align}
    We now show that this function is concave and then find its maximum.
    We start by developing the right-hand side.
    \begin{multline}
        \utilityFunction(\allocationStrategy{1}, \dots, \allocationStrategy{n-1}) \\
        = \validatorUtility{\validator}{\allAllocations_\validator \cup \allNashAllocations_{-\validator}} 
        \underset{\eqref{equation:incentives_for_a_target_restaking_degree:validator_utility_2}}{=} \sum_{i=1}^n \left( 1 + \frac{\nashConstant{i}}{\allocationStrategy{i}} \right)^{-1} \cdot \serviceReward{\service_i} \\
        = \left( 1 + \frac{\nashConstant{n}}{\allocationStrategy{n}} \right)^{-1} \cdot \serviceReward{\service_n} + \sum_{i=1}^{n-1} \left( 1 + \frac{\nashConstant{i}}{\allocationStrategy{i}} \right)^{-1} \cdot \serviceReward{\service_i} \\
        \underset{\eqref{equation:incentives_for_a_target_restaking_degree:allocation_service}}{=} \left( 1 + \frac{\nashConstant{n}}{\targetRestakingDegree \cdot \stake{\validator} - \sum_{i=1}^{n-1} \allocationStrategy{i}} \right)^{-1} \cdot \serviceReward{\service_n} + \\
        \sum_{i=1}^{n-1} \left( 1 + \frac{\nashConstant{i}}{\allocationStrategy{i}} \right)^{-1} \cdot \serviceReward{\service_i} .
    \end{multline}
    Now, let~$g_{c}(x) = \left( 1 + \frac{c}{x} \right)^{-1}$, with the discontinuity at~$x=0$ defined again as~$g_{c}(0) = 0$.
    Notice that~$g_{c}(x)$ is a concave function for all~$x \geq 0$ and~$c \geq 0$:
    \begin{align}
        \label{equation:incentives_for_a_target_restaking_degree:g_c_derivative}
        \frac{dg_{c}(x)}{dx}
        &= \frac{c}{x^2} \cdot \left( 1 + \frac{c}{x} \right)^{-2}
        = \frac{c}{\left( c + x \right)^2} ; \\
        \frac{d^2g_{c}(x)}{dx^2}
        &= -\frac{c}{\left( c + x \right)^3} \leq 0 .
    \end{align}
    We can now rewrite the utility function as
    \begin{multline}
        \utilityFunction(\allocationStrategy{1}, \dots, \allocationStrategy{n-1}) \\
        = g_{\nashConstant{n}} \left( \targetRestakingDegree \cdot \stake{\validator} - \sum_{i=1}^{n-1} \allocationStrategy{i}\right) \cdot \serviceReward{\service_n} + \sum_{i=1}^{n-1} g_{\nashConstant{i}} \left( \allocationStrategy{i} \right) \cdot \serviceReward{\service_i} .
    \end{multline}
    Since an affine transformation of a concave function is concave and a sum of concave functions is concave, the utility function~$\utilityFunction$ is concave.

    We can then calculate the partial derivatives of~$\utilityFunction$ with respect to~$\left\{ \allocationStrategy{i} \right\}_{i=1}^{n-1}$ using the chain rule and~\eqreft{equation:incentives_for_a_target_restaking_degree:g_c_derivative}.
    For all~$j \in [n-1]$, the first derivative of~$\utilityFunction$ with respect to~$\allocationStrategy{i}$ is
    \begin{multline}
        \frac{\partial \utilityFunction}{\partial \allocationStrategy{j}}        
        = -\nashConstant{n} \cdot \serviceReward{\service_n} \cdot \left( \nashConstant{n} + \targetRestakingDegree \cdot \stake{\validator} - \sum_{i=1}^{n-1} \allocationStrategy{i} \right)^{-2} \\
        + \nashConstant{j}  \cdot \serviceReward{\service_j} \cdot \left( \nashConstant{j} + \allocationStrategy{j} \right)^{-2} .
    \end{multline}
    We search for critical points of~$\utilityFunction$ by solving the system of equations
    \begin{equation}
        \forall j \in [n-1] , \quad \frac{\partial \utilityFunction}{\partial \allocationStrategy{j}}
        = 0 .
    \end{equation}
    
    It is time to substitute~$\nashConstant{i}$ back.
    Before we develop them.
    For each~$i \in [n]$, we have
    \begin{multline}
        \label{equation:incentives_for_a_target_restaking_degree:nash_constant_developed}
        \nashConstant{i}
        \underset{\eqref{equation:incentives_for_a_target_restaking_degree:nash_constant}}{=} \sum_{\validator' \in \allValidators \setminus \{ \validator \}} \nashAllocation{\validator'}{\service_i}
        \underset{\eqref{equation:incentives_for_a_target_restaking_degree:nash_allocations}}{=} \sum_{\validator' \in \allValidators \setminus \{ \validator \}} \targetRestakingDegree \cdot \frac{\serviceReward{\service_i}}{\sum_{j=1}^n \serviceReward{\service_j}} \cdot \stake{\validator'} \\
        = \serviceReward{\service_i} \cdot \frac{\sum_{\validator' \in \allValidators \setminus \{ \validator \}} \stake{\validator'}}{\sum_{j=1}^n \serviceReward{\service_j}} \cdot \targetRestakingDegree
        = \serviceReward{\service_i} \cdot k,
    \end{multline}
    where~$k$ is a constant:
    \begin{equation}
        \label{equation:incentives_for_a_target_restaking_degree:nash_constant_k}
        k = \frac{\sum_{\validator' \in \allValidators \setminus \{ \validator \}} \stake{\validator'}}{\sum_{j=1}^n \serviceReward{\service_j}} \cdot \targetRestakingDegree .
    \end{equation}

    Developing the equation for each~$j \in [n-1]$, we get
    \begin{align}
        \nashConstant{j}  \cdot \serviceReward{\service_j} \cdot \left( \nashConstant{j} + \allocationStrategy{j} \right)^{-2}
        &= \nashConstant{n} \cdot \serviceReward{\service_n} \cdot \left( \nashConstant{n} + \targetRestakingDegree \cdot \stake{\validator} - \sum_{i=1}^{n-1} \allocationStrategy{i} \right)^{-2} ; \\
        \frac{\nashConstant{j} \serviceReward{\service_j}}{\left( \nashConstant{j} + \allocationStrategy{j} \right)^2}
        &= \frac{\nashConstant{n} \serviceReward{\service_n}}{\left( \nashConstant{n} + \targetRestakingDegree \cdot \stake{\validator} - \sum_{i=1}^{n-1} \allocationStrategy{i} \right)^2} ; \\
        \frac{k \serviceReward{\service_j}^2}{\left( k \serviceReward{\service_j} + \allocationStrategy{j} \right)^2}
        &= \frac{ k \serviceReward{\service_n}^2}{\left( k \serviceReward{\service_n} + \targetRestakingDegree \cdot \stake{\validator} - \sum_{i=1}^{n-1} \allocationStrategy{i} \right)^2} .
    \end{align}
    Since all terms are positive, we can take the square root of both sides and then take the inverse:
    \begin{align}
        \frac{\serviceReward{\service_j}}{ k \serviceReward{\service_j} + \allocationStrategy{j}}
        &= \frac{ \serviceReward{\service_n}}{ k \serviceReward{\service_n} + \targetRestakingDegree \cdot \stake{\validator} - \sum_{i=1}^{n-1} \allocationStrategy{i}} ; \\
        \frac{ k \serviceReward{\service_j} + \allocationStrategy{j}}{\serviceReward{\service_j}}
        &= \frac{ k \serviceReward{\service_n} + \targetRestakingDegree \cdot \stake{\validator} - \sum_{i=1}^{n-1} \allocationStrategy{i}}{\serviceReward{\service_n}} ; \\
        \frac{\allocationStrategy{j}}{\serviceReward{\service_j}}
        &= \frac{\targetRestakingDegree \cdot \stake{\validator} - \sum_{i=1}^{n-1} \allocationStrategy{i}}{\serviceReward{\service_n}} ; \\
        \label{equation:incentives_for_a_target_restaking_degree:allocation_j}
        \serviceReward{\service_n} \allocationStrategy{j}
        &= \serviceReward{\service_j} \cdot \targetRestakingDegree \cdot \stake{\validator} - \serviceReward{\service_j} \sum_{i=1}^{n-1} \allocationStrategy{i} .
    \end{align}
    Summing over all~$j \in [n-1]$, we get
    \begin{equation}
        \sum_{j=1}^{n-1} \serviceReward{\service_n} \allocationStrategy{j}
        = \sum_{j=1}^{n-1} \serviceReward{\service_j} \cdot \targetRestakingDegree \cdot \stake{\validator} - \sum_{j=1}^{n-1} \serviceReward{\service_j} \sum_{i=1}^{n-1} \allocationStrategy{i} .
    \end{equation}
    Switching sides and developing further, we get
    \begin{align}
        \serviceReward{\service_n} \sum_{j=1}^{n-1} \allocationStrategy{j} + \left( \sum_{j=1}^{n-1} \serviceReward{\service_j} \right) \sum_{j=1}^{n-1} \allocationStrategy{j}
        &= \targetRestakingDegree \cdot \stake{\validator} \sum_{j=1}^{n-1} \serviceReward{\service_j} ; \\
        \left( \sum_{j=1}^n \serviceReward{\service_j} \right) \sum_{j=1}^{n-1} \allocationStrategy{j}
        &= \targetRestakingDegree \cdot \stake{\validator} \sum_{j=1}^{n-1} \serviceReward{\service_j} ; \\
        \label{equation:incentives_for_a_target_restaking_degree:allocation_j_sum}
        \sum_{j=1}^{n-1} \allocationStrategy{j}
        &= \targetRestakingDegree \cdot \stake{\validator} \frac{\sum_{j=1}^{n-1} \serviceReward{\service_j}}{\sum_{j=1}^n \serviceReward{\service_j}} .
    \end{align}
    Plugging this back into~\eqreft{equation:incentives_for_a_target_restaking_degree:allocation_j}, we get
    \begin{multline}
        \label{equation:incentives_for_a_target_restaking_degree:allocation_j_plugged}
        \serviceReward{\service_n} \allocationStrategy{j}
        \underset{\eqref{equation:incentives_for_a_target_restaking_degree:allocation_j}}{=} \serviceReward{\service_j} \cdot \targetRestakingDegree \cdot \stake{\validator} - \serviceReward{\service_j} \sum_{i=1}^{n-1} \allocationStrategy{i} \\
        \underset{\eqref{equation:incentives_for_a_target_restaking_degree:allocation_j_sum}}{=} \serviceReward{\service_j} \cdot \targetRestakingDegree \cdot \stake{\validator} - \serviceReward{\service_j} \cdot \targetRestakingDegree \cdot \stake{\validator} \frac{\sum_{i=1}^{n-1} \serviceReward{\service_i}}{\sum_{i=1}^n \serviceReward{\service_i}} \\
        = \serviceReward{\service_j} \cdot \targetRestakingDegree \cdot \stake{\validator} \left( 1 - \frac{\sum_{i=1}^{n-1} \serviceReward{\service_i}}{\sum_{i=1}^n \serviceReward{\service_i}} \right) \\
        = \serviceReward{\service_j} \cdot \targetRestakingDegree \cdot \stake{\validator} \left( \frac{\sum_{i=1}^n \serviceReward{\service_i} - \sum_{i=1}^{n-1} \serviceReward{\service_i}}{\sum_{i=1}^n \serviceReward{\service_i}} \right) \\
        = \serviceReward{\service_j} \cdot \targetRestakingDegree \cdot \stake{\validator} \left( \frac{\serviceReward{\service_n}}{\sum_{i=1}^n \serviceReward{\service_i}} \right) .
    \end{multline}
    Overall, we get for each~$j \in [n-1]$
    \begin{equation}
        \allocationStrategy{j}
        = \targetRestakingDegree \cdot \stake{\validator} \left( \frac{\serviceReward{\service_j}}{\sum_{i=1}^n \serviceReward{\service_i}} \right)
        \underset{\eqref{equation:incentives_for_a_target_restaking_degree:nash_allocations}}{=} \nashAllocation{\validator}{\service_j} .
    \end{equation}
    Therefore, we find a single critical point of~$\utilityFunction$ within the feasible region.
    Since the~$\utilityFunction$ is concave, this critical point is a global maximum.

    For~$j = n$, we get
    \begin{multline}
        \allocationStrategy{n}
        \underset{\eqref{equation:incentives_for_a_target_restaking_degree:allocation_service}}{=} \targetRestakingDegree \cdot \stake{\validator} - \sum_{i=1}^{n-1} \allocationStrategy{i}
        \underset{\eqref{equation:incentives_for_a_target_restaking_degree:allocation_j_sum}}{=} \targetRestakingDegree \cdot \stake{\validator} - \targetRestakingDegree \cdot \stake{\validator} \frac{\sum_{j=1}^{n-1} \serviceReward{\service_j}}{\sum_{j=1}^n \serviceReward{\service_j}} \\
        = \targetRestakingDegree \cdot \stake{\validator} \left( 1 - \frac{\sum_{j=1}^{n-1} \serviceReward{\service_j}}{\sum_{j=1}^n \serviceReward{\service_j}} \right) \\
        = \targetRestakingDegree \cdot \stake{\validator} \frac{\sum_{j=1}^n \serviceReward{\service_j} - \sum_{j=1}^{n-1} \serviceReward{\service_j}}{\sum_{j=1}^n \serviceReward{\service_j}}
        = \targetRestakingDegree \cdot \stake{\validator} \frac{\serviceReward{\service_n}}{\sum_{j=1}^n \serviceReward{\service_j}}
    \end{multline}
    Hence the optimal strategy~$\allAllocations_\validator$ we find is precisely the strategy of~$\validator$ in the strategy profile~$\allNashAllocations$.
\end{proof}

\end{document}